\documentclass[journal]{IEEEtran}
\usepackage{}
\usepackage{amsfonts}
\usepackage{savesym}
\usepackage{amsmath}
\usepackage{cases}
\savesymbol{iint}
\restoresymbol{TXF}{iint}
\usepackage{amssymb}
\allowdisplaybreaks[4]
\usepackage{booktabs}
\usepackage{latexsym}
\usepackage{psfrag}
\usepackage{graphicx,subfigure}

\usepackage{hyperref}
\usepackage{color}
\usepackage{diagbox}
\usepackage{amsfonts,mathrsfs,bbm,amsmath,stmaryrd,amssymb,pifont}%
\usepackage{amscd,graphicx,array,dsfont,texdraw,tikz}%
\usetikzlibrary{arrows,decorations.markings,shapes,shadows,fadings}
\usepackage{cite}
\usepackage{array}
\usepackage{epstopdf}
\usepackage[linesnumbered,ruled,vlined]{algorithm2e}
\usepackage{bm}

\definecolor{blue}{rgb}{0,0,0}
\newtheorem{theorem}{Theorem}
\newtheorem{lemma}{Lemma}
\newtheorem{remark}{Remark}
\newtheorem{assumption}{Assumption}
\newtheorem{definition}{Definition}
\newtheorem{corollary}{Corollary}

\begin{document}

\title{Distributed Primal-Dual Algorithm for Constrained  Multi-agent Reinforcement Learning under Coupled Policies}

\author{Pengcheng Dai, He Wang, Dongming Wang, and Wenwu Yu~\IEEEmembership{Senior Member,~IEEE}
\thanks{This work was supported in part by the National Key Research and Development Program of China under Grant 2022ZD0120002 and in part by the National Natural Science Foundation of China under Grant 62233004.
\emph{(Corresponding author: Wenwu Yu and He Wang.)}
}
\thanks{Pengcheng Dai is with the Engineering Systems and Design Pillar, Singapore University of Technology and Design, Singapore 487372 (e-mail: Jldaipc@163.com).}
\thanks{He Wang is with the School of Mathematics, Southeast University, Nanjing 210096, China (e-mail: wanghe91@seu.edu.cn).}
\thanks{Dongming Wang is with the Department of Electrical and Computer Engineering, University of California, Riverside, CA 92521, USA. (e-mail: wdong025@ucr.edu).}
\thanks{Wenwu Yu is with the Frontiers Science Center for Mobile Information Communication and Security, School of Mathematics, Southeast University, Nanjing 210096, China, and also with Purple Mountain Laboratories, Nanjing 211102, China (e-mail: wwyu@seu.edu.cn).}
}

\markboth{Manuscript For
 Review}
{Manuscript For Review}

\maketitle

\begin{abstract}
This paper investigates constrained multi-agent reinforcement learning (CMARL) in coupled environments,
where agents collaboratively maximize the sum of local objectives while satisfying
individual safety constraints. Existing studies face two limitations: (1) most rely on independent policies that fail to capture complex interactions in coupled environments;
and (2) {\color{blue}agents require the global Lagrange multipliers, which are sensitive learned variables whose global sharing risks exposing private agent-specific information.}
To overcome these issues, we propose a framework where agents adopt coupled policies that depend on both the states and policy parameters of their $\kappa_p$-hop neighbors, where $\kappa_p>0$ denotes the coupling distance, and develop a distributed and scalable primal-dual (DSPD) algorithm wherein each agent accesses only information within a prescribed local neighborhood. {\color{blue}In the proposed algorithm, agents
exchange sensitive parameters only with immediate neighbors over a separate time-varying network, while maintaining local estimates to execute the coupled policy.}
We establish that the proposed algorithm achieves $\epsilon$-policy stationary convergence with approximation error $\mathcal{O}(\gamma^{\frac{\kappa+1}{\kappa_p}})$, {\color{blue}where $\kappa>0$ is the truncated distance and $\gamma\in(0,1)$ is discount factor}.
{\color{blue}Simulations on a wireless access-control network demonstrate that the proposed algorithm outperforms existing state-of-the-art algorithms, validating its effectiveness.}
\end{abstract}


\IEEEpeerreviewmaketitle

\section{Introduction}

Reinforcement learning (RL)~\cite{Bhandari2024,Mei2020} has achieved considerable
success across a broad spectrum of decision-making applications.
However, many real-world systems impose stringent safety requirements, wherein agent must maximize the objective performance while ensuring compliance with safety constraints.
{\color{blue}For example, in wireless access networks~\cite{YingNIPS2024},
transmitting agents must collaboratively maximize packet transmission throughput subject to individual transmission budget constraints.
These safety-critical applications have spurred extensive research on} constrained reinforcement learning
(CRL)~\cite{PaternainNIPS2019,PaternainTAC2022,FullanaTAC2024,ChenTAC2024,GuTPAMI2024}.
For single-agent CRL, a widely adopted approach is provided by Lagrangian-based
algorithms~\cite{DingICAIS2021,YingICAIS2022,YingAAAI2023}, where the primal variable
is updated through policy gradient ascent and the dual variable is adjusted using
projected subgradient descent.
\par
{\color{blue}Extending CRL to multi-agent environments gives rise to {\color{blue}constrained multi-agent reinforcement learning (CMARL)}, which faces two fundamental challenges: i) the curse of dimensionality arising from the exponentially growing joint state-action space, and ii) achieving cooperative optimization while satisfying individual safety constraints.
We note that cooperative CMARL constitutes an identical-interest game and therefore falls within the class of Markov potential games~(MPGs).
Recent studies on (constrained) MPGs~\cite{Das2024,Maheshwari2025,Guo2026}
have established convergence to Nash equilibria.
However, these methods require every agent to access the global state, resulting in the curse of dimensionality and poor scalability in large-scale networks.
This limitation also persists in existing decentralized
multi-agent RL algorithms~\cite{BediCDC2021,LuAAAI2021}, which likewise
assume global state access.
Although distributed gradient descent and consensus-based methods in distributed optimization~\cite{NedicTAC2015,LiuTCNS2025} share conceptual similarities with CMARL, their analyses cannot be directly extended to the present setting.
In particular, both the objective and constraint functions must be estimated from sampled trajectories, the resulting optimization problem is generally non-convex and therefore admits at best convergence to first-order stationary solutions.}
{\color{blue}To address these challenges while ensuring scalability,~\cite{QuCLDC2020} proposed a truncation-based
approach exploiting the exponential decay property, and~\cite{YingNIPS2024} extended
this to scalable CMARL using a primal-dual actor-critic algorithm with local
neighborhood information.}
\par
Despite these advancements,~\cite{YingNIPS2024} presents two key limitations.
{\color{blue}First, its parameter-independent policy architecture fails to capture
the inter-agent dependencies arising from coupled state-action dynamics in
safety-critical applications such as
intelligent transportation systems~\cite{ChuTITS2020,WangTCYB2021,DaiTII2024}.}
Second,~\cite{YingNIPS2024} requires agents to access the global Lagrange multipliers, introducing the risk of sensitive information exposure.
These limitations inspire our investigation of CMARL under coupled policies, where
each agent's policy depends on the states and parameters of its $\kappa_p$-hop
neighbors, with private information (e.g., policy parameters and Lagrange multipliers) restricted to a separate time-varying
network.
{\color{blue}This design introduces three fundamental challenges.
(1) Collaboratively optimizing under coupled policies while satisfying individual
safety constraints is significantly more challenging than the independent policy
setting;
(2) Restricting the exchange of sensitive policy parameters and Lagrange
multipliers to local neighborhoods while guaranteeing convergence is nontrivial; and
(3) Existing analytical frameworks for independent policies are inapplicable,
necessitating new theoretical techniques.}
\par
\par
To address these challenges, the main contributions are as follows.
\begin{enumerate}
\item A novel approximated coupled policy gradient is proposed, which is directly
    constructed from reward information within the $(\kappa+2\kappa_p)$-hop
    neighborhood and state-action pairs within the $2\kappa_p$-hop neighborhood, {\color{blue}where $\kappa>0$ is the truncation distance and $\kappa_p>0$ is the coupling distance,}
    eliminating the need for global state-action information.
    {\color{blue}A geometric 2-horizon
    sampling scheme is employed to provide an unbiased estimate of the approximated coupled policy gradient and the Lagrangian gradient,}
    enabling coordinated policy updates while satisfying local safety constraints.
    \item {\color{blue}To reduce the exposure of sensitive agent-specific information,} each agent maintains local estimates of other agents' policy parameters and Lagrange multipliers, while exchanging true
    variables only with immediate neighbors over time-varying learning networks.
    {\color{blue}A distributed update mechanism based on the push-sum protocol is developed for
    these local estimates}, which not only refines each agent's understanding of
    others' parameters but also promotes the generation of informative samples that
    guide coordinated policy improvement.
    {\color{blue}Building upon the above two designs, a distributed and scalable primal-dual (DSPD) algorithm is proposed, where each agent operates solely based on information from a prescribed local neighborhood without requiring access to global information.}
    \item {\color{blue}A new theoretical analysis framework is developed for the proposed DSPD algorithm under coupled policies.
        Specifically,} we establish: (i) an approximation error bound for the coupled policy gradient involving both the truncation distance $\kappa$ and the coupling distance $\kappa_p$ (see Theorem~\ref{thelemmaofapproximatedpolicygradient});
        (ii) the convergence of both policy parameter and Lagrange multiplier estimates to their true values {via a carefully ordered interleaved analysis that resolves the circular dependency between the two estimation processes} (see Theorem~\ref{Convergencetheoreminparameterestimation}); and (iii) an $\epsilon$-policy stationary convergence guarantee with approximation error $\mathcal{O}(\gamma^{\frac{\kappa+1}{\kappa_p}})$, {\color{blue}where $\gamma\in(0,1)$ is discount factor} (see Theorem~\ref{thetheoremconvergenceofpolicygradient}).
\end{enumerate}

\par
The rest of this paper is organized as follows. {\color{blue}Section~\ref{SectionIIpreliminaries}
introduces network preliminaries.
Section~\ref{SectionIII} formulates the CMARL problem and derives the coupled policy gradient.
Section~\ref{DSPDalgorithm} develops the {\color{blue}DSPD algorithm}.
Section~\ref{thesectionofanalysis} establishes the $\epsilon$-policy stationary convergence guarantee {\color{blue}of the DSPD algorithm}.
Section~\ref{SectionSimulations} validates the proposed algorithm on a wireless access-control network.
Section~\ref{SectionVConclusions} concludes the paper.}

\par

{\color{blue}\textbf{Notations}: Let $\mathbb{N}$, $\mathbb{R}$, and $\mathbb{R}_{+}$ denote the sets of natural, real, and positive real numbers, respectively.
For $N\in\mathbb{N}$, $\mathbb{R}^{N}$ ($\mathbb{R}_{+}^{N}$) denotes the $N$-dimensional Euclidean space (positive orthant), and $\mathbf{0}_{N}$ denotes the zero vector.
For $a,b\in\mathbb{R}$, $\mathrm{mod}(a,b)$ and $\lfloor a\rfloor$ denote the modulo and floor operations, respectively.
For $x\in\mathbb{R}^{N}$, $\|x\|$ denotes its Euclidean norm. $(\cdot)^{\top}$ and $\mathbb{E}[\cdot]$ denote the transpose and expectation operators, respectively.}


\section{Network preliminaries}\label{SectionIIpreliminaries}
Let $\mathcal{G}(\mathcal{N},\mathcal{E})$ represent a directed network, where
$\mathcal{N}=\{1,\cdots,N\}$ denotes the set of agents, $\mathcal{E}$ represents the set of directed
edges, {\color{blue}and $e_{ij}\in\mathcal{E}$ indicates that agent $i$ can receive information from agent $j$.}
$\mathcal{G}$ is strongly connected if there exists at
least one directed path from any agent $j\in\mathcal{N}$ to every other agent $i\in\mathcal{N}$.
Define $\mathcal{N}_i=\{j|e_{ij}\in\mathcal{E}\}$ as the
neighborhood of agent $i$ (including
$i$ itself), {\color{blue}and $\mathcal{N}_i^\kappa$ as the $\kappa$-hop neighborhood, i.e., agents within graph distance $\kappa$ from agent $i$ (including
$i$ itself).}
Let $\mathcal{N}_{i,-j}^\kappa=\mathcal{N}_i^\kappa\setminus\{j\}$
and $\mathcal{N}_{-i}^\kappa=\mathcal{N}\setminus\mathcal{N}_i^\kappa$.
\par
Let $\{\mathcal{G}_m(\mathcal{N},\mathcal{E}_m)\}_{m\geq 1}$ represent time-varying
networks, which are uniformly strongly connected if there exists an integer $D>0$ such that $(\mathcal{N},\bigcup_{m=kD+1}^{(k+1)D}\mathcal{E}_m)$ is strongly connected for any $k\in\mathbb{N}$.
For $\mathcal{G}_m$, denote
$\mathcal{N}_{i,m}^{out}=\{j|e_{ji,m}\in\mathcal{E}_m,\forall j\in\mathcal{N}\}$
as the out-neighborhood of agent $i$, and $W_m=[w_{ij,m}]_{N\times N}$ {\color{blue}as the weight matrix with}
\begin{align}\label{thedefinitionofweightmatrix}
w_{ij,m}=\left\{
\begin{array}{ll}
\frac{1}{|\mathcal{N}^{out}_{j,m}|},~\mathrm{if}~i\in\mathcal{N}^{out}_{j,m}\\
0,\;\;\;\;\;\;\;\;\;\mathrm{otherwise},
\end{array}\right.
\end{align}
where $|\mathcal{N}^{out}_{j,m}|$ is the number of agents in agent $j$'s out-neighborhood.

\section{CMARL problem under coupled policies}\label{SectionIII}
In this section, {\color{blue}we formulate the CMARL problem under coupled policies, define the $\epsilon$-FOSP, and derive the coupled policy gradient.}
\subsection{{\color{blue}Problem formulation}}
The CMARL problem under coupled policies is formally defined as $\big(\mathcal{G}^{E}(\mathcal{N},\mathcal{E}^{E}),\{\mathcal{S}_{i}\}_{i\in\mathcal{N}},\{\mathcal{A}_{i}\}_{i\in\mathcal{N}},\{\mathcal{P}_{i}\}_{i\in\mathcal{N}},$
$\{f_{i}\}_{i\in\mathcal{N}},$ $\{g_{i}\}_{i\in\mathcal{N}},\{\pi_{i}\}_{i\in\mathcal{N}},\bm{\rho}, \gamma\big)$, {\color{blue}where each element is summarized below}.
\par
\textbf{Environmental network}: $\mathcal{G}^{E}\big(\mathcal{N},\mathcal{E}^{E}\big)$ is an environmental network among agents, where $\mathcal{N}=\{1,\cdots,N\}$ denotes the set of agents and the edge set $\mathcal{E}^{E}$ specifies the links between them.
{\color{blue}We employ
the superscript ``E'' for the corresponding concepts in $\mathcal{G}^{E}$.}
\par
\textbf{State and action}: $s_i\in\mathcal{S}_i$ and $a_i\in\mathcal{A}_i$ denote
the local state and action of agent $i$. The global state and action are
$\bm{s}=(s_1,\cdots,s_N)\in\bm{\mathcal{S}}=\prod_{i=1}^N\mathcal{S}_i$ and
$\bm{a}=(a_1,\cdots,a_N)\in\bm{\mathcal{A}}=\prod_{i=1}^N\mathcal{A}_i$, and
$s_{\mathcal{N}_i^{E,\kappa}}\in\mathcal{S}_{\mathcal{N}_i^{E,\kappa}}=
\prod_{j\in\mathcal{N}_i^{E,\kappa}}\mathcal{S}_j$ and
$a_{\mathcal{N}_i^{E,\kappa}}\in\mathcal{A}_{\mathcal{N}_i^{E,\kappa}}=
\prod_{j\in\mathcal{N}_i^{E,\kappa}}\mathcal{A}_j$ denote the aggregated states
and actions of agent $i$'s $\kappa$-hop neighbors.
\par
\textbf{State transition probability function}: $\mathcal{P}_i(s'_i|s_{\mathcal{N}_i^E},a_i)$ specifies
that agent $i$'s next state is determined by its neighbors' states and its own
action, {\color{blue}with the global transition} $\bm{\mathcal{P}}(\bm{s}'|\bm{s},\bm{a})=
\prod_{i=1}^N\mathcal{P}_i(s'_i|s_{\mathcal{N}_i^E},a_i)$.
\par
\textbf{Reward functions}:
$f_i(s_i,a_i)$ and $g_i(s_i,a_i)$ {\color{blue}denote the objective
and constraint reward functions of agent $i$, respectively.}
\par
\textbf{Coupled policy}:
$\pi_i(\cdot|s_{\mathcal{N}_i^{E,\kappa_p}},\theta_i,
\theta_{\mathcal{N}_{i,-i}^{E,\kappa_p}})$ represents the local coupled policy of agent $i$,
where $\kappa_p\geq 1$ is the coupling distance, $\theta_i\in\Theta_i\subseteq
\mathbb{R}^{d_{\theta_i}}$ is the local policy parameter, and
$\theta_{\mathcal{N}_{i,-i}^{E,\kappa_p}}$ are the policy parameters of agent $i$'s
$\kappa_p$-hop neighbors. The joint policy is
$\bm{\pi_\theta}(\bm{a}|\bm{s})=\prod_{i=1}^N\pi_i(a_i|s_{\mathcal{N}_i^{E,\kappa_p}},
\theta_i,\theta_{\mathcal{N}_{i,-i}^{E,\kappa_p}})$ with joint parameter
$\bm{\theta}=(\theta_1^\top,\cdots,\theta_N^\top)^\top\in\bm{\Theta}=
\prod_{i=1}^N\Theta_i$.
{\color{blue}This coupled policy framework is naturally suited for environments with coupled state transition dynamics, such as intelligent transportation systems~\cite{ChuTITS2020,WangTCYB2021,DaiTII2024}, where the next state of each agent is jointly determined by the current states of its neighbors.
It can also be readily implemented as a neural network, where agent $i$'s policy takes $s_{\mathcal{N}_i^{E,\kappa_p}}$ and $\theta_{\mathcal{N}_{i,-i}^{E,\kappa_p}}$ as inputs while each agent optimizes solely its own parameter $\theta_i$.
}
{\color{blue}\begin{remark}
The coupled policy $\pi_i(\cdot|s_{\mathcal{N}_i^{E,\kappa_p}},\theta_i,
\theta_{\mathcal{N}_{i,-i}^{E,\kappa_p}})$ is motivated by empirical findings in cooperative multi-agent applications, where coupled policies conditioning on both state and policy information from neighbors consistently outperform independent policies~\cite{ChuTITS2020,WangTCYB2021,DaiTII2024}.
Although coupled policies have demonstrated empirical success, their theoretical foundations in distributed constrained learning remain largely unexamined.
This work fills this gap by providing a rigorous convergence analysis under the coupled policy setting.
\end{remark}}

\par
\textbf{Initial state distribution}: $\bm{\rho}$ is the distribution of the initial state $\bm{s}_{0}$.
\par
\textbf{Discount factor}: $\gamma\in(0,1)$ is the discount factor.
\par
In {\color{blue}the} CMARL problem, define the local objective function {\color{blue}and constraint function} of agent $i$ under joint policy $\bm{\pi}_{\bm{\theta}}$ as
\begin{align}\label{thelongtermreturninNMARLofagenti}
F_{i}(\bm{\theta})\!=\!\mathbb{E}_{\bm{s}\sim\bm{\rho}}\Big[\sum_{t=0}^{\infty}\!\gamma^{t}f_{i}(s_{i,t},a_{i,t})\Big|\bm{s}_{0}\!=\!\bm{s},\bm{a}_{t}\!\sim\!\bm{\pi_{\theta}}(\cdot|\bm{s}_{t})\Big]
\end{align}
and
\begin{align}\label{thelongtermreturninNMARLGG}
G_{i}(\bm{\theta})\!=\!\mathbb{E}_{\bm{s}\sim\bm{\rho}}\Big[\sum_{t=0}^{\infty}\!\gamma^{t}g_{i}(s_{i,t},a_{i,t})\Big|\bm{s}_{0}\!=\!\bm{s},\bm{a}_{t}\!\sim\!\bm{\pi_{\theta}}(\cdot|\bm{s}_{t})\Big],
\end{align}
where $\bm{s}_{t}=(s_{1,t},\cdots,s_{N,t})$ and $\bm{a}_{t}=(a_{1,t},\cdots,a_{N,t})$ represent the global state and global action at time $t$, respectively.
The optimization objective {\color{blue}of agents} is to find a joint policy parameter $\bm{\theta}$ that
maximizes the discounted average cumulative rewards while satisfying safety constraints:
\begin{align}\label{theproblemofCMARL}
\max\limits_{\bm{\theta}\in\bm{\Theta}}F(\bm{\theta})\triangleq\frac{1}{N}\sum_{i=1}^{N}F_{i}(\bm{\theta}),\;\;\mathrm{s.t.}\;G_{i}(\bm{\theta})\geq c_{i},
\end{align}
where $c_{i}\in\mathbb{R}$ denotes the constraint {\color{blue}threshold} value for $G_{i}(\bm{\theta})$ of agent $i$~\footnote{{\color{blue}
Although the unconstrained cooperative objective admits a MPG formulation, our constrained formulation~(\ref{theproblemofCMARL}) is analyzed through the Lagrangian, and the solution concept is an $\epsilon$-FOSP rather than a Nash equilibrium.
}}.
\subsection{Approximate first-order stationary point}\label{Approximatefirstorderstationarypoint}
{\color{blue}The Lagrangian function associated with~(\ref{theproblemofCMARL}) is}
\begin{align}\label{LagrangianfunctionofCMARL}
\mathcal{L}(\bm{\theta},\bm{\mu}) = F(\bm{\theta}) + \frac{1}{N}\sum_{i=1}^{N}\mu_{i}\big(G_{i}(\bm{\theta})-c_{i}),
\end{align}
where $\bm{\mu} = (\mu_{1},\cdots,\mu_{N})^{\top} \in \mathbb{R}^{N}_{+}$ denotes the vector of Lagrange multipliers.
Accordingly, the Lagrangian formulation of (\ref{theproblemofCMARL}) is expressed as
\begin{align}\label{Lagrangianformulationoforiginalproblem}
\max\limits_{\bm{\theta}\in\bm{\Theta}}\min\limits_{\bm{\mu}\geq\mathbf{0}_{N}}\mathcal{L}(\bm{\theta},\bm{\mu}).
\end{align}
{\color{blue}For the subsequent theoretical development, we introduce the following assumption and definition.}
\begin{assumption}\label{theassumptionofFOSP}
In CMARL problem (\ref{theproblemofCMARL}), there exists a first-order stationary point (FOSP) $(\bm{\theta}^{\star},\bm{\mu}^{\star})$ and a positive constant $\mu_{max}>0$, such that $\mu^{\star}_{i} \leq \mu_{max}$ holds for all $i\in\mathcal{N}$.
\end{assumption}
\par
Assumption~\ref{theassumptionofFOSP} adheres to the standard practice in constrained optimization by assuming that the feasible set of the dual variables is bounded, which is consistent with~\cite{LuAAAI2021,YingNIPS2024}.
\par
Since $F(\bm{\theta})$ and
$\{G_i(\bm{\theta})\}_{i\in\mathcal{N}}$ are generally nonconvex in $\bm{\theta}$, computing the
global optimum of~(\ref{theproblemofCMARL}) is NP-hard.
{\color{blue}In this work, we instead
seek an approximate FOSP of
problem~(\ref{theproblemofCMARL})}.
Under Assumption~\ref{theassumptionofFOSP}, let $\bm{\mathcal{U}}=\prod_{i=1}^{N}\mathcal{U}_{i}$ as the feasible set of the dual variables, where
$\mathcal{U}_{i}=[0,\mu_{max}]$ for each agent $i\in\mathcal{N}$.
{\color{blue}The $\epsilon$-FOSP is defined as follows.}
\begin{definition}\label{thedefinitionofepsilonFOSP}
($\epsilon$-FOSP) A solution $(\bm{\theta},\bm{\mu})\in\bm{\Theta}\times\bm{\mathcal{U}}$ is said to be an $\epsilon$-FOSP, if
\begin{align}\label{theequationindefinitionofepsilonFOSP}
\mathcal{E}(\bm{\theta},\bm{\mu})=\mathcal{X}(\bm{\theta},\bm{\mu})^{2}+\mathcal{Y}(\bm{\theta},\bm{\mu})^{2}\leq\epsilon,
\end{align}
where $\mathcal{X}(\bm{\theta},\bm{\mu})$ and $\mathcal{Y}(\bm{\theta},\bm{\mu})$ are respectively defined as
\begin{align}\label{thedefinitionofmathcalX}
\mathcal{X}(\bm{\theta},\bm{\mu})=-\min\limits_{\bm{\mu}'\in\bm{\mathcal{U}},\|\bm{\mu}'-\bm{\mu}\|\leq1}\langle\nabla_{\bm{\mu}}\mathcal{L}(\bm{\theta},\bm{\mu}),\bm{\mu}'-\bm{\mu}\rangle
\end{align}
and
\begin{align}\label{thedefinitionofmathcalY}
\mathcal{Y}(\bm{\theta},\bm{\mu})=\max\limits_{\bm{\theta}'\in\bm{\Theta},\|\bm{\theta}'-\bm{\theta}\|\leq1}\langle\nabla_{\bm{\theta}}\mathcal{L}(\bm{\theta},\bm{\mu}),\bm{\theta}'-\bm{\theta}\rangle.
\end{align}
\end{definition}
\par
Since $\mathcal{X}(\bm{\theta},\bm{\mu})$ in (\ref{thedefinitionofmathcalX}) and $\mathcal{Y}(\bm{\theta},\bm{\mu})$ in (\ref{thedefinitionofmathcalY}) are defined based on the first-order optimality conditions~\cite{NouiehedNIPS2019},
if there exists $(\bm{\theta}^{*}, \bm{\mu}^{*})\in\bm{\Theta}\times\bm{\mathcal{U}}$ such that $\mathcal{E}(\bm{\theta}^{*}, \bm{\mu}^{*})=0$, then $(\bm{\theta}^{*},\bm{\mu}^{*})$ constitutes a FOSP of problem~(\ref{theproblemofCMARL}).

\subsection{Coupled policy gradient}\label{Coupledpolicygradient}
{\color{blue}For the CMARL problem in~(\ref{theproblemofCMARL}), given any joint policy $\bm{\pi_\theta}$ and Lagrangian multiplier $\bm{\mu}$, we define the global and local Lagrangian $Q$-functions as}
\begin{align}
Q^{\bm{\pi_{\theta}}}(\bm{s},\bm{a};\bm{\mu})=&\mathbb{E}_{\bm{\pi_{\theta}}}\Big[\frac{1}{N}\sum_{t=0}^{\infty}\sum_{i=1}^{N}\gamma^{t}\big(f_{i}(s_{i,t},a_{i,t})\notag\\
&+\mu_{i}g_{i}(s_{i,t},a_{i,t})\big)\Big|\bm{s}_{0}=\bm{s},\bm{a}_{0}=\bm{a}\Big]\label{thedefinitionofglobalQfunction}
\end{align}
and
\begin{align}
Q^{\bm{\pi_{\theta}}}_{i}(\bm{s},\bm{a};\bm{\mu})=&\mathbb{E}_{\bm{\pi_{\theta}}}\Big[\sum_{t=0}^{\infty}\gamma^{t}\big(f_{i}(s_{i,t},a_{i,t})+\mu_{i}g_{i}(s_{i,t},a_{i,t})\big)\Big|\notag\\
&\bm{s}_{0}=\bm{s},\bm{a}_{0}=\bm{a}\Big].\label{thedefinitionoflocalQfunction}
\end{align}
Based on the above definitions, the relationship between the global and local
Lagrangian $Q$-functions {\color{blue}satisfies}
\begin{align}\label{thedecomposeofQfunction}
Q^{\bm{\pi_{\theta}}}(\bm{s},\bm{a};\bm{\mu})=\frac{1}{N}\sum_{i=1}^{N}Q^{\bm{\pi_{\theta}}}_{i}(\bm{s},\bm{a};\bm{\mu}).
\end{align}
\par
{\color{blue}We further define $Q^{\bm{\pi_\theta}}_{g,i}(\bm{s},\bm{a})$ as the local
constraint $Q$-function for agent $i$:}
\begin{align}
Q^{\bm{\pi_{\theta}}}_{g,i}(\bm{s},\bm{a})=&\mathbb{E}_{\bm{\pi_{\theta}}}\Big[\sum_{t=0}^{\infty}\gamma^{t}g_{i}(s_{i,t},a_{i,t}\big)\Big|\bm{s}_{0}=\bm{s},\bm{a}_{0}=\bm{a}\Big].\label{thedefinitionoflocalQfunctionconstraint}
\end{align}
Denote $d^{\bm{\pi_{\theta}}}_{\bm{\rho}}(\bm{s})$ as the discounted state visitation distribution of $\bm{s}\in\bm{\mathcal{S}}$ generated by the joint policy $\bm{\pi_{\theta}}$ and satisfy
\begin{align}\label{Thediscountedstatevisitationdistribution}
d^{\bm{\pi_{\theta}}}_{\bm{\rho}}(\bm{s})=(1-\gamma)\sum_{t=0}^{\infty}\gamma^{t}\mathrm{Pr}^{\bm{\pi_{\theta}}}(\bm{s}_{t}=\bm{s}|\bm{s}_{0}\sim\bm{\rho}),
\end{align}
where $\mathrm{Pr}^{\bm{\pi_{\theta}}}(\bm{s}_{t}=\bm{s}|\bm{s}_{0}\sim\bm{\rho})$ is the probability of occurrence of $\bm{s}_{t}=\bm{s}$ at time $t$ under joint policy $\bm{\pi_{\theta}}$ and initial state distribution $\bm{\rho}$.
\par
{\color{blue}Based on the above definitions, the coupled policy gradients of $\mathcal{L}(\bm{\theta},
\bm{\mu})$ with respect to $\theta_i$ and $\mu_i$ are given in the following lemma.}
\begin{lemma}\label{thelemmaofthepolicygradienttheoreminNMARLCP}
In CMARL problem (\ref{theproblemofCMARL}), for any joint policy $\bm{\pi_{\theta}}$, the gradients of $\mathcal{L}(\bm{\theta},\bm{\mu})$ with respect to $\theta_{i}$ and $\mu_{i}$ are respectively represented as
\begin{align}\label{thecoupledpolicygradienttheorem2}
\nabla_{\mu_{i}}\mathcal{L}(\bm{\theta},\bm{\mu})=\frac{1}{N}\big(G_{i}(\bm{\theta})-c_{i}\big)
\end{align}
and
\begin{align}\label{thecoupledpolicygradienttheorem}
\nabla_{\theta_{i}}\mathcal{L}(\bm{\theta},\bm{\mu})=&\frac{1}{1-\gamma}\mathbb{E}_{\bm{s}\sim d^{\bm{\pi_{\theta}}}_{\bm{\rho}},\bm{a}\sim\bm{\pi_{\theta}}}\Big[Q^{\bm{\pi_{\theta}}}(\bm{s},\bm{a};\bm{\mu})\notag\\
&\times\sum_{j\in\mathcal{N}^{E,\kappa_{p}}_{i}}\nabla_{\theta_{i}}\log\pi_{j}(a_{j}|s_{\mathcal{N}^{E,\kappa_{p}}_{j}},\theta_{j},\theta_{\mathcal{N}^{E,\kappa_{p}}_{j,-j}})\Big].
\end{align}
\end{lemma}
\par
Lemma~\ref{thelemmaofthepolicygradienttheoreminNMARLCP} provides the explicit expressions for the gradients of both the Lagrange multiplier and the coupled policy, with the corresponding proof detailed in Section~\ref{ProofofLemmathelemmaofthepolicygradienttheoreminNMARLCP}.

\section{DSPD algorithm}\label{DSPDalgorithm}
{\color{blue}In this section, we introduce the approximated coupled policy gradient and propose
the DSPD algorithm to address the CMARL problem under coupled policies.}
\subsection{Approximated policy gradient}
{\color{blue}To reduce the reliance on global state-action information, we adopt a truncation-based approach~\cite{QuCLDC2020} and design a novel neighbors' averaged Lagrangian $Q$-function for each agent $i$:}
\begin{align}\label{theaction-averagedQfunctionofagenti}
\widehat{Q^{\bm{\pi_{\theta}}}_{i}}(\bm{s},\bm{a};\bm{\mu})=&\mathbb{E}_{\bm{\pi_{\theta}}}\Big[\frac{1}{N}\sum^{\infty}_{t=0}\gamma^{t}\sum_{j\in\mathcal{N}^{E,\kappa+2\kappa_{p}}_{i}}\big(f_{j}(s_{j,t},a_{j,t})\notag\\
&+\mu_{j}g_{j}(s_{j,t},a_{j,t})\big)\Big|\bm{s}_{0}=\bm{s},\bm{a}_{0}=\bm{a}\Big].
\end{align}
{\color{blue}The $(\kappa+2\kappa_p)$-hop neighborhood in~(\ref{theaction-averagedQfunctionofagenti})
arises naturally from two components of the coupled policy gradient~(\ref{thecoupledpolicygradienttheorem}):
the log-policy gradients of $\kappa_p$-hop neighbors require state-action pairs from
their own $\kappa_p$-hop neighborhoods (i.e., the $2\kappa_p$-hop neighborhood of
agent $i$), and inspired by the $\kappa$-distance truncation in~\cite{QuCLDC2020},
an additional $\kappa$-hop range is introduced to approximate the global Lagrangian
$Q$-function, yielding the total range of $\kappa+2\kappa_p$ hops.}
Based on~(\ref{theaction-averagedQfunctionofagenti}), the approximated policy gradient
for agent $i$ is
\begin{align}\label{thepolicygradientapproxiamtionbynewfunction}
g^{\bm{\pi}_{\bm{\theta}},\bm{\mu}}_{app,i}=&\frac{1}{1-\gamma}\mathbb{E}_{\bm{s}\sim d^{\bm{\pi_{\theta}}}_{\bm{\rho}},\bm{a}\sim\bm{\pi_{\theta}}}\Big[\widehat{Q^{\bm{\pi_{\theta}}}_{i}}(\bm{s},\bm{a};\bm{\mu})\notag\\
&\times\!\!\sum_{j\in\mathcal{N}^{E,\kappa_{p}}_{i}}\!\!\nabla_{\theta_{i}}\log\pi_{j}(a_{j}|s_{\mathcal{N}^{E,\kappa_{p}}_{j}},\theta_{j},\theta_{\mathcal{N}^{E,\kappa_{p}}_{j,-j}})\Big],
\end{align}
which replaces the global Lagrangian $Q$-function in~(\ref{thecoupledpolicygradienttheorem})
with the neighbors' averaged Lagrangian $Q$-function, offering an approximation to
$\nabla_{\theta_i}\mathcal{L}(\bm{\theta},\bm{\mu})$.
{\color{blue}Since
$\widehat{Q^{\bm{\pi_\theta}}_i}(\bm{s},\bm{a};\bm{\mu})$ can be estimated from reward signals within the
$(\kappa+2\kappa_p)$-hop neighborhood and $g^{\bm{\pi_\theta},\bm{\mu}}_{app,i}$
from state-action pairs within the $2\kappa_p$-hop neighborhood, this provides a
theoretical basis for fully distributed implementation.}
\subsection{{\color{blue}Estimation of agent-specific parameters}}
{\color{blue}In the $m$-th iteration, let $\bm{\theta}_m$ and $\bm{\mu}_m$ denote
the joint policy parameter and Lagrangian multiplier. To reduce the exposure of
sensitive learned parameters, each agent exchanges true values only with immediate
neighbors over the time-varying learning networks $\{\mathcal{G}^L_m\}_{m\geq 1}$,
and maintains local estimates of other agents' parameters. Specifically, we define}
$\bm{\hat{\theta}}^i_m=\big((\hat{\theta}^i_{1,m})^\top,\cdots,(\hat{\theta}^i_{N,m})^\top\big)^\top$
and $\bm{\hat{\mu}}^i_m=\big((\hat{\mu}^i_{1,m})^\top,\cdots,(\hat{\mu}^i_{N,m})^\top\big)^\top$
as agent $i$'s estimates of other agents' policy parameters and Lagrange multipliers,
and $\bm{\breve{\theta}}^i_m=\big((\breve{\theta}^i_{1,m})^\top,\cdots,(\theta{\mu}^i_{N,m})^\top\big)^\top$, $\bm{\breve{\mu}}^i_m=\big((\breve{\mu}^i_{1,m})^\top,\cdots,(\breve{\mu}^i_{N,m})^\top\big)^\top$ as the corresponding
intermediate variables in the learning process.
\begin{assumption}\label{theassumptionofcommunicationdistance}
During the learning process, each agent $i$ receives state-action pairs
$\{(s_j,a_j)\}_{j\in\mathcal{N}_i^{E,2\kappa_p}}$ and reward signals
$\{f_j,g_j\}_{j\in\mathcal{N}_i^{E,\kappa+2\kappa_p}}$ through $\mathcal{G}^E$,
and exchanges intermediate variables $\{\bm{\breve{\theta}}^i_m,\bm{\breve{\mu}}^i_m\}$
and true values $\{\theta_{i,m},\mu_{i,m}\}$ with direct neighboring agents over
$\{\mathcal{G}^L_m(\mathcal{N},\mathcal{E}^L_m)\}_{m\geq 1}$.
\end{assumption}
\par
{\color{blue}Assumption~\ref{theassumptionofcommunicationdistance} introduces a dedicated learning network $\{\mathcal{G}^L_m\}_{m\geq 1}$, separate from the environmental network $\mathcal{G}^{E}$, motivated by practical systems such as networked control systems~\cite{XuShi2025} and digital twin-enabled systems~\cite{QiTao2018}, where the communication infrastructure for parameter exchange naturally differs from that for environmental interactions. This design limits the exposure of each agent's
decision-making variables to a minimal set of neighbors, thereby reducing vulnerability to potential network attacks compared to frameworks requiring global parameter broadcast~\cite{YingNIPS2024}.}

\subsection{DSPD algorithm designing}
Under Assumption~\ref{theassumptionofcommunicationdistance}, agents receive
states from theirs $\kappa_{p}$-hop environmental neighbors and execute joint policy $\bm{\hat{\pi}}_{\bm{\hat{\theta}}_{m}}=\prod_{i=1}^{N}\pi_{i}(a_{i}|s_{\mathcal{N}^{E,\kappa_{p}}_{i}},\hat{\theta}^{i}_{i,m},\hat{\theta}^{i}_{\mathcal{N}^{E,\kappa_{p}}_{i,-i},m})$.
{\color{blue}The DSPD algorithm proceeds in three steps.}
\par
\textbf{Step 1. Parameter estimation}: In the $m$-th iteration, {\color{blue}the estimates $\hat{\mu}^i_{j,m}$ and $\hat{\theta}^i_{j,m}$ are computed via the push-sum
protocol~\cite{NedicTAC2015}:}
\begin{subequations}\label{calculatetheestimates}
\begin{numcases}{}
\!\!p_{i,m+1}=\sum_{l\in\mathcal{N}^{L}_{i,m}}w^{L}_{il,m}p_{l,m}\label{thekeyupdateofpolicyparameterp}\\
\!\!\hat{\mu}^{i}_{j,m}=\frac{1}{p_{i,m+1}}\sum_{l\in\mathcal{N}^{L}_{i,m}}w^{L}_{il,m}\breve{\mu}^{l}_{j,m}\label{thekeyupdateofpolicyparametermu-3}\\
\!\!\hat{\theta}^{i}_{j,m}=\frac{1}{p_{i,m+1}}\sum_{l\in\mathcal{N}^{L}_{i,m}}w^{L}_{il,m}\breve{\theta}^{l}_{j,m},\label{thekeyupdateofpolicyparametertheta-3}
\end{numcases}
\end{subequations}
{\color{blue}with $\breve{\mu}^i_{j,1}=0$, $\breve{\theta}^i_{j,1}=\mathbf{0}_{d_{\theta_j}}$,
and $p_{i,1}=1$ for all $i,j\in\mathcal{N}$.}
\par
\textbf{Step 2. Lagrangian multiplier update}: Each agent $i$ estimates $\nabla_{\mu_i}\mathcal{L}(\bm{\theta}_m,\bm{\mu}_m)$ using $K_\mu$ batches of samples.
{\color{blue}For each batch $k\in\{1,\cdots,K_\mu\}$, agents execute
$\bm{\hat{\pi}}_{\bm{\hat{\theta}}_m}$ to generate a sample trajectory $\{\bm{s},\bm{a}\}_{0:T_1(k)}$ of length $T_1(k)\sim\mathrm{Geom}(1-\gamma^{1/2})$, and each agent $i$ computes the gradient estimate:}
\begin{align}\label{theapproximatedgradientinalgorithmdesignmu}
\hat{h}^{\bm{\hat{\pi}}_{\bm{\hat{\theta}}_{m}},\bm{\hat{\mu}}_{m}}_{app,i}(k)=\frac{1}{N}\Big(\sum_{t=0}^{T_{1}(k)}\gamma^{t/2}g_{i,t}-c_{i}\Big).
\end{align}
Based on (\ref{theapproximatedgradientinalgorithmdesignmu}), we define the averaged estimate $\bar{h}^{\bm{\hat{\pi}}_{\bm{\hat{\theta}}_{m}},\bm{\hat{\mu}}_{m}}_{app,i}$ as
\begin{align}\label{theapproximatedgradientinalgorithmdesignmubar}
\bar{h}^{\bm{\hat{\pi}}_{\bm{\hat{\theta}}_{m}},\bm{\hat{\mu}}_{m}}_{app,i}=\frac{1}{K_{\mu}}\sum_{k=1}^{K_{\mu}}\hat{h}^{\bm{\hat{\pi}}_{\bm{\hat{\theta}}_{m}},\bm{\hat{\mu}}_{m}}_{app,i}(k)
\end{align}
and update the Lagrangian multiplier $\mu_{i,m+1}$ according to
\begin{align}\label{theupdateoftruepolicyparametersmu}
\mu_{i,m+1}=P_{\mathcal{U}_{i}}(\mu_{i,m}-\eta_{\mu,m}\bar{h}^{\bm{\hat{\pi}}_{\bm{\hat{\theta}}_{m}},\bm{\hat{\mu}}_{m}}_{app,i}),
\end{align}
where $\eta_{\mu,m}$ is the learning rate of Lagrangian multiplier.
\par
{\color{blue}The intermediate variable of the Lagrangian multipliers $\breve{\mu}^{i}_{j,m+1}$ is then revised as}
\begin{align}
\!\!\breve{\mu}^{i}_{j,m+1}\!\!=\!\!\left\{
\begin{array}{ll}\label{thekeyupdateofpolicyparametermu-1}
\!\!\sum\limits_{l\in\mathcal{N}^{L}_{i,m}}\!\!\!w^{L}_{il,m}\breve{\mu}^{l}_{j,m},~\mathrm{if}~j\notin\mathcal{N}^{L}_{i,m}\\
\!\!\sum\limits_{l\in\mathcal{N}^{L}_{i,m}}\!\!\!w^{L}_{il,m}\breve{\mu}^{l}_{j,m}\!\!+\!\!w^{L}_{ij,m}N(\mu_{j,m+1}\!\!-\!\!\mu_{j,m}),\mathrm{o/w},
\end{array}\right.
\end{align}
{\color{blue}which ensures $(1/N)\sum_{i=1}^N\breve{\mu}^i_{j,m}=\mu_{j,m}$, guaranteeing average consensus convergence of $\hat{\mu}^i_{j,m}$ towards $\mu_{j,m}$.}
\par
\textbf{Step 3. Policy parameter update}: After Step~2, each agent $i$ first
updates $\bm{\hat{\mu}}^i_{m+1}$ via~(\ref{thekeyupdateofpolicyparametermu-3}).
Analogous to Step~2, each agent $i$ estimates $\nabla_{\theta_i}\mathcal{L}
(\bm{\theta}_m,\bm{\mu}_{m+1})$ using $K_\theta$ batches of samples.
{\color{blue}For each batch $k\in\{1,\cdots,K_\theta\}$, agents execute $\bm{\hat{\pi}}_{\bm{\hat{\theta}}_m}$ to generate two sample trajectories $\{\bm{s},\bm{a}\}_{0:T_2(k)}$ and
$\{\bm{s},\bm{a}\}_{T_2(k):T_2(k)+T_3(k)}$, where $T_2(k)\sim\mathrm{Geom}(1-\gamma)$ and $T_3(k)\sim\mathrm{Geom}(1-\gamma^{1/2})$, and each agent $i$ computes the approximated policy gradient estimate:}
\begin{align}\label{theapproximatedgradientinalgorithmdesign}
\hat{g}^{\bm{\hat{\pi}}_{\bm{\hat{\theta}}_{m}},\bm{\hat{\mu}}_{m+1}}_{app,i}(k)=&\frac{1}{1-\gamma}\hat{Q}^{\bm{\hat{\pi}}_{\bm{\hat{\theta}}_{m}},\bm{\hat{\mu}}_{m+1}}_{i,T_{2}(k)}\!\!\!\sum_{j\in\mathcal{N}^{E,\kappa_{p}}_{i}}\!\!\!\nabla_{\theta_{i}}\log\pi_{j}(a_{j,T_{2}(k)}|\notag\\
&s_{\mathcal{N}^{E,\kappa_{p}}_{j},T_{2}(k)},\hat{\theta}^{i}_{j,m},\hat{\theta}^{i}_{\mathcal{N}^{E,\kappa_{p}}_{j,-j},m}),
\end{align}
where $\hat{Q}^{\bm{\hat{\pi}}_{\bm{\hat{\theta}}_m},\bm{\hat{\mu}}_{m+1}}_{i,T_2(k)}$
is {\color{blue}a sample-based estimate} of the neighbors' averaged Lagrangian $Q$-function
$\widehat{Q^{\bm{\hat{\pi}}_{\bm{\hat{\theta}}_m}}_i}(\bm{s}_{T_2(k)},\bm{a}_{T_2(k)};
\bm{\hat{\mu}}^i_{m+1})$ in~(\ref{theaction-averagedQfunctionofagenti}), {\color{blue}computed
from the second trajectory $\{\bm{s},\bm{a}\}_{T_2(k):T_2(k)+T_3(k)}$:}
\begin{align}\label{theestimateofwildehatQi}
\hat{Q}^{\bm{\hat{\pi}}_{\bm{\hat{\theta}}_{m}},\bm{\hat{\mu}}_{m+1}}_{i,T_{2}(k)}=&\frac{1}{N}\sum_{t=0}^{T_{3}(k)}\gamma^{t/2}\sum_{j\in\mathcal{N}^{E,\kappa+2\kappa_{p}}_{i}}\big(f_{j,T_{2}(k)+t}+\notag\\
&\hat{\mu}^{i}_{j,m+1}g_{j,T_{2}(k)+t}\big).
\end{align}
Define the averaged policy gradient estimate for agent $i$ as
\begin{align}\label{theapproximatedgradientinalgorithmdesignthetabar}
\bar{g}^{\bm{\hat{\pi}}_{\bm{\hat{\theta}}_{m}},\bm{\hat{\mu}}_{m+1}}_{app,i}=\frac{1}{K_{\theta}}\sum_{k=1}^{K_{\theta}}\hat{g}^{\bm{\hat{\pi}}_{\bm{\hat{\theta}}_{m}},\bm{\hat{\mu}}_{m+1}}_{app,i}(k),
\end{align}
the policy parameter is updated according to
\begin{align}\label{theupdateoftruepolicyparameterstheta}
\theta_{i,m+1}=P_{\Theta_{i}}(\theta_{i,m}+\eta_{\theta,m}\bar{g}^{\bm{\hat{\pi}}_{\bm{\hat{\theta}}_{m}},\bm{\hat{\mu}}_{m+1}}_{app,i}),
\end{align}
where $\eta_{\theta,m}$ is the learning rate of policy parameter.
\begin{remark}
The use of $\bm{\hat{\mu}}_{m+1}$ instead of $\bm{\hat{\mu}}_{m}$ in $\bar{g}^{\bm{\hat{\pi}}_{\bm{\hat{\theta}}_{m}},\bm{\hat{\mu}}_{m+1}}_{app,i}$ in (\ref{theapproximatedgradientinalgorithmdesignthetabar}) is advantageous, as it facilitates the theoretical analysis of Theorem~\ref{thetheoremconvergenceofpolicygradient}, where the $\epsilon$-policy stationary convergence is rigorously established.
\end{remark}
\par
{\color{blue}Similar to (\ref{thekeyupdateofpolicyparametermu-1}),
the intermediate variable for policy parameters is updated as}
\begin{align}
\!\!\breve{\theta}^{i}_{j,m+1}\!\!=\!\!\left\{
\begin{array}{ll}\label{thekeyupdateofpolicyparametertheta-1}
\!\!\!\sum\limits_{l\in\mathcal{N}^{L}_{i,m}}\!\!\!w^{L}_{il,m}\breve{\theta}^{l}_{j,m},~\mathrm{if}~j\notin\mathcal{N}^{L}_{i,m}\\
\!\!\!\sum\limits_{l\in\mathcal{N}^{L}_{i,m}}\!\!\!w^{L}_{il,m}\breve{\theta}^{l}_{j,m}\!\!+\!\!w^{L}_{ij,m}N(\theta_{j,m+1}\!\!-\!\!\theta_{j,m}),\mathrm{o/w},
\end{array}\right.
\end{align}
{\color{blue}which ensures $(1/N)\sum_{i=1}^N\breve{\theta}^i_{j,m}=\theta_{j,m}$ and average
consensus convergence of $\hat{\theta}^i_{j,m}$ to $\theta_{j,m}$.}
\par
To integrate Steps~1-3, the pseudocode of the DSPD algorithm is proposed in Algorithm~\ref{distributedscalableprimaldualAlgorithm}.
\par
As illustrated in Algorithm~\ref{distributedscalableprimaldualAlgorithm}, the DSPD
algorithm operates in a fully distributed manner involving two types of information:
{\color{blue}environmental data (state-action pairs and rewards) and sensitive learned information
(policy parameters and Lagrange multipliers).
Specifically, each agent $i$ receives state-action pairs and rewards through $\mathcal{G}^E$ (refer to Line~15), and exchanges estimated and true parameter values with direct neighbors through $\mathcal{G}^L_m$ (refer to Lines~4,~11, and~18).}
{\color{blue}The
information exchange structure is summarized in Table~\ref{table:information}.
\begin{table}[!h]
\centering
\caption{{\color{blue}Summary of information exchange of agent $i$ in the DSPD algorithm.}}
\label{table:information}
{\color{blue}\begin{tabular}{lcc}
\hline
\textbf{Information} & \textbf{Range} & \textbf{Network}\\
\hline
State-action pairs & $2\kappa_p$-hop & $\mathcal{G}^E$\\
Reward/constraint signals & $(\kappa+2\kappa_p)$-hop &
$\mathcal{G}^E$\\
True parameters $\{\theta_{i,m},\mu_{i,m}\}$ & 1-hop &
$\mathcal{G}^L_m$\\
Intermediate variables $\{\breve{\theta}^i_{j,m},
\breve{\mu}^i_{j,m}\}$ & 1-hop & $\mathcal{G}^L_m$\\
\hline
\end{tabular}}
\end{table}

\begin{remark}
Although agents execute the coupled policy using estimated parameters
$\{\hat{\theta}^i_{j,m}\}_{j\in\mathcal{N}}$, the coupling is genuine in two
aspects: (i) the state information $s_{\mathcal{N}_i^{E,\kappa_p}}$ is directly
accessible through $\mathcal{G}^E$; and (ii) as established in
Theorem~\ref{Convergencetheoreminparameterestimation}, $\lim_{m\to\infty}
\hat{\theta}^i_{j,m}=\theta_{j,m}$, ensuring the
executed policy asymptotically converges to the intended coupled policy.
\end{remark}}

\begin{algorithm}[htb!]\label{distributedscalableprimaldualAlgorithm}
\SetAlgoLined
\textbf{Require:} $M,K_\mu,K_\theta\in\mathbb{N}^+$, learning rates $\eta_{\mu,m}$ and $\eta_{\theta,m}$\;
\textbf{Initialize:} $\mu_{j,1}=\breve{\mu}^i_{j,1}=0$, $\theta_{j,1}=\breve{\theta}^i_{j,1}=\mathbf{0}_{d_{\theta_i}}$, $\forall i,j\in\mathcal{N}$\;
\For{$m=1,\cdots,M$}{
    Compute $\{\hat{\mu}^i_{j,m}\}_{j\in\mathcal{N}}$ and $\{\hat{\theta}^i_{j,m}\}_{j\in\mathcal{N}}$ via~(\ref{calculatetheestimates})\;
    \tcc{Distributed update of Lagrangian multipliers}
    \For{$k=1,\cdots,K_\mu$}{
        {\color{blue}Draw} $\bm{s}_0\sim\bm{\rho}$, $T_1(k)\sim\mathrm{Geom}(1-\gamma^{1/2})$\;
        Agents execute $\bm{\hat{\pi}}_{\bm{\hat{\theta}}_m}$ to generate $\{\bm{s},\bm{a}\}_{0:T_1(k)}$\;
        Each agent $i$ computes $\hat{h}^{\bm{\hat{\pi}}_{\bm{\hat{\theta}}_m},\bm{\hat{\mu}}_m}_{app,i}(k)$ {\color{blue}via}~(\ref{theapproximatedgradientinalgorithmdesignmu})\;
    }
    Each agent $i$ computes $\bar{h}^{\bm{\hat{\pi}}_{\bm{\hat{\theta}}_m},\bm{\hat{\mu}}_m}_{app,i}$ {\color{blue}via}~(\ref{theapproximatedgradientinalgorithmdesignmubar}) and updates $\mu_{i,m+1}$ {\color{blue}via}~(\ref{theupdateoftruepolicyparametersmu})\;
    Each agent $i$ updates $\{\breve{\mu}^i_{j,m+1}\}_{j\in\mathcal{N}}$ via~(\ref{thekeyupdateofpolicyparametermu-1}) and $\{\hat{\mu}^i_{j,m+1}\}_{j\in\mathcal{N}}$ via~(\ref{thekeyupdateofpolicyparametermu-3})\;
    \tcc{Distributed update of policy parameters}
    \For{$k=1,\cdots,K_\theta$}{
        {\color{blue}Draw} $\bm{s}_0\sim\bm{\rho}$, $T_2(k)\sim\mathrm{Geom}(1-\gamma)$, $T_3(k)\sim\mathrm{Geom}(1-\gamma^{1/2})$\;
        Agents execute $\bm{\hat{\pi}}_{\bm{\hat{\theta}}_m}$ to generate $\{\bm{s},\bm{a}\}_{0:T_2(k)}$ and $\{\bm{s},\bm{a}\}_{T_2(k):T_2(k)+T_3(k)}$\;
        Each agent $i$ collects $(s_{\mathcal{N}_i^{E,2\kappa_p},T_2(k)},a_{\mathcal{N}_i^{E,2\kappa_p},T_2(k)})$ and $\{f_{\mathcal{N}_i^{E,\kappa+2\kappa_p},T_2(k)+t},g_{\mathcal{N}_i^{E,\kappa+2\kappa_p},T_2(k)+t}\}_{t=0}^{T_3(k)}$ to compute $\hat{g}^{\bm{\hat{\pi}}_{\bm{\hat{\theta}}_m},\bm{\hat{\mu}}_{m+1}}_{app,i}(k)$ {\color{blue}via}~(\ref{theapproximatedgradientinalgorithmdesign})\;
    }
    Each agent $i$ computes $\bar{g}^{\bm{\hat{\pi}}_{\bm{\hat{\theta}}_m},\bm{\hat{\mu}}_{m+1}}_{app,i}$ {\color{blue}via}~(\ref{theapproximatedgradientinalgorithmdesignthetabar}) and updates $\theta_{i,m+1}$ {\color{blue}via}~(\ref{theupdateoftruepolicyparameterstheta})\;
    Each agent $i$ updates $\{\breve{\theta}^i_{j,m+1}\}_{j\in\mathcal{N}}$ {\color{blue}via}~(\ref{thekeyupdateofpolicyparametertheta-1})\;
}
\textbf{Output:} $\{\bm{\theta}_m\}_{m=1}^{M+1}$ and $\{\bm{\mu}_m\}_{m=1}^{M+1}$.
\caption{The DSPD Algorithm}
\end{algorithm}

\section{Convergence results}\label{thesectionofanalysis}
{\color{blue}Before proceeding with the analysis, we introduce a set of standard assumptions~\cite{LuAAAI2021,YingNIPS2024,QuCLDC2020,Zhouzhaoyi2023,Zhangrunyu2022} as follows.}
\begin{assumption}\label{theassumptionofreward}
In CMARL problem (\ref{theproblemofCMARL}), there exists $R_f,R_g>0$ such that $|f_i(s_i,a_i)|\leq R_f$ and $|g_i(s_i,a_i)|\leq R_g$ for all $i\in\mathcal{N}$ and $(s_i,a_i)\in\mathcal{S}_i\times\mathcal{A}_i$.
\end{assumption}
\begin{assumption}\label{theassumptionofdistributionofs}
For any joint policy parameter $\bm{\theta}$, $d^{\bm{\pi_{\theta}}}_{\bm{\rho}}(\bm{s})$ and $\xi^{\bm{\pi_{\theta}}}_{\bm{\rho}}(\bm{s},\bm{a})$ satisfy  $\inf_{\bm{\theta}}\min_{\bm{s}\in\bm{\mathcal{S}}}d^{\bm{\pi_{\theta}}}_{\bm{\rho}}(\bm{s})>0$ and $\inf_{\bm{\theta}}\min_{(\bm{s},\bm{a})\in\bm{\mathcal{S}}\times\bm{\mathcal{A}}}\xi^{\bm{\pi_{\theta}}}_{\bm{\rho}}(\bm{s},\bm{a})>0$.
\end{assumption}
\par
Assumption~\ref{theassumptionofdistributionofs} indicates that the coupled policies explore all the states and actions with some positive probability, {\color{blue}which is a common assumption in the related literature~\cite{Zhouzhaoyi2023,Zhangrunyu2022}.}
\begin{assumption}\label{theassumptionofpolicy}
For any joint policy $\bm{\pi_\theta}$ and agent $i\in\mathcal{N}$, the local policy is differentiable with respect to $\{\theta_j\}_{j\in\mathcal{N}_i^{E,\kappa_p}}$.
For any $j\in\mathcal{N}_i^{E,\kappa_p}$, the gradient $\nabla_{\theta_j}\log\pi_i(a_i|s_{\mathcal{N}_i^{E,\kappa_p}},\theta_i,\theta_{\mathcal{N}_{i,-i}^{E,\kappa_p}})$ exists, satisfies $\|\nabla_{\theta_j}\log\pi_i(a_i|s_{\mathcal{N}_i^{E,\kappa_p}},\theta_i,\theta_{\mathcal{N}_{i,-i}^{E,\kappa_p}})\|\leq B$ for some $B>0$, and is $L$-Lipschitz continuous with respect to $\theta_j$.
\end{assumption}
\par
Assumption~\ref{theassumptionofpolicy} implies that the gradients of the log of the policy functions have bounded norms and exhibit Lipschitz continuity.
{\color{blue}Specifically, this assumption is satisfied by common policy parameterizations such as softmax or log-linear policies.}
\begin{assumption}\label{theassumptiontime-varyingnetworks}
The time-varying learning networks $\{\mathcal{G}^{L}_{m}\}_{m\geq1}$ are uniformly strongly connected.
\end{assumption}
\begin{assumption}\label{theassumptionoflearningrate}
In Algorithm~\ref{distributedscalableprimaldualAlgorithm}, the learning rates $\eta_{\theta,m}$ and $\eta_{\mu,m}$ are set such that $\eta_{\mu,m}\sim{\color{blue}\mathcal{O}\big(\frac{1}{m}}\big)$ and  $\eta_{\theta,m}\sim{\color{blue}\mathcal{O}\big(\frac{1}{\sqrt{m}}\big)}$.
\end{assumption}
\par
{\color{blue}Assumptions~\ref{theassumptiontime-varyingnetworks}-\ref{theassumptionoflearningrate} are standard conditions on time-varying networks and learning rates essential for guaranteeing convergence.}
\subsection{Main results}
{\color{blue}Three main theoretical contributions of the proposed DSPD algorithm are established as follows.
Theorem~\ref{thelemmaofapproximatedpolicygradient} bounds the approximation error of the approximated coupled policy gradient, involving a novel joint function of both $\kappa$ and $\kappa_p$.
Theorem~\ref{Convergencetheoreminparameterestimation} establishes the convergence
of the distributed parameter estimates via a carefully ordered interleaved
analysis.
Theorem~\ref{thetheoremconvergenceofpolicygradient} establishes the
overall $\epsilon$-policy stationary convergence of the proposed algorithm.}
\begin{theorem}[Approximated policy gradient error]\label{thelemmaofapproximatedpolicygradient}
Suppose Assumptions~\ref{theassumptionofFOSP} and~\ref{theassumptionofreward}-\ref{theassumptionofdistributionofs} hold.
For any joint policy $\bm{\pi_\theta}$, Lagrangian multiplier $\bm{\mu}$, and agent $i\in\mathcal{N}$, the error between the approximated policy gradient $g^{\bm{\pi_\theta},\bm{\mu}}_{app,i}$ and the exact gradient $\nabla_{\theta_i}\mathcal{L}(\bm{\theta},\bm{\mu})$ satisfies
\begin{align}\label{theerroroftruncatedpolicygradientapproximate}
&\|g^{\bm{\pi}_{\bm{\theta}},\bm{\mu}}_{app,i}-\nabla_{\theta_{i}}\mathcal{L}(\bm{\theta},\bm{\mu})\|\notag\\
\leq&\frac{2(R_{f}+\mu_{max}R_{g})BN}{(1-\gamma)^{2}}\gamma^{h(\kappa,\kappa_{p})+1},
\end{align}
where $h(\kappa,\kappa_{p})$ is defined as
\begin{align}\label{thefunctionofhkappap}
h(\kappa,\kappa_{p})=\left\{
\begin{array}{ll}
\frac{\kappa+1}{\kappa_{p}}-1,~\mathrm{if}~\mathrm{mod}(\kappa+1,\kappa_{p})=0\\
\lfloor\frac{\kappa+1}{\kappa_{p}}\rfloor,~~~\mathrm{otherwise}.
\end{array}\right.
\end{align}
\end{theorem}
\par
The detailed analysis of Theorem~\ref{thelemmaofapproximatedpolicygradient} is provided in Section~\ref{AnalysisofTheoremthelemmaofapproximatedpolicygradient}.
Theorem~\ref{thelemmaofapproximatedpolicygradient} establishes an upper bound on the approximation error of the coupled policy gradient. {\color{blue}Compared to existing works~\cite{QuCLDC2020,YingNIPS2024} where the approximation error depends only on $\kappa$, the error bound here involves the joint function $h(\kappa,\kappa_p)$ of both $\kappa$ and $\kappa_p$, reflecting the additional approximation complexity introduced by the coupled policy structure.}
\par
For the constraints $\{c_{i}\}_{i\in\mathcal{N}}$ in problem (\ref{theproblemofCMARL}), we define $c_{max}=\max_{i\in\mathcal{N}}|c_{i}|$ and have the following theorem.
\begin{theorem}[Convergence of parameter estimations]\label{Convergencetheoreminparameterestimation}
Suppose Assumptions~\ref{theassumptionofFOSP}-\ref{theassumptionoflearningrate} hold.
In Algorithm~\ref{distributedscalableprimaldualAlgorithm}, for any $i,j\in\mathcal{N}$, there exists $\tilde{\mu}_{max}>\mu_{max}$, $L^{\mu}_{1}=\frac{R_{g}+c_{max}}{(1-\gamma^{1/2})N}$, and  $L^{\theta}_{1}=\frac{BN(R_{f}+\tilde{\mu}_{max}R_{g})}{(1-\gamma)(1-\gamma^{1/2})}$, such that
\par
(i) $\|\hat{h}^{\bm{\hat{\pi}}_{\bm{\hat{\theta}}_{m}},\bm{\hat{\mu}}_{m}}_{app,i}(k)\|\leq L^{\mu}_{1}$ for all $k\in\{1,\cdots,K_{\mu}\}$ and $\lim_{m\rightarrow\infty}\hat{\mu}^{i}_{j,m}=\mu_{j,m}$ with rate {\color{blue}$\mathcal{O}\big(\frac{1}{m}\big)$};
\par
(ii) $\|\hat{g}^{\bm{\hat{\pi}}_{\bm{\hat{\theta}}_{m}},\bm{\hat{\mu}}_{m+1}}_{app,i}(k)\|\leq L^{\theta}_{1}$ for all $k\in\{1,\cdots,K_{\theta}\}$ and $\lim_{m\rightarrow\infty}\hat{\theta}^{i}_{j,m}=\theta_{j,m}$ with rate {\color{blue}$\mathcal{O}\big(\frac{1}{\sqrt{m}}\big)$}.
\end{theorem}
\par
The proof of Theorem~\ref{Convergencetheoreminparameterestimation} can be found in Section~\ref{proofofTheoremConvergencetheoreminparameterestimation}.
Theorem~\ref{Convergencetheoreminparameterestimation} demonstrates that all parameter estimates converge to their true values, ensuring the accuracy of the policies executed by the agents and providing reliable gradient information for the updates in~(\ref{theupdateoftruepolicyparametersmu}) and~(\ref{theupdateoftruepolicyparameterstheta}). {\color{blue}The convergence of both estimates is established through a carefully ordered interleaved analysis that resolves the circular dependency between the policy parameter and Lagrange multiplier estimation processes.}
\begin{theorem}[$\epsilon$-policy stationary convergence]\label{thetheoremconvergenceofpolicygradient}
Suppose Assumptions~\ref{theassumptionofFOSP}-\ref{theassumptionoflearningrate} hold.
For every $\epsilon>0$ and $\delta>0$, let $\{(\bm{\theta}_{m},\bm{\mu}_{m})\}_{m=1}^{M}$ be the sequence generated by Algorithm~\ref{distributedscalableprimaldualAlgorithm} with {\color{blue}$M=\tilde{\mathcal{O}}(1/\epsilon^{2})$}, {\color{blue}$\eta_{\theta,m}=1/(2\sqrt{m}+L_{\theta\theta})$},  {\color{blue}$\eta_{\mu,m}=1/2m$}, and $K_{\theta}=K_{\mu}=\log{(2/\delta)}/2\epsilon^{2}$, where $L_{\theta\theta}=LN^{\frac{3}{2}}(R_{f}+\tilde{\mu}_{max}R_{g})/(1-\gamma)^{2}+(1+\gamma)B^{2}N^{3}(R_{f}+\tilde{\mu}_{max}R_{g})/(1-\gamma)^{3}$ is Lipschitz constant.
Then, with probability $1-\delta$, there exists
$m^{*}\in\{1,\cdots,M-1\}$ such that
\begin{align}\label{theresultoftheorem3}
\big[\mathcal{Y}(\bm{\theta}_{m^{*}},\bm{\mu}_{m^{*}+1})\big]^{2}=\mathcal{O}(\epsilon)+\mathcal{O}\big(\epsilon_{0}(\kappa,\kappa_{p})^{2}\big),
\end{align}
where $\epsilon_{0}(\kappa,\kappa_{p})=2(R_{f}+\tilde{\mu}_{max}R_{g})BN^{\frac{3}{2}}\gamma^{h(\kappa,\kappa_{p})+1}/(1-\gamma)^{2}$.
\end{theorem}
\par
The detailed analysis of Theorem~\ref{thetheoremconvergenceofpolicygradient} is provided in Section~\ref{AnalysisofTheoremthetheoremconvergenceofpolicygradient}.
Theorem~\ref{thetheoremconvergenceofpolicygradient} establishes  that, with high probability, our Algorithm~\ref{distributedscalableprimaldualAlgorithm} achieves an $\epsilon$-policy stationary convergence with an approximation error of $\mathcal{O}\big(\epsilon_{0}(\kappa,\kappa_{p})^{2}\big)$.
{\color{blue}This error term
reveals the following trade-offs: (i) increasing $\kappa$ leads to a larger $h(\kappa,\kappa_p)$, reducing the approximation error at the cost of increased communication overhead; (ii) decreasing $\kappa_p$ also leads to a larger $h(\kappa,\kappa_p)$, reducing the approximation error but limiting the ability of agents to capture inter-agent dependencies. Therefore, $\kappa$ and $\kappa_p$ should be chosen to balance communication overhead and approximation accuracy in practice.}
{\color{blue}We remark that Theorem~\ref{thetheoremconvergenceofpolicygradient} only
establishes stationary with respect to the policy parameter $\bm{\theta}$. In general,
convergence of a primal-dual algorithm to a stationary point of the Lagrangian function does not
imply that the resulting policy is feasible. Establishing a rigorous convergence rate for
constraint violation in our distributed, coupled-policy setting therefore remains an
important direction for future research. In Section~\ref{SectionSimulations}, we validate
constraint satisfaction empirically.}

\section{Simulation}\label{SectionSimulations}
{\color{blue}In this section, we study the empirical performance of the proposed Algorithm~\ref{distributedscalableprimaldualAlgorithm} on a wireless access-control network with safety constraints~\cite{YingNIPS2024}.

\subsection{Wireless access-control network environment}
We consider a cooperative task involving $\mathcal{N}=\{1,\cdots,25\}$ transmitting
agents within a $5\times5$ wireless access-control network, which
comprises $16$ access points $Y=\{1,\cdots,16\}$, as illustrated in Fig.~\ref{fig:wireless_env}.
\begin{figure}[htbp]
    \centering
    \includegraphics[width=0.3\textwidth]{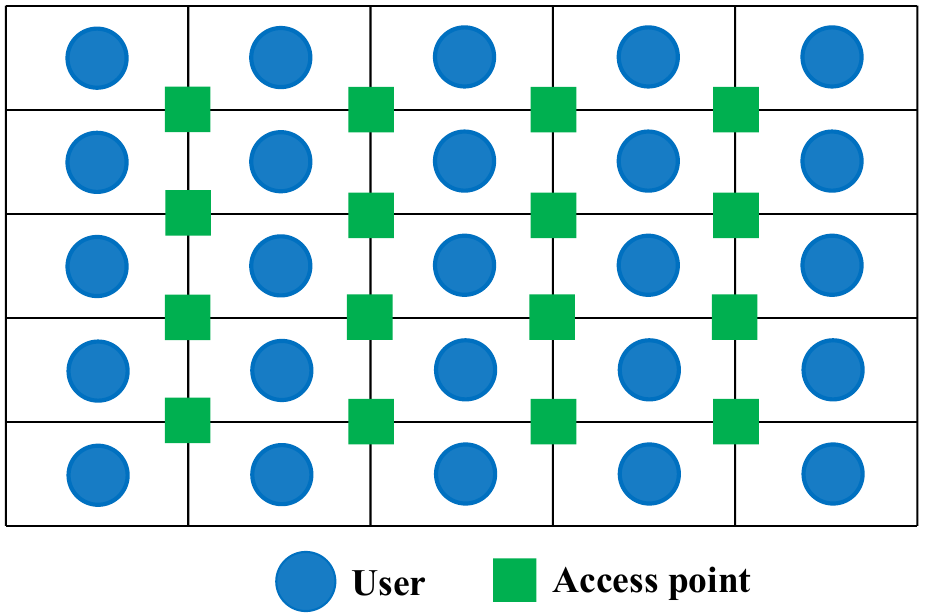}
    \caption{{\color{blue}Wireless access-control network.}}
    \label{fig:wireless_env}
\end{figure}
\par
Each agent $i\in\mathcal{N}$ is connected to a set $Y_i\subset Y$ of access points located at the corner of the block it resides in, and two agents are direct neighbors if they share a common access point, which naturally induces the environmental network
$\mathcal{G}^E$.
In every time slot, agent $i$ receives a new packet by deadline $d^{s}_{i}>0$ with probability $p^{new}_i\in(0,1)$, and can choose to send the earliest packet in its queue to an access
point $y\in Y_i$ or do nothing.
Agent $i$ receives a unit objective reward if and only if access point $y$ does not receive transmissions from other agents and successfully processes the packet, which occurs with probability $q_y\in(0,1)$.
Otherwise, the objective reward is zero.
In this environment, the state of each agent $i$ is a $d^{s}_i$-dimensional binary vector $s_i\in\mathcal{S}_{i}=\{0,1\}^{d^{s}_i}$, where the $k$-th entry takes the value $1$ when agent $i$ currently has a packet with $k$ time steps remaining until the deadline.
The action $a_i\in\mathcal{A}_i=Y_i\cup\{\mathrm{null}\}$ of agent $i$ specifies whether to send the earliest queued packet to an access point $y\in Y_i$ or to do nothing
(i.e., $a_i=\mathrm{null}$).
The safety constraint is a cumulative transmission budget with per-step constraint reward
$g_i(s_i,a_i)=-\mathds{1}_{\{a_i\neq\mathrm{null}\}}$.
Since the local transition dynamics and reward of each agent depend on the states and actions of its neighboring agents, the
problem exhibits strongly coupled state-action dynamics.

\subsection{Parameters setting}
The coupled policy $\pi_{i}(\cdot|s_{\mathcal{N}^{E,\kappa_{p}}_{i}}, \theta_{i}, \theta_{\mathcal{N}^{E,\kappa_{p}}_{i,-i}})$ of agent $i$ in the proposed Algorithm~\ref{distributedscalableprimaldualAlgorithm} adopts the following softmax form to capture the inter-agent dependencies:
\begin{align}
&\pi_{i}(a_{i}|s_{\mathcal{N}^{E,\kappa_{p}}_{i}}, \theta_{i}, \theta_{\mathcal{N}^{E,\kappa_{p}}_{i,-i}})\notag\\
=&\frac{\exp{(0.9*\theta_{i,s_{i},a_{i}}+\frac{0.1}{|\mathcal{N}^{E,\kappa_{p}}_{i}|-1}\sum_{j\in\mathcal{N}^{\kappa_{p}}_{i,-i}}\theta_{j,s_{i},a_{i}})}}{\sum_{a'_{i}}\exp{(0.9*\theta_{i,s_{i},a'_{i}}+\frac{0.1}{|\mathcal{N}^{E,\kappa_{p}}_{i}|-1}\sum_{j\in\mathcal{N}^{\kappa_{p}}_{i,-i}}\theta_{j,s_{i},a'_{i}})}},\label{policyexpress}
\end{align}
where $\theta_{i,s_{i},a_{i}}\in\mathbb{R}$ is the policy parameter of agent $i$ corresponding to local state-action $(s_{i},a_{i})$, and $\theta_{i}$ is the aggregation of $\theta_{i,s_{i},a_{i}}$ across all state-action pair $(s_{i},a_{i})\in\mathcal{S}_{i}\times\mathcal{A}_{i}$.
The parameter setting in Algorithm~\ref{distributedscalableprimaldualAlgorithm} and environment are configured as follows: $p^{new}_{i}=0.5$, $q_{y}=0.8$, $d^{s}_{i}=2$, $\gamma=0.9$, $c_{i}=-3.56$ for all $i\in\mathcal{N}$, $\kappa=\kappa_p=1$, $K_\mu=4$, $K_\theta=1$, and $\mu_{max}=50$.
Since the transmission success of each agent is directly determined by whether its immediate neighboring agents transmit simultaneously, setting $\kappa_p=1$ is sufficient to capture the inter-agent dependencies arising from the coupled state-action dynamics in this environment.
The learning networks among agents switch between
different topologies $\{0\rightarrow1, 2\rightarrow3,\cdots, 22\rightarrow23, 24\rightarrow0\}$ and $\{1\rightarrow2, 3\rightarrow4,\cdots, 23\rightarrow24\}$.
Further implementation details and hyperparameter
settings are available in our open-source code at \url{https://github.com/EricDMW/DSPD_New.git}.

\subsection{Simulation results}
To verify the effectiveness of our Algorithm~\ref{distributedscalableprimaldualAlgorithm}, we compare it with the scalable primal-dual actor-critic (SPDAC) algorithm~\cite{YingNIPS2024} that focuses on parameter-independent policies, and the MAPPO-Lagrangian (MAPPO-L) algorithm~\cite{MAPPO-L2021}. All algorithms are trained
under 16 different random seeds, and the results are shown in
Figs.~\ref{fig:wireless_return}-\ref{fig:wireless_estimation}.

\begin{figure}[htbp]
    \centering
    \subfigure[Objective return]{\includegraphics[width=0.24\textwidth]
    {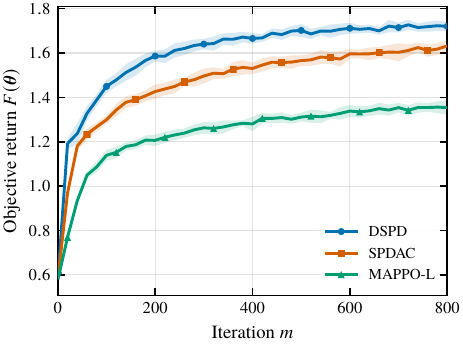}\label{fig:return_objective}}
    \subfigure[Average constraint return]{\includegraphics[width=0.24\textwidth]
    {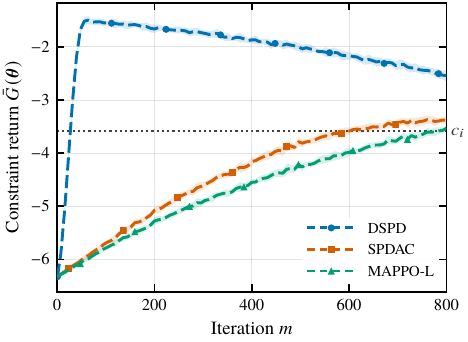}\label{fig:return_constraint}}
    \caption{{\color{blue}The evolution of
 the objective return and average constraint return of all algorithms.}}
    \label{fig:wireless_return}
\end{figure}

\par
Fig.~\ref{fig:return_objective} illustrates the evolution of the objective return
$F(\bm{\theta})$ of all algorithms, where the solid lines represent the mean values across different random seeds and the shaded areas denote the variance.
As observed, our Algorithm~\ref{distributedscalableprimaldualAlgorithm} achieves the highest objective return, clearly separated from SPDAC and MAPPO-L with non-overlapping bands.
This result aligns with our hypothesis that a coupled policy framework is capable of achieving superior learning performance compared to an independent policy structure, as it more effectively captures the inter-agent dependencies.

\par
The constraint performances of all algorithms are presented in
Fig.~\ref{fig:return_constraint}, where all algorithms are initialized over the transmission budget and start infeasible.
The proposed DSPD algorithm restores feasibility first
and achieves with the largest safety margin above constraint threshold value, demonstrating its superior
constraint satisfaction performance.

\begin{figure}[htbp]
    \centering
    \subfigure[$\frac{1}{N^{2}}\sum_{i,j=1}^{N}\|\hat{\theta}^{i}_{j,m}-
    \theta_{j,m}\|$]{\includegraphics[width=0.24\textwidth]
    {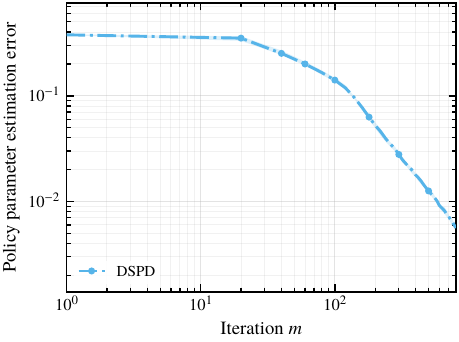}\label{fig:wireless_estimation_theta}}
    \subfigure[$\frac{1}{N^{2}}\sum_{i,j=1}^{N}|\hat{\mu}^{i}_{j,m}-
    \mu_{j,m}|$]{\includegraphics[width=0.24\textwidth]
    {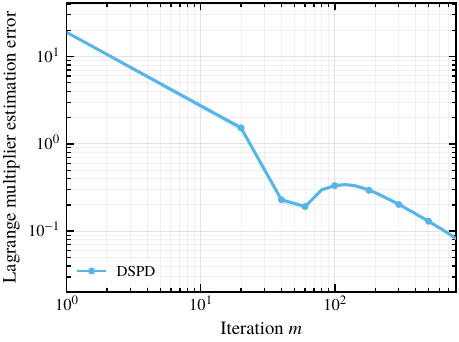}\label{fig:wireless_estimation_mu}}
    \caption{{\color{blue}The evolution of the estimation errors for policy parameters and Lagrange multipliers.}}
    \label{fig:wireless_estimation}
\end{figure}
\par
To further validate the accuracy of the proposed DSPD algorithm in parameter estimation,
Figs.~\ref{fig:wireless_estimation_theta}-\ref{fig:wireless_estimation_mu} present the evolution of the estimation errors for policy parameters and Lagrange
multipliers throughout the learning process.
As depicted, both
$\frac{1}{N^2}\sum_{i,j=1}^N\|\hat{\theta}^i_{j,m}-\theta_{j,m}\|$ and
$\frac{1}{N^2}\sum_{i,j=1}^N|\hat{\mu}^i_{j,m}-\mu_{j,m}|$ asymptotically converge
to zero, which aligns with the theoretical result established in
Theorem~\ref{Convergencetheoreminparameterestimation}.
}

{\color{blue}\subsection{Ablation experiment}
Note that in practice, we adopt the diminishing learning rate
$\eta_{\theta,m}\sim\mathcal{O}(\frac{\eta_{\theta}}{m})$ as required by Assumption~\ref{theassumptionoflearningrate}, where $\eta_\theta$ denotes the initial learning rate.
Fig.~\ref{fig:lr_ablation} presents the objective return $F(\bm{\theta})$ and
average constraint return $\bar{G}(\bm{\theta})$ under different initial learning
rates $\eta_\theta\in\{0.005, 0.02, 0.05, 0.15, 0.30\}$.
Both too small an initial rate (e.g., $\eta_\theta=0.005$) that converges slowly
to a lower plateau and too large an initial rate (e.g., $\eta_\theta=0.30$) that
leads to instability result in degraded objective performance.
An intermediate rate $\eta_\theta\in\{0.02, 0.05, 0.15\}$ achieves a good balance between convergence speed, objective performance, and constraint satisfaction.


\begin{figure}[htbp]
    \centering
    \subfigure[Objective return]{\includegraphics[width=0.24\textwidth]
    {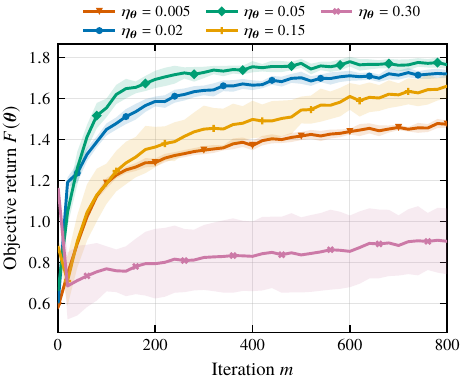}\label{fig:lr_ablation1}}
    \subfigure[Average constraint return]{\includegraphics[width=0.24\textwidth]
    {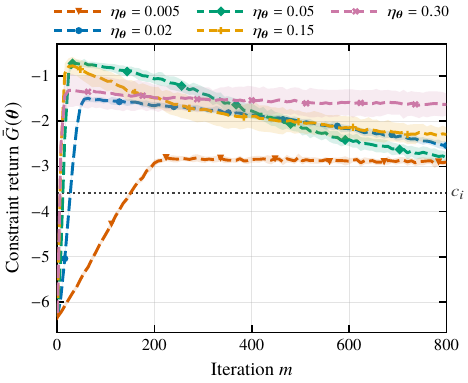}\label{fig:lr_ablation2}}
    \caption{{\color{blue}The evolution of
 the objective return and average constraint return of DSPD algorithm under different learning rates.}}
    \label{fig:lr_ablation}
\end{figure}
}


\section{Conclusions}\label{SectionVConclusions}
In this paper, we study the CMARL problem under coupled policies in a distributed setting without global environmental information. We propose a DSPD algorithm that operates solely based on a {\color{blue}prescribed} local neighborhood information, and establish theoretical guarantees showing that the proposed algorithm achieves an $\epsilon$-policy statinary convergence with an approximation error of $\mathcal{O}(\gamma^{\frac{\kappa+1}{\kappa_p}})$.
{\color{blue}Simulation results on a wireless access-control network validate the effectiveness of the proposed algorithm.}
\par
{\color{blue}In future research, several important directions remain open: (i) improving the convergence rate by designing more efficient learning rate schedules that simultaneously handle the coupled policy structure and distributed estimation
errors, with the goal of achieving a convergence rate comparable to that of existing independent policy-based algorithms; (ii) incorporating formal privacy mechanisms such as differential privacy into the time-varying learning network to enforce
rigorous privacy guarantees on transmitted parameters; (iii) extending the framework to handle truly lossy communication by replacing the uniformly strongly connected assumption with a probabilistic connectivity condition, and
(iv) extending the proposed framework to more complex settings
such as heterogeneous autonomous swarm environments with large-scale agents.}



\section{Analysis of Main Results}\label{AppendixProofofMainResults}
\subsection{Analysis of Theorem~\ref{thelemmaofapproximatedpolicygradient}}\label{AnalysisofTheoremthelemmaofapproximatedpolicygradient}
To facilitate a clear presentation of the proof for Theorem~\ref{thelemmaofapproximatedpolicygradient}, we introduce some definitions of auxiliary functions and the corresponding lemmas.
\par
Recalling the definitions of $Q^{\bm{\pi_{\theta}}}_{i}(\bm{s},\bm{a};\bm{\mu})$ in (\ref{thedefinitionoflocalQfunction}) and $h(\kappa,\kappa_{p})$ in (\ref{thefunctionofhkappap}), we establish the exponential decay property of CMARL under coupled policies as follows.
\begin{lemma}\label{thelemmaofexponentialdecayproperty}
The CMARL problem under coupled policies satisfies the $(\vartheta,\varrho)$-exponential decay property, i.e., for any joint policy $\bm{\pi_{\theta}}$, Lagrangian multiplier $\bm{\mu}$, agent $i\in\mathcal{N}$, $\kappa\geq0$,  $s_{\mathcal{N}^{E,\kappa}_{i}}\in\mathcal{S}_{\mathcal{N}^{E,\kappa}_{i}}$, $a_{\mathcal{N}^{E,\kappa}_{i}}\in\mathcal{A}_{\mathcal{N}^{E,\kappa}_{i}}$,
$s_{\mathcal{N}^{E,\kappa}_{-i}},s'_{\mathcal{N}^{E,\kappa}_{-i}}\in\mathcal{S}_{\mathcal{N}^{E,\kappa}_{-i}}$, and
$a_{\mathcal{N}^{E,\kappa}_{-i}},a'_{\mathcal{N}^{E,\kappa}_{-i}}\in\mathcal{A}_{\mathcal{N}^{E,\kappa}_{-i}}$,
the local Lagrangian $Q$-function $Q^{\bm{\pi_{\theta}}}_{i}(\bm{s},\bm{a};\bm{\mu})$ satisfies
\begin{align}\label{thesdefinitionofdecayproperty}
\Big|&Q^{\bm{\pi_{\theta}}}_{i}(s_{\mathcal{N}^{E,\kappa}_{i}},s_{\mathcal{N}^{E,\kappa}_{-i}},a_{\mathcal{N}^{E,\kappa}_{i}},a_{\mathcal{N}^{E,\kappa}_{-i}};\bm{\mu})\notag\\
&-Q^{\bm{\pi_{\theta}}}_{i}(s_{\mathcal{N}^{E,\kappa}_{i}},s'_{\mathcal{N}^{E,\kappa}_{-i}},a_{\mathcal{N}^{E,\kappa}_{i}},a'_{\mathcal{N}^{E,\kappa}_{-i}};\bm{\mu})\Big|\leq \vartheta\varrho^{\kappa+1},
\end{align}
where $\vartheta=\frac{2(R_{f}+\mu_{max}R_{g})}{(1-\gamma)\gamma^{\kappa-h(\kappa,\kappa_{p})}}$ and $\varrho=\gamma$.
\end{lemma}
\par
The detailed proof of Lemma~\ref{thelemmaofexponentialdecayproperty} is provided in
Section~\ref{ProofofLemmathelemmaofexponentialdecayproperty}.
Lemma~\ref{thelemmaofexponentialdecayproperty} shows that each agent $i$ can
approximate its local Lagrangian $Q$-function $Q^{\bm{\pi_{\theta}}}_{i}(\bm{s},
\bm{a};\bm{\mu})$ using only the state-action pairs $(s_{\mathcal{N}^{E,\kappa}_{i}},
a_{\mathcal{N}^{E,\kappa}_{i}})$ from its $\kappa$-hop neighbors, and the approximation
error decays exponentially as $\kappa$ increases. {\color{blue}While this result shares
a similar spirit with Lemma~3 in~\cite{QuCLDC2020}, the coupled policy structure
introduces a key difference: under independent policies in~\cite{QuCLDC2020}, the
decay exponent is $\kappa+1$, whereas under coupled policies, the policy coupling
accelerates the propagation of out-of-neighborhood differences through the
$\kappa_p$-hop neighborhood, yielding a tighter decay exponent $h(\kappa,\kappa_p)+1$.}
\par
Recalling the definition of $d^{\bm{\pi_{\theta}}}_{\bm{\rho}}(\bm{s})$ in (\ref{Thediscountedstatevisitationdistribution}),
we define $\xi^{\bm{\pi_{\theta}}}_{\bm{\rho}}(\bm{s},\bm{a})$ as the discounted state-action visitation distribution of $(\bm{s},\bm{a})\in\bm{\mathcal{S}}\times\bm{\mathcal{A}}$, which is given by
\begin{align}\label{thestationarydistributionofsa}
\xi^{\bm{\pi_{\theta}}}_{\bm{\rho}}(\bm{s},\bm{a})=d^{\bm{\pi_{\theta}}}_{\bm{\rho}}(\bm{s})\bm{\pi_{\theta}}(\bm{a}|\bm{s}).
\end{align}
Inspired by the exponential decay property in Lemma~\ref{thelemmaofexponentialdecayproperty} and the definition of $\xi^{\bm{\pi_{\theta}}}_{\bm{\rho}}(\bm{s},\bm{a})$ in (\ref{thestationarydistributionofsa}), we design a
class of truncated Lagrangian $Q$-functions as
\begin{align}\label{thetruncatedQfunctions}
&Q^{\bm{\pi_{\theta}}}_{tru,i}(s_{\mathcal{N}^{E,\kappa}_{i}},a_{\mathcal{N}^{E,\kappa}_{i}};\bm{\mu})\notag\\
=&\sum_{s_{\mathcal{N}^{E,\kappa}_{-i}},a_{\mathcal{N}^{E,\kappa}_{-i}}}\xi^{\bm{\pi_{\theta}}}_{\bm{\rho}}(s_{\mathcal{N}^{E,\kappa}_{-i}},a_{\mathcal{N}^{E,\kappa}_{-i}}|s_{\mathcal{N}^{E,\kappa}_{i}},a_{\mathcal{N}^{E,\kappa}_{i}})\notag\\
&\times Q^{\bm{\pi_{\theta}}}_{i}(s_{\mathcal{N}^{E,\kappa}_{i}},s_{\mathcal{N}^{E,\kappa}_{-i}},a_{\mathcal{N}^{E,\kappa}_{i}},a_{\mathcal{N}^{E,\kappa}_{-i}};\bm{\mu}),
\end{align}
where  $\xi^{\bm{\pi_{\theta}}}_{\bm{\rho}}(s_{\mathcal{N}^{E,\kappa}_{-i}},a_{\mathcal{N}^{E,\kappa}_{-i}}|s_{\mathcal{N}^{E,\kappa}_{i}},a_{\mathcal{N}^{E,\kappa}_{i}})$ denotes the weight coefficient and is represented as
\begin{align}\label{thestationarydistributionof-s-a}
&\xi^{\bm{\pi_{\theta}}}_{\bm{\rho}}(s_{\mathcal{N}^{E,\kappa}_{-i}},a_{\mathcal{N}^{E,\kappa}_{-i}}|s_{\mathcal{N}^{E,\kappa}_{i}},a_{\mathcal{N}^{E,\kappa}_{i}})\notag\\
=&\frac{\xi^{\bm{\pi_{\theta}}}_{\bm{\rho}}(s_{\mathcal{N}^{E,\kappa}_{i}},s_{\mathcal{N}^{E,\kappa}_{-i}},a_{\mathcal{N}^{E,\kappa}_{i}},a_{\mathcal{N}^{E,\kappa}_{-i}})}{\sum_{s'_{\mathcal{N}^{E,\kappa}_{-i}},a'_{\mathcal{N}^{E,\kappa}_{-i}}}\xi^{\bm{\pi_{\theta}}}_{\bm{\rho}}(s_{\mathcal{N}^{E,\kappa}_{i}},s'_{\mathcal{N}^{E,\kappa}_{-i}},a_{\mathcal{N}^{E,\kappa}_{i}},a'_{\mathcal{N}^{E,\kappa}_{-i}})}.
\end{align}
By definition of $\xi^{\bm{\pi_{\theta}}}_{\bm{\rho}}(s_{\mathcal{N}^{E,\kappa}_{-i}},a_{\mathcal{N}^{E,\kappa}_{-i}}|s_{\mathcal{N}^{E,\kappa}_{i}},a_{\mathcal{N}^{E,\kappa}_{i}})$ in (\ref{thestationarydistributionof-s-a}) and Assumption~\ref{theassumptionofdistributionofs}, it is obvious that the weight coefficient $\xi^{\bm{\pi_{\theta}}}_{\bm{\rho}}(s_{\mathcal{N}^{E,\kappa}_{-i}},a_{\mathcal{N}^{E,\kappa}_{-i}}|s_{\mathcal{N}^{E,\kappa}_{i}},a_{\mathcal{N}^{E,\kappa}_{i}})$ is non-negative and satisfies
\begin{align}\label{thesumofweight}
\sum_{s'_{\mathcal{N}^{E,\kappa}_{-i}},a'_{\mathcal{N}^{E,\kappa}_{-i}}}\!\!\!\!\xi^{\bm{\pi_{\theta}}}_{\bm{\rho}}(s'_{\mathcal{N}^{E,\kappa}_{-i}},a'_{\mathcal{N}^{E,\kappa}_{-i}}|s_{\mathcal{N}^{E,\kappa}_{i}},a_{\mathcal{N}^{E,\kappa}_{i}})=1.
\end{align}
Based on Lemma~\ref{thelemmaofexponentialdecayproperty} and (\ref{thesumofweight}), we further obtain
\begin{align}\label{thepropertyoftruncatedlagrangianQfunction}
&\Big|Q^{\bm{\pi_{\theta}}}_{tru,i}(s_{\mathcal{N}^{E,\kappa}_{i}},a_{\mathcal{N}^{E,\kappa}_{i}};\bm{\mu})-Q^{\bm{\pi_{\theta}}}_{i}(\bm{s},\bm{a};\bm{\mu})\Big|\notag\\
\leq&\frac{2(R_{f}+\mu_{max}R_{g})}{1-\gamma}\gamma^{h(\kappa,\kappa_{p})+1}.
\end{align}
\par
Based on the definition of $Q^{\bm{\pi_{\theta}}}_{tru,i}(s_{\mathcal{N}^{E,\kappa}_{i}}, a_{\mathcal{N}^{E,\kappa}_{i}}; \bm{\mu})$ in (\ref{thetruncatedQfunctions}),
we can derive the following lemma.
\begin{lemma}\label{thelemmaoftruncatederror}
For any agent $i\in\mathcal{N}$, $l\in\mathcal{N}^{E,\kappa+2\kappa_{p}}_{-i}$, joint policy $\bm{\pi}_{\bm{\theta}}$, and Lagrangian multiplier $\bm{\mu}$, we have
\begin{align}\label{thenewlemmaequality}
&\mathbb{E}_{\bm{s}\sim d^{\bm{\pi_{\theta}}}_{\bm{\rho}},\bm{a}\sim\bm{\pi_{\theta}}}\Big[Q^{\bm{\pi_{\theta}}}_{tru,l}(s_{\mathcal{N}^{E,\kappa}_{l}},a_{\mathcal{N}^{E,\kappa}_{l}};\bm{\mu})\notag\\
&\times\sum_{j\in\mathcal{N}^{E,\kappa_{p}}_{i}}\nabla_{\theta_{i}}\log\pi_{j}(a_{j}|s_{\mathcal{N}^{E,\kappa_{p}}_{j}},\theta_{j},\theta_{\mathcal{N}^{E,\kappa_{p}}_{j,-j}})\Big]\notag\\
=&0.
\end{align}
\end{lemma}
\par
The detailed proof of Lemma~\ref{thelemmaoftruncatederror} can be found in Section~\ref{ProofofLemmathelemmaoftruncatederror}.
{\color{blue}Based on Lemma~\ref{thelemmaoftruncatederror}, the proof of Theorem~\ref{thelemmaofapproximatedpolicygradient} is presented as follows.}
\begin{proof}
By the definitions of $\widehat{Q^{\bm{\pi_{\theta}}}_{i}}(\bm{s},\bm{a};\bm{\mu})$ in (\ref{theaction-averagedQfunctionofagenti}) and  $g^{\bm{\pi}_{\bm{\theta}},\bm{\mu}}_{app,i}$ in (\ref{thepolicygradientapproxiamtionbynewfunction}), we directly have
\begin{align}\label{thepolicygradientapproxiamtionbynewfunctionrefine}
g^{\bm{\pi}_{\bm{\theta}},\bm{\mu}}_{app,i}=&\frac{1}{1-\gamma}\mathbb{E}_{\bm{s}\sim d^{\bm{\pi_{\theta}}}_{\bm{\rho}},\bm{a}\sim\bm{\pi_{\theta}}}\Big[\frac{1}{N}\sum_{l\in\mathcal{N}^{\kappa+2\kappa_{p}}_{i}}Q^{\bm{\pi_{\theta}}}_{l}(\bm{s},\bm{a};\bm{\mu})\notag\\
&\times\!\!\sum_{j\in\mathcal{N}^{E,\kappa_{p}}_{i}}\!\!\nabla_{\theta_{i}}\log\pi_{j}(a_{j}|s_{\mathcal{N}^{E,\kappa_{p}}_{j}},\theta_{j},\theta_{\mathcal{N}^{E,\kappa_{p}}_{j,-j}})\Big].
\end{align}
By using (\ref{thepolicygradientapproxiamtionbynewfunctionrefine}) and the definition of $\nabla_{\theta_{i}}\mathcal{L}(\bm{\theta},\bm{\mu})$ in (\ref{thecoupledpolicygradienttheorem}), we further have
\begin{align}
&\Big\|g^{\bm{\pi}_{\bm{\theta}},\bm{\mu}}_{app,i}-\nabla_{\theta_{i}}\mathcal{L}(\bm{\theta},\bm{\mu})\Big\|\notag\\
=&\Big\|\frac{1}{1-\gamma}\mathbb{E}_{\bm{s}\sim d^{\bm{\pi_{\theta}}}_{\bm{\rho}},\bm{a}\sim\bm{\pi_{\theta}}}\Big[\frac{1}{N}\sum_{l\in\mathcal{N}^{E,\kappa+2\kappa_{p}}_{-i}}Q^{\bm{\pi_{\theta}}}_{l}(\bm{s},\bm{a};\bm{\mu})\notag\\
&\times\sum_{j\in\mathcal{N}^{E,\kappa_{p}}_{i}}\nabla_{\theta_{i}}\log\pi_{j}(a_{j}|s_{\mathcal{N}^{E,\kappa_{p}}_{j}},\theta_{j},\theta_{\mathcal{N}^{E,\kappa_{p}}_{j,-j}})\Big]\Big\|\label{theerroroftruncatedpolicygradientapproximate2-1}\\
\leq&\Big\|\frac{1}{1-\gamma}\mathbb{E}_{\bm{s}\sim d^{\bm{\pi_{\theta}}}_{\bm{\rho}},\bm{a}\sim\bm{\pi_{\theta}}}\Big[\frac{1}{N}\sum_{l\in\mathcal{N}^{E,\kappa+2\kappa_{p}}_{-i}}\Big(Q^{\bm{\pi_{\theta}}}_{l}(\bm{s},\bm{a};\bm{\mu})\notag\\
&-Q^{\bm{\pi_{\theta}}}_{tru,l}(s_{\mathcal{N}^{E,\kappa}_{l}},a_{\mathcal{N}^{E,\kappa}_{l}};\bm{\mu})\Big)\notag\\
&\times\sum_{j\in\mathcal{N}^{E,\kappa_{p}}_{i}}\nabla_{\theta_{i}}\log\pi_{j}(a_{j}|s_{\mathcal{N}^{E,\kappa_{p}}_{j}},\theta_{j},\theta_{\mathcal{N}^{E,\kappa_{p}}_{j,-j}})\Big]\Big\|\notag\\
&+\Big\|\frac{1}{1-\gamma}\mathbb{E}_{\bm{s}\sim d^{\bm{\pi_{\theta}}}_{\bm{\rho}},\bm{a}\sim\bm{\pi_{\theta}}}\Big[\frac{1}{N}\!\!\!\!\sum_{l\in\mathcal{N}^{E,\kappa+2\kappa_{p}}_{-i}}\!\!\!\!Q^{\bm{\pi_{\theta}}}_{tru,l}(s_{\mathcal{N}^{E,\kappa}_{l}},a_{\mathcal{N}^{E,\kappa}_{l}};\bm{\mu})\notag\\
&\times\sum_{j\in\mathcal{N}^{E,\kappa_{p}}_{i}}\nabla_{\theta_{i}}\log\pi_{j}(a_{j}|s_{\mathcal{N}^{E,\kappa_{p}}_{j}},\theta_{j},\theta_{\mathcal{N}^{E,\kappa_{p}}_{j,-j}})\Big]\Big\|\label{theerroroftruncatedpolicygradientapproximate2-2}\\
\leq&\frac{2(R_{f}+\mu_{max}R_{g})BN}{(1-\gamma)^{2}}\gamma^{h(\kappa,\kappa_{p})+1},\label{theerroroftruncatedpolicygradientapproximate2-3}
\end{align}
where (\ref{theerroroftruncatedpolicygradientapproximate2-2}) follows from the triangle inequality of $\mathcal{L}_{2}$-norm, and the last inequality comes from Assumption~\ref{theassumptionofpolicy}, (\ref{thepropertyoftruncatedlagrangianQfunction}), and (\ref{thenewlemmaequality}).
\end{proof}

\subsection{Analysis of Theorem~\ref{Convergencetheoreminparameterestimation}}\label{proofofTheoremConvergencetheoreminparameterestimation}
{\color{blue}Before presenting the proof of Theorem~\ref{Convergencetheoreminparameterestimation}, we first introduce an important lemma that will be used throughout the proof.}
\begin{lemma}\cite{NedicTAC2015}\label{thelemmaoftime-varyingnetwork}
Suppose Assumption~\ref{theassumptiontime-varyingnetworks} holds.
For each integer $k\geq1$, there is a stochastic vector sequence $\{\phi_{t}\}_{t\geq1}$ such that for all $i,j\in\mathcal{N}$ and $t\geq k$,
\begin{align}
|W^{L}_{ij,t:k}-\phi_{i,t}|\leq M_{1}\lambda^{t-k},
\end{align}
where $M_{1}>0$, $\lambda\in(0,1)$, $W^{L}_{t:k}=W^{L}_{t}W^{L}_{t-1}\cdots W^{L}_{k}$, and $W^{L}_{ij,t:k}$ being the $i$-th row and the $j$-th column
element in $W^{L}_{t:k}$, and $\phi_{i,t}$ represents the $i$-th element in $\phi_{t}$.
\end{lemma}
\par
Lemma~\ref{thelemmaoftime-varyingnetwork} indicates that there exists $w_{max}>1$, such that $\|W^{L}_{t:k}\|\leq w_{max}$ for all $t\geq k$.
\par
{\color{blue}In what follows, we provide the detailed proof of Theorem~\ref{Convergencetheoreminparameterestimation}.}
\begin{proof}
For Case (i), by the definition of $\hat{h}^{\bm{\hat{\pi}}_{\bm{\hat{\theta}}_{m}},\bm{\hat{\mu}}_{m}}_{app,i}(k)$ in (\ref{theapproximatedgradientinalgorithmdesignmu}), we have that \begin{align}\label{theapproximatedpolicygradientmu1}
\|\hat{h}^{\bm{\hat{\pi}}_{\bm{\hat{\theta}}_{m}},\bm{\hat{\mu}}_{m}}_{app,i}(k)\|=&\Big\|\frac{1}{N}\Big(\sum_{t=0}^{T_{1}(k)}\gamma^{t/2}g_{i,t}-c_{i}\Big)\Big\|\notag\\
\leq&\frac{R_{g}+c_{max}}{(1-\gamma^{1/2})N}=L^{\mu}_{1},
\end{align}
where the inequality comes from Assumption~\ref{theassumptionofreward} and holds for all $m\geq1$.
\par
Next, we will demonstrate the convergence of $\hat{\mu}^{i}_{j,m}$. For the convenience of expression, we define $\breve{\mu}_{j,m}=\big(\breve{\mu}^{1}_{j,m},\cdots,\breve{\mu}^{N}_{j,m}\big)^{\top}\in\mathbb{R}^{N}$,
based on (\ref{thekeyupdateofpolicyparametermu-1}), the update of $\breve{\mu}_{j,m}$ can be written as
\begin{align}\label{thepolicyparameteriterationinm}
\breve{\mu}_{j,m+1}=W^{L}_{m}\Big(\breve{\mu}_{j,m}+N\big(\mu_{j,m+1}-\mu_{j,m}\big)e_{j}\Big),
\end{align}
where $e_{j}$ is the unit vector with the $j$-th element is 1 and other elements are 0.
Since $W^{L}_{m}$ is the column stochastic matrix, we multiply both sides of (\ref{thepolicyparameteriterationinm}) by $(1/N)\mathbf{1}^{\top}_{N}$ and have
\begin{align}\label{theaveragedpart}
\frac{1}{N}\mathbf{1}^{\top}_{N}\breve{\mu}_{j,m+1}=&\frac{1}{N}\mathbf{1}^{\top}_{N}\Big(\breve{\mu}_{j,m}+N\big(\mu_{j,m+1}-\mu_{j,m}\big)e_{j}\Big)\notag\\
=&\frac{1}{N}\mathbf{1}^{\top}_{N}\breve{\mu}_{j,m}+\big(\mu_{j,m+1}-\mu_{j,m}\big).
\end{align}
From (\ref{theaveragedpart}), we further have
\begin{align}\label{theaveragedpart2}
\frac{1}{N}\mathbf{1}^{\top}_{N}\breve{\mu}_{j,m+1}-\mu_{j,m+1}=\frac{1}{N}\mathbf{1}^{\top}_{N}\breve{\mu}_{j,m}-\mu_{j,m}=0,
\end{align}
where the last equality can be obtained from the initial setting of $\breve{\mu}^{i}_{j,1}=\mu_{j,1}=0$ on Line 2 of Algorithm~\ref{distributedscalableprimaldualAlgorithm}.
\par
Denote $\delta_{m+1}=W^{L}_{m}N\big(\mu_{j,m+1}-\mu_{j,m}\big)e_{j}$, (\ref{thepolicyparameteriterationinm}) can be further written as
\begin{align}
\breve{\mu}_{j,m+1}
=&W^{L}_{m}\breve{\mu}_{j,m}+\delta_{m+1}\notag\\
=&W^{L}_{m:1}\breve{\mu}_{j,1}+\sum_{k=2}^{m}W^{L}_{m:k}\delta_{k}+\delta_{m+1},\label{theconsensuspart1}
\end{align}
where $W^{L}_{m:k}=W^{L}_{m}W^{L}_{m-1}\cdots W^{L}_{k},\forall m\geq k$.
Multiplying $W^{L}_{m+1}$ on both sides of (\ref{theconsensuspart1}), we have
\begin{align}\label{theconsensuspart2}
W^{L}_{m+1}\breve{\mu}_{j,m+1}=W^{L}_{m+1:1}\breve{\mu}_{j,1}+\sum_{k=2}^{m+1}W^{L}_{m+1:k}\delta_{k}.
\end{align}
Given that both $W^{L}_{m}$ and $W^{L}_{m:k}$ are column stochastic matrices, we multiply both sides of (\ref{theconsensuspart2}) by $\mathbf{1}^{\top}_{N}$ and have
\begin{align}\label{theconsensuspart3}
\mathbf{1}^{\top}_{N}\breve{\mu}_{j,m+1}=\mathbf{1}^{\top}_{N}\breve{\mu}_{j,1}+\sum_{k=2}^{m+1}\mathbf{1}^{\top}_{N}\delta_{k}.
\end{align}
Subtracting (\ref{theconsensuspart3}) multiplied by the vector $\phi_{m}\in\mathbb{R}^{N}$ in Lemma~\ref{thelemmaoftime-varyingnetwork} from (\ref{theconsensuspart2}), we obtain
\begin{align}\label{theconsensuspart4}
&(W^{L}_{m+1}-\phi_{m+1}\mathbf{1}^{\top}_{N})\breve{\mu}_{j,m+1}\notag\\
=&(W^{L}_{m+1:1}-\phi_{m+1}\mathbf{1}^{\top}_{N})\breve{\mu}_{j,1}+\sum_{k=2}^{m+1}(W^{L}_{m+1:k}-\phi_{m+1}\mathbf{1}^{\top}_{N})\delta_{k}.
\end{align}
Denote $\xi^{L}_{m+1:k}=W^{L}_{m+1:k}-\phi_{m+1}\mathbf{1}^{\top}_{N}\in\mathbb{R}^{N\times N}$, (\ref{theconsensuspart4}) can be written as
\begin{align}
W^{L}_{m+1}\breve{\mu}_{j,m+1}
=&\phi_{m+1}\mathbf{1}^{\top}_{N}\breve{\mu}_{j,m+1}+\xi^{L}_{m+1:1}\breve{\mu}_{j,1}\notag\\
&+\sum_{k=2}^{m+1}\xi^{L}_{m+1:k}\delta_{k}.\label{theorem2_appendix}
\end{align}
\par
Recalling the update (\ref{thekeyupdateofpolicyparameterp}) with $p_{i,1}=1$ for all $i\in\mathcal{N}$, we let $P_{m}=(p_{1,m},\cdots,p_{N,m})^{\top}\in\mathbb{R}^{N}$ and  directly have that $P_{m+1}=W^{L}_{m}P_{m}=W^{L}_{m:1}P_{1}$ and $\mathbf{1}^{\top}_{N}P_{m+1}=\mathbf{1}^{\top}_{N}P_{m}=\mathbf{1}^{\top}_{N}P_{1}=N$. Furthermore, we have
\begin{align}
P_{m+1}-\phi_{m}\mathbf{1}^{\top}_{N}P_{m+1}=(W^{L}_{m:1}-\phi_{m}\mathbf{1}^{\top}_{N})P_{1}\notag
\end{align}
and
\begin{align}\label{theexpressionofPt}
P_{m+1}=N\phi_{m}+\xi^{L}_{m:1}\mathbf{1}_{N}.
\end{align}
Let $\xi^{L}_{i,m:k}\in\mathbb{R}^{N}$ be the $i$-th row element of $\xi^{L}_{m:k}$.
By using (\ref{theorem2_appendix}) and (\ref{theexpressionofPt}),  (\ref{thekeyupdateofpolicyparametermu-3}) can derive
\begin{align}
&\Big|\hat{\mu}^{i}_{j,m}-\frac{1}{N}\mathbf{1}^{\top}_{N}\breve{\mu}_{j,m}\Big|\notag\\
=&\Bigg|\frac{\phi_{i,m}\mathbf{1}^{\top}_{N}\breve{\mu}_{j,m}+\xi^{L}_{i,m:1}\breve{\mu}_{j,1}+\sum_{k=2}^{m}\xi^{L}_{i,m:k}\delta_{k}}{N\phi_{i,m}+\xi^{L}_{i,m:1}\mathbf{1}_{N}}-\frac{\mathbf{1}^{\top}_{N}\breve{\mu}_{j,m}}{N}\Bigg|\notag\\
=&\Bigg|\frac{N\xi^{L}_{i,m:1}\breve{\mu}_{j,1}+N\sum_{k=2}^{m}\xi^{L}_{i,m:k}\delta_{k}-\xi^{L}_{i,m:1}\mathbf{1}_{N}\mathbf{1}^{\top}_{N}\breve{\mu}_{j,m}}{N(N\phi_{i,m}+\xi^{L}_{i,m:1}\mathbf{1}_{N})}\Bigg|\notag\\
\leq&\frac{\|\xi^{L}_{i,m:1}\|\|\breve{\mu}_{j,1}\|+\|\sum_{k=2}^{m}\xi^{L}_{i,m:k}\delta_{k}\|}{N\phi_{i,m}+\xi^{L}_{i,m:1}\mathbf{1}_{N}}\notag\\
&+\frac{\|\xi^{L}_{i,m:1}\mathbf{1}_{N}\|\|\mathbf{1}^{\top}_{N}\breve{\mu}_{j,m}\|}{N(N\phi_{i,m}+\xi^{L}_{i,m:1}\mathbf{1}_{N})}\label{theconsensuspart6-1}\\
\leq&\frac{\sqrt{N}\big(M_{1}\lambda^{m-1}\|\breve{\mu}_{j,1}\|+\sum_{k=2}^{m}M_{1}\lambda^{m-k}\|\delta_{k}\|\big)}{\zeta}\notag\\
&+\frac{M_{1}\lambda^{m-1}\|\breve{\mu}_{j,m}\|}{\zeta},\label{theconsensuspart6-2}
\end{align}
where the first inequality uses the fact that $|a^{\top}b|\leq\|a\|\|b\|,\forall a,b\in\mathbb{R}^{N}$ and the last inequality comes from Lemma~\ref{thelemmaoftime-varyingnetwork} and the observation that $N\phi_{i,m}+\xi^{L}_{i,m:1}\mathbf{1}_{N}\geq\zeta$, a result established in the proof of Lemma~1 in~\cite{NedicTAC2015}.
\par
For $\lambda^{m-1}\|\breve{\mu}_{j,m}\|$ in the right side of (\ref{theconsensuspart6-2}),
we can get
\begin{align}
\lambda^{m-1}\|\breve{\mu}_{j,m}\|
=&\lambda^{m-1}\Big\|W^{L}_{m-1:1}\breve{\mu}_{j,1}+\sum_{k=2}^{m-1}W^{L}_{m-1:k}\delta_{k}+\delta_{m}\Big\|\notag\\
\leq&\lambda^{m-1}\sum_{k=2}^{m-1}w_{max}\|\delta_{k}\|+\|\delta_{m}\|\label{theconsensuspart8-1}\\
\leq&w_{max}\sum_{k=2}^{m}\lambda^{m-k}\|\delta_{k}\|,\label{theconsensuspart8-2}
\end{align}
where the first inequality (\ref{theconsensuspart8-1}) can be  obtained by the initial setting of $\breve{\mu}^{i}_{j,1}=0,\forall i,j\in\mathcal{N}$ on Line 2 of Algorithm~\ref{distributedscalableprimaldualAlgorithm} and $\|W^{L}_{t:k}\|\leq w_{max}$ for all $t\geq k$ claimed after Lemma~\ref{thelemmaoftime-varyingnetwork},
and the last inequality comes from that $\lambda\in(0,1)$.
\par
For $\|\delta_{k}\|$ with $2\leq k\leq m$ in the right side of (\ref{theconsensuspart6-2}), by using $\delta_{k}=W^{L}_{k-1}N\big(\mu_{j,k}-\mu_{j,k-1}\big)e_{j}$, we have
\begin{align}
\|\delta_{k}\|=&\|W^{L}_{k-1}N\big(\mu_{j,k}-\mu_{j,k-1}\big)e_{j}\|\notag\\
\leq&N\|W^{L}_{k-1}\|\|\mu_{j,k}-\mu_{j,k-1}\|\notag\\
=&N\|W^{L}_{k-1}\|\|P_{\mathcal{U}_{j}}(\mu_{j,k-1}-\eta_{\mu,k-1}\bar{h}^{\bm{\hat{\pi}}_{\bm{\hat{\theta}}_{k-1}},\bm{\hat{\mu}}_{k-1}}_{app,j})\notag\\
&-\mu_{j,k-1}\|\label{theconsensuspart7-1}\\
\leq& w_{max}N\eta_{\mu,k-1}\|\bar{h}^{\bm{\hat{\pi}}_{\bm{\hat{\theta}}_{k-1}},\bm{\hat{\mu}}_{k-1}}_{app,j}\|\label{theconsensuspart7-2}\\
\leq&w_{max}L^{\mu}_{1}N\eta_{\mu,k-1},\label{theconsensuspart7}
\end{align}
where (\ref{theconsensuspart7-1}) uses the update (\ref{theupdateoftruepolicyparametersmu}), (\ref{theconsensuspart7-2}) can be obtained by the projection inequality,
and the last inequality (\ref{theconsensuspart7}) can be obtained by (\ref{theapproximatedpolicygradientmu1}) and (\ref{theapproximatedgradientinalgorithmdesignmubar}).
\par
Substituting (\ref{theconsensuspart8-2}), (\ref{theconsensuspart7}), and the fact that $\breve{\mu}^{i}_{j,1}=0,\forall i,j\in\mathcal{N}$ into (\ref{theconsensuspart6-2}), we can have
\begin{align}
&\Big|\hat{\mu}^{i}_{j,m}-\mu_{j,m}\Big|\notag\\
\leq&\frac{(\sqrt{N}+w_{max})M_{1}}{\zeta}\sum_{k=2}^{m}\lambda^{m-k}\|\delta_{k}\|\notag\\
\leq&\frac{(\sqrt{N}+w_{max})w_{max}L^{\mu}_{1}M_{1}N}{\zeta}\sum_{k=2}^{m}\lambda^{m-k}\eta_{\mu,k-1}.\label{theconsensuspart9-1}
\end{align}
Taking limits on $m$ on both sides of (\ref{theconsensuspart9-1}), we further have
\begin{align}
&\lim_{m\rightarrow\infty}|\hat{\mu}^{i}_{j,m}-\mu_{j,m}|\notag\\
\leq&\lim_{m\rightarrow\infty}\frac{(\sqrt{N}+w_{max})w_{max}L^{\mu}_{1}M_{1}N}{\zeta}\sum_{k=2}^{m}\lambda^{m-k}\eta_{\mu,k-1},\label{thelastinequalityofconsensus2}\\
=&0\label{thelastinequalityofconsensus}
\end{align}
with rate {\color{blue}$\mathcal{O}\big(\frac{1}{m}\big)$},
where the last inequality can be achieved by $\eta_{\mu,m}\sim{\color{blue}\mathcal{O}\big(\frac{1}{m}}\big)$ in Assumption~\ref{theassumptionoflearningrate} and Lemma~3.1 in~\cite{Ram2010}.
Hence, we can prove the Case (i).
\par
For Case (ii), by the definition of $\hat{g}^{\bm{\hat{\pi}}_{\bm{\hat{\theta}}_{m}},\bm{\hat{\mu}}_{m+1}}_{app,i}(k)$ in (\ref{theapproximatedgradientinalgorithmdesign}), we have
\begin{align}
&\|\hat{g}^{\bm{\hat{\pi}}_{\bm{\hat{\theta}}_{m}},\bm{\hat{\mu}}_{m+1}}_{app,i}(k)\|\notag\\
=&\frac{1}{1-\gamma}\Big\|\hat{Q}^{\bm{\hat{\pi}}_{\bm{\hat{\theta}}_{m}},\bm{\hat{\mu}}_{m+1}}_{i,T_{2}(k)}\sum_{j\in\mathcal{N}^{E,\kappa_{p}}_{i}}\nabla_{\theta_{i}}\log\pi_{j}(a_{j,T_{2}(k)}|s_{\mathcal{N}^{E,\kappa_{p}}_{j},T_{2}(k)},\notag\\
&\hat{\theta}^{i}_{j,m},\hat{\theta}^{i}_{\mathcal{N}^{E,\kappa_{p}}_{j,-j},m})\Big\|\notag\\
\leq&\frac{1}{1-\gamma}\Big|\frac{1}{N}\sum_{t=0}^{T_{3}(k)}\gamma^{t/2}\sum_{j\in\mathcal{N}^{E,\kappa+2\kappa_{p}}_{i}}(f_{j,T_{2}(k)+t}+\hat{\mu}^{i}_{j,m+1}g_{j,T_{2}(k)+t})\Big|\notag\\
&\times\sum_{j\in\mathcal{N}^{E,\kappa_{p}}_{i}}\big\|\nabla_{\theta_{i}}\log\pi_{j}(a_{j,T_{2}(k)}|s_{\mathcal{N}^{E,\kappa_{p}}_{j},T_{2}(k)},\hat{\theta}^{i}_{j,m},\hat{\theta}^{i}_{\mathcal{N}^{E,\kappa_{p}}_{j,-j},m})\big\|,\label{theapproximatedpolicygradient1}
\end{align}
where the last inequality (\ref{theapproximatedpolicygradient1}) can be obtained by the definition of $\hat{Q}^{\bm{\hat{\pi}}_{\bm{\hat{\theta}}_{m}},\bm{\hat{\mu}}_{m+1}}_{i,T_{2}(k)}$ in (\ref{theestimateofwildehatQi}) and triangular inequality.
\par
Since (\ref{theapproximatedpolicygradientmu1}) holds for any joint policy, the conclusion
$\lim_{m\rightarrow\infty}\hat{\mu}^{i}_{j,m}=\mu_{j,m}$ can be consistently upheld throughout the learning process.
Based on this fact and $\mu_{j,m}\in\mathcal{U}_{j}=[0,\mu_{max}]$, we have that $\hat{\mu}^{i}_{j,m}$ is also bounded, i.e., there exists $\tilde{\mu}_{max}\geq\mu_{max}>0$ such that
\begin{align}
|\hat{\mu}^{i}_{j,m}|\leq\tilde{\mu}_{max}, \forall i,j\in\mathcal{N}, m\geq1.\label{theboundofhatmuij}
\end{align}
Substituting (\ref{theboundofhatmuij}) into (\ref{theapproximatedpolicygradient1}), we have
\begin{align}
\|\hat{g}^{\bm{\hat{\pi}}_{\bm{\hat{\theta}}_{m}},\bm{\hat{\mu}}_{m+1}}_{app,i}(k)\|
\leq&\frac{BN(R_{f}+\tilde{\mu}_{max}R_{g})}{(1-\gamma)(1-\gamma^{1/2})}=L^{\theta}_{1},\label{theapproximatedpolicygradient2}
\end{align}
where the inequality (\ref{theapproximatedpolicygradient2}) can be achieved by Assumptions~\ref{theassumptionofreward}-\ref{theassumptionofpolicy}.
Similar to the proof of Case (i), by employing (\ref{theapproximatedpolicygradient2}) and $\eta_{\theta,m}\sim{\color{blue}\mathcal{O}\big(\frac{1}{\sqrt{m}}}\big)$ in Assumption~\ref{theassumptionoflearningrate}, we can obtain $\lim_{m\rightarrow\infty}\hat{\theta}^{i}_{j,m}=\theta_{j,m}$ with rate $\mathcal{O}(\frac{1}{\sqrt{m}})$ for all $i,j\in\mathcal{N}$.
Given that the proof procedure is strikingly analogous to that of Case (i), it is omitted for brevity.
\end{proof}

\subsection{Analysis of Theorem~\ref{thetheoremconvergenceofpolicygradient}}\label{AnalysisofTheoremthetheoremconvergenceofpolicygradient}
To clearly present the progress of the proof for Theorem~\ref{thetheoremconvergenceofpolicygradient}, we introduce several key definitions, along with intermediate lemmas and corollaries.
\par
Recalling $Q^{\bm{\pi_{\theta}}}(\bm{s},\bm{a};\bm{\mu})$ in (\ref{thedefinitionofglobalQfunction}), $Q^{\bm{\pi_{\theta}}}_{g,i}(\bm{s},\bm{a})$ in (\ref{thedefinitionoflocalQfunctionconstraint}), and $\widehat{Q^{\bm{\pi_{\theta}}}_{i}}(\bm{s},\bm{a};\bm{\mu})$ in (\ref{theaction-averagedQfunctionofagenti}),
we define $Q^{\bm{\hat{\pi}}_{\bm{\hat{\theta}}_{m}}}(\bm{s},\bm{a};\bm{\hat{\mu}}^{i}_{m+1})$ and $\widehat{Q^{\bm{\hat{\pi}}_{\bm{\hat{\theta}}_{m}}}_{i}}(\bm{s},\bm{a};\bm{\hat{\mu}}^{i}_{m+1})$ as the global Lagrangian $Q$-function and the neighbors' averaged Lagrangian $Q$-function, respectively, associated with the joint policy $\bm{\hat{\pi}}_{\bm{\hat{\theta}}_{m}}$ and the Lagrange multiplier $\bm{\hat{\mu}}^{i}_{m+1}$.
Moreover, we denote by $Q^{\bm{\hat{\pi}}_{\bm{\hat{\theta}}_{m}}}_{g,i}(\bm{s},\bm{a})$ the local $Q$-function corresponding to the constraint reward under the same joint policy.
\par
On the basis of these definitions, we define
$\nabla_{\mu_{i}}\tilde{\mathcal{L}}(\bm{\hat{\theta}}_{m},\bm{\hat{\mu}}_{m})$ as the Lagrangian gradient and $\nabla_{\theta_{i}}\tilde{\mathcal{L}}(\bm{\hat{\theta}}_{m},\bm{\hat{\mu}}_{m+1})$ as the policy gradient of agent $i$ under the executed joint policy $\bm{\hat{\pi}}_{\bm{\hat{\theta}}_{m}}$, Lagrangian multipliers $\bm{\hat{\mu}}_{m}$ and $\bm{\hat{\mu}}_{m+1}$, respectively.
%
They are described as follows:
\begin{align}\label{thepolicygradientapproxiamtionbynewfunctionmu2}
\nabla_{\mu_{i}}\tilde{\mathcal{L}}(\bm{\hat{\theta}}_{m},\bm{\hat{\mu}}_{m})=&\frac{1}{N}\Big(\mathbb{E}_{\bm{s}\sim\bm{\rho},\bm{a}\sim\bm{\hat{\pi}}_{\bm{\hat{\theta}}_{m}}}\Big[Q^{\bm{\hat{\pi}}_{\bm{\hat{\theta}}_{m}}}_{g,i}(\bm{s},\bm{a})\Big]-c_{i}\Big)
\end{align}
and
\begin{align}
&\nabla_{\theta_{i}}\Tilde{\mathcal{L}}(\bm{\hat{\theta}}_{m},\bm{\hat{\mu}}_{m+1})\notag\\
=&\frac{1}{1-\gamma}\mathbb{E}_{\bm{s}\sim d^{\bm{\hat{\pi}}_{\bm{\hat{\theta}}_{m}}}_{\bm{\rho}},\bm{a}\sim\bm{\hat{\pi}}_{\bm{\hat{\theta}}_{m}}}\Big[Q^{\bm{\hat{\pi}}_{\bm{\hat{\theta}}_{m}}}(\bm{s},\bm{a};\bm{\hat{\mu}}^{i}_{m+1})\notag\\
&\times\sum_{j\in\mathcal{N}^{E,\kappa_{p}}_{i}}\nabla_{\theta_{i}}\log\pi_{j}(a_{j}|s_{\mathcal{N}^{E,\kappa}_{j}},\hat{\theta}^{i}_{j,m},\hat{\theta}^{i}_{\mathcal{N}^{E,\kappa_{p}}_{j,-j},m})\Big].\label{thepolicygradientofexecutionpolicy} \end{align}
\par
Similar to the approximated policy gradient $g^{\bm{\pi}_{\bm{\theta}},\bm{\mu}}_{app,i}$ in (\ref{thepolicygradientapproxiamtionbynewfunction}),
we define $g^{\bm{\hat{\pi}}_{\bm{\hat{\theta}}_{m}},\bm{\hat{\mu}}_{m+1}}_{app,i}$ as the approximated coupled policy gradient of agent $i$ under the executed joint policy $\bm{\hat{\pi}}_{\bm{\hat{\theta}}_{m}}$ and  Lagrangian multiplier $\bm{\hat{\mu}}_{m+1}$, which is described as
\begin{align}\label{thepolicygradientapproxiamtionbynewfunction2}
g^{\bm{\hat{\pi}}_{\bm{\hat{\theta}}_{m}},\bm{\hat{\mu}}_{m+1}}_{app,i}=&\frac{1}{1-\gamma}\mathbb{E}_{\bm{s}\sim d^{\bm{\hat{\pi}}_{\bm{\hat{\theta}}_{m}}}_{\bm{\rho}},\bm{a}\sim\bm{\hat{\pi}}_{\bm{\hat{\theta}}_{m}}}\Big[\widehat{Q^{\bm{\hat{\pi}}_{\bm{\hat{\theta}}_{m}}}_{i}}(\bm{s},\bm{a};\bm{\hat{\mu}}^{i}_{m+1})\notag\\
&\times\sum_{j\in\mathcal{N}^{E,\kappa_{p}}_{i}}\!\!\nabla_{\theta_{i}}\log\pi_{j}(a_{j}|s_{\mathcal{N}^{E,\kappa_{p}}_{j}},\hat{\theta}^{i}_{j,m},\hat{\theta}^{i}_{\mathcal{N}^{E,\kappa_{p}}_{j,-j},m})\Big].
\end{align}
\par
Having introduced the auxiliary definitions (\ref{thepolicygradientapproxiamtionbynewfunctionmu2})-(\ref{thepolicygradientapproxiamtionbynewfunction2}), we now present several key lemmas and corollaries that form the theoretical foundation of the analysis of Theorem~\ref{thetheoremconvergenceofpolicygradient}.
\begin{lemma}[Unbiased gradient estimates for (\ref{theapproximatedgradientinalgorithmdesignmu}) and  (\ref{theapproximatedgradientinalgorithmdesign})]\label{theunbiasedestimationlemma}

Suppose Assumptions~\ref{theassumptionofFOSP}-\ref{theassumptionoflearningrate} hold.
In each $m$-th iteration, the Lagrangian gradient estimation $\hat{h}^{\bm{\hat{\pi}}_{\bm{\hat{\theta}}_{m}},\bm{\hat{\mu}}_{m}}_{app,i}(k)$  and the approximated policy gradient estimation $\hat{g}^{\bm{\hat{\pi}}_{\bm{\hat{\theta}}_{m}},\bm{\hat{\mu}}_{m+1}}_{app,i}(k)$ satisfy
\par
(i) $\hat{h}^{\bm{\hat{\pi}}_{\bm{\hat{\theta}}_{m}},\bm{\hat{\mu}}_{m}}_{app,i}(k)$ is an unbiased estimate of $\nabla_{\mu_{i}}\tilde{\mathcal{L}}(\bm{\hat{\theta}}_{m},\bm{\hat{\mu}}_{m})$, i.e., $\mathbb{E}_{T_{1}(k)}\Big[\hat{h}^{\hat{\bm{\pi}}_{\bm{\hat{\theta}}_{m}},\bm{\hat{\mu}}_{m}}_{app,i}(k)\Big|\bm{s}_{0}\sim\bm{\rho},\bm{\hat{\pi}}_{\bm{\hat{\theta}}_{m}}\Big]=\nabla_{\mu_{i}}\tilde{\mathcal{L}}(\bm{\hat{\theta}}_{m},\bm{\hat{\mu}}_{m})$ for all $k\in\{1,\cdots,K_{\mu}\}$;
\par
(ii) $\hat{g}^{\bm{\hat{\pi}}_{\bm{\hat{\theta}}_{m}},\bm{\hat{\mu}}_{m+1}}_{app,i}(k)$ is an unbiased estimate of $g^{\bm{\hat{\pi}}_{\bm{\hat{\theta}}_{m}},\bm{\hat{\mu}}_{m+1}}_{app,i}$, i.e., $\mathbb{E}_{T_{2}(k),T_{3}(k)}\Big[\hat{g}^{\bm{\hat{\pi}}_{\bm{\hat{\theta}}_{m}},\bm{\hat{\mu}}_{m+1}}_{app,i}(k)\Big|\bm{s}_{0}\sim\bm{\rho},\bm{\hat{\pi}}_{\bm{\hat{\theta}}_{m}}\Big]=g^{\bm{\hat{\pi}}_{\bm{\hat{\theta}}_{m}},\bm{\hat{\mu}}_{m+1}}_{app,i}$ for all $k\in\{1,\cdots,K_{\theta}\}$.
\end{lemma}
\par
The proof of Lemma~\ref{theunbiasedestimationlemma} is provided
in Section~\ref{ProofofLemmatheunbiasedestimationlemma}.
Lemma~\ref{theunbiasedestimationlemma} establishes that each agent $i$ can achieve unbiased estimates of the Lagrangian gradient $\nabla_{\mu_{i}}\tilde{\mathcal{L}}(\bm{\hat{\theta}}_{m},\bm{\hat{\mu}}_{m})$ and the approximated policy gradient $g^{\bm{\hat{\pi}}_{\bm{\hat{\theta}}_{m}},\bm{\hat{\mu}}_{m+1}}_{app,i}$ through samples $\{\bm{s},\bm{a}\}_{0:T_{1}(k)}$ and $\{\bm{s},\bm{a}\}_{0:T_{2}(k)+T_{3}(k)}$, respectively.
From Lemma~\ref{theunbiasedestimationlemma}, the following corollary immediately follows.
\begin{corollary}\label{theunbiasedestimationcorollary}
Suppose Assumptions~\ref{theassumptionofFOSP}-\ref{theassumptionoflearningrate} hold.
In each $m$-th iteration, for any $\delta\in(0,1)$, the following inequalities hold with probability at least $1-\delta$:
\begin{align}\label{theunbiasedestimationcorollary1}
\big\|\bar{h}^{\bm{\hat{\pi}}_{\bm{\hat{\theta}}_{m}},\bm{\hat{\mu}}_{m}}_{app,i}-\nabla_{\mu_{i}}\tilde{\mathcal{L}}(\bm{\hat{\theta}}_{m},\bm{\hat{\mu}}_{m})\big\|\leq2L^{\mu}_{1}\sqrt{\frac{\log{(2/\delta)}}{2K_{\mu}}}
\end{align}
and
\begin{align}\label{theunbiasedestimationcorollary2}
\big\|\bar{g}^{\bm{\hat{\pi}}_{\bm{\hat{\theta}}_{m}},\bm{\hat{\mu}}_{m+1}}_{app,i}-g^{\bm{\hat{\pi}}_{\bm{\hat{\theta}}_{m}},\bm{\hat{\mu}}_{m+1}}_{app,i}\big\|\leq2L^{\theta}_{1}\sqrt{\frac{\log{(2/\delta)}}{2K_{\theta}}}.
\end{align}
\end{corollary}
\begin{proof}
By Theorem~\ref{Convergencetheoreminparameterestimation}, Lemma~\ref{theunbiasedestimationlemma}, and Hoeffding inequality, (\ref{theunbiasedestimationcorollary1}) and (\ref{theunbiasedestimationcorollary2}) can be obtained directly.
\end{proof}
\par
Corollary~\ref{theunbiasedestimationcorollary} establishes the upper bound of error between $\bar{h}^{\bm{\hat{\pi}}_{\bm{\hat{\theta}}_{m}},\bm{\hat{\mu}}_{m}}_{app,i}$ in (\ref{theupdateoftruepolicyparametersmu}) and $\nabla_{\mu_{i}}\tilde{\mathcal{L}}(\bm{\hat{\theta}}_{m},\bm{\hat{\mu}}_{m})$ in (\ref{thepolicygradientapproxiamtionbynewfunctionmu2}), along with the error between $\bar{g}^{\bm{\hat{\pi}}_{\bm{\hat{\theta}}_{m}},\bm{\hat{\mu}}_{m+1}}_{app,i}$ in (\ref{theupdateoftruepolicyparameterstheta}) and  $g^{\bm{\hat{\pi}}_{\bm{\hat{\theta}}_{m}},\bm{\hat{\mu}}_{m+1}}_{app,i}$ in (\ref{thepolicygradientapproxiamtionbynewfunction2}).
\par
In order to further quantify the discrepancy between $\bar{h}^{\bm{\hat{\pi}}_{\bm{\hat{\theta}}_{m}},\bm{\hat{\mu}}_{m}}_{app,i}$  and the true Lagrangian gradient $\nabla_{\mu_{i}}\mathcal{L}(\bm{\theta}_{m},\bm{\mu}_{m})$ as well as the discrepancy between $\bar{g}^{\bm{\hat{\pi}}_{\bm{\hat{\theta}}_{m}},\bm{\hat{\mu}}_{m+1}}_{app,i}$ and true policy gradient $\nabla_{\theta_{i}}\mathcal{L}(\bm{\theta}_{m},\bm{\mu}_{m+1})$, we define $L^{\mu}_{2}=BR_{g}/(1-\gamma)^{2}$, $L^{\theta\mu}_{2}=B\sqrt{N}R_{g}/(1-\gamma)^{2}$, $L^{\theta\theta}_{2}=LN(R_{f}+\tilde{\mu}_{max}R_{g})/(1-\gamma)^{2}+(1+\gamma)B^{2}N^{2}(R_{f}+\tilde{\mu}_{max}R_{g})/(1-\gamma)^{3}$, and derive the following lemma.
\begin{lemma}\label{thelemmaofsmooth}
Suppose Assumptions~\ref{theassumptionofreward}-\ref{theassumptionofpolicy} hold.
For any agent $i\in\mathcal{N}$ and $m\geq1$, we have
\par
(i) $\big\|\nabla_{\mu_{i}}\mathcal{L}(\bm{\theta}_{m},\bm{\mu}_{m})-\nabla_{\mu_{i}}\tilde{\mathcal{L}}(\bm{\hat{\theta}}_{m},\bm{\hat{\mu}}_{m})\big\|\leq L^{\mu}_{2}\big\|\mathbf{1}_{N}\otimes\bm{\theta}_{m}-\bm{\hat{\theta}}_{m}\big\|$;
\par
(ii) $\big\|\nabla_{\theta_{i}}\mathcal{L}(\bm{\theta}_{m},\bm{\mu}_{m+1})-\nabla_{\theta_{i}}\tilde{\mathcal{L}}(\bm{\hat{\theta}}_{m},\bm{\hat{\mu}}_{m+1})\big\|\leq L^{\theta\mu}_{2}\big\|\bm{\mu}_{m+1}$ $-\bm{\hat{\mu}}^{i}_{m+1}\big\|+L^{\theta\theta}_{2}\big\|\mathbf{1}_{N}\otimes\bm{\theta}_{m}-\bm{\hat{\theta}}_{m}\big\|$.
\end{lemma}
\par
The detailed proof of Lemma~\ref{thelemmaofsmooth} is provided in Section~\ref{ProofofLemmathelemmaofsmooth}.
Lemma~\ref{thelemmaofsmooth} indicates that the upper bounds on the error between $\nabla_{\mu_{i}}\tilde{\mathcal{L}}(\bm{\hat{\theta}}_{m},\bm{\hat{\mu}}_{m})$ and the true Lagrangian gradient $\nabla_{\mu_{i}}\mathcal{L}(\bm{\theta}_{m},\bm{\mu}_{m})$ depends on the discrepancy between the estimated policy parameters $\{\bm{\hat{\theta}}^{i}_{m}\}_{i\in\mathcal{N}}$ and the true policy parameter $\bm{\theta}_{m}$.
In contrast, the upper bounds on the error between $\nabla_{\theta_{i}}\tilde{\mathcal{L}}(\bm{\hat{\theta}}_{m},\bm{\hat{\mu}}_{m+1})$ and the true policy gradient $\nabla_{\theta_{i}}\mathcal{L}(\bm{\theta}_{m},\bm{\mu}_{m+1})$ depend not only on the aforementioned discrepancy but also on the discrepancy between the estimated Lagrangian multipliers $\{\bm{\hat{\mu}}^{i}_{m+1}\}_{i\in\mathcal{N}}$ and the true Lagrangian multiplier $\bm{\mu}_{m+1}$.
\par
Define
$\epsilon_{\mu}(m,K_{\mu})=L^{\mu}_{2}\sqrt{N}\big\|\mathbf{1}_{N}\otimes\bm{\theta}_{m}-\bm{\hat{\theta}}_{m}\big\|+2L^{\mu}_{1}\sqrt{N\log{(2/\delta)}/2K_{\mu}}$, $\epsilon_{\theta}(m,K_{\theta})=L^{\theta\mu}_{2}\sqrt{N}\big\|\mathbf{1}_{N}\otimes\bm{\mu}_{m+1}-\bm{\hat{\mu}}_{m+1}\big\|+L^{\theta\theta}_{2}\sqrt{N}\big\|\mathbf{1}_{N}\otimes\bm{\theta}_{m}-\bm{\hat{\theta}}_{m}\big\|+2L^{\theta}_{1}\sqrt{N\log{(2/\delta)}/2K_{\theta}}$, and $\epsilon_{0}(\kappa,\kappa_{p})=2(R_{f}+\tilde{\mu}_{max}R_{g})BN^{\frac{3}{2}}\gamma^{h(\kappa,\kappa_{p})+1}/(1-\gamma)^{2}$.
By combining Theorem~\ref{thelemmaofapproximatedpolicygradient}, Corollary~\ref{theunbiasedestimationcorollary}, and Lemma~\ref{thelemmaofsmooth}, we have the following corollary.
\begin{corollary}\label{thecorollaryerrorbetweentrueandused}
Suppose Assumptions~\ref{theassumptionofreward}-\ref{theassumptionofpolicy} hold.
For any $i\in\mathcal{N}$ and $t\geq1$, we have
\par
(i) $\big\|\nabla_{\mu_{i}}\mathcal{L}(\bm{\theta}_{m},\bm{\mu}_{m})-\bar{h}^{\bm{\hat{\pi}}_{\bm{\hat{\theta}}_{m}},\bm{\hat{\mu}}_{m}}_{app,i}\big\|
\leq \epsilon_{\mu}(m,K_{\mu})/\sqrt{N}$;
\par
(ii) $\big\|\nabla_{\theta_{i}}\mathcal{L}(\bm{\theta}_{m},\bm{\mu}_{m+1})-\bar{g}^{\bm{\hat{\pi}}_{\bm{\hat{\theta}}_{m}},\bm{\hat{\mu}}_{m+1}}_{app,i}\big\|
\leq \big(\epsilon_{\theta}(m,K_{\theta})+\epsilon_{0}(\kappa,\kappa_{p})\big)/\sqrt{N}$.
\end{corollary}
\par
The detailed proof of Corollary~\ref{thecorollaryerrorbetweentrueandused} can be found in Section~\ref{ProofofCorollarythecorollaryerrorbetweentrueandused}.
Corollary~\ref{thecorollaryerrorbetweentrueandused} establishes the the upper bounds of the error between $\nabla_{\mu_{i}}\mathcal{L}(\bm{\theta}_{m},\bm{\mu}_{m})$ and $\bar{h}^{\bm{\hat{\pi}}_{\bm{\hat{\theta}}_{m}},\bm{\hat{\mu}}_{m}}_{app,i}$, as well as the upper bounds of the error between $\nabla_{\theta_{i}}\mathcal{L}(\bm{\theta}_{m},\bm{\mu}_{m+1})$ and $\bar{g}^{\bm{\hat{\pi}}_{\bm{\hat{\theta}}_{m}},\bm{\hat{\mu}}_{m+1}}_{app,i}$.
It is important to note that in Case (i), the upper bound of the error is determined by the discrepancy between $\{\bm{\hat{\theta}}^{i}_{m}\}_{i\in\mathcal{N}}$ and $\bm{\theta}_{m}$, as well as the number of sample batches $K_{\mu}$.
Conversely, in Case (ii), the upper bound of the error is influenced by the errors between $\{\bm{\hat{\theta}}^{i}_{m}\}_{i\in\mathcal{N}}$ and $\bm{\theta}_{m}$, the errors between $\{\bm{\hat{\mu}}^{i}_{m+1}\}_{i\in\mathcal{N}}$ and $\bm{\mu}_{m+1}$, the number of sample batches $K_{\mu}$, and the truncated distance $\kappa$.
\par
Recall the definition of $L_{\theta\theta}$ in Theorem~\ref{thetheoremconvergenceofpolicygradient}.
We denote $\nabla_{\bm{\theta}}\mathcal{L}(\bm{\theta},\bm{\mu})=\big(\nabla_{\theta_{1}}\mathcal{L}(\bm{\theta},\bm{\mu})^{\top},\cdots,\nabla_{\theta_{N}}\mathcal{L}(\bm{\theta},\bm{\mu})^{\top}\big)^{\top}$ and can derive the following corollary.
\begin{corollary}\label{thecorollaryofLipschitzcontinuousinequality}
Suppose Assumptions~\ref{theassumptionofreward}-\ref{theassumptionofpolicy} hold,
we have that
$\nabla_{\bm{\theta}}\mathcal{L}(\bm{\theta},\bm{\mu})$
is Lipschitz continuous with respect to $\bm{\theta}$, i.e.,
$\|\nabla_{\bm{\theta}}\mathcal{L}(\bm{\theta},\bm{\mu})-\nabla_{\bm{\theta}}\mathcal{L}(\bm{\theta}',\bm{\mu})\|\leq L_{\theta\theta}\|\bm{\theta}-\bm{\theta}'\|$ for any $\bm{\theta},\bm{\theta}'\in\bm{\Theta}$.
\end{corollary}
\begin{proof}
The proof of this corollary follows a similar process to that of Lemma~\ref{thelemmaofsmooth}.
Only note that the primary difference lies in the variation of coefficients.
For the sake of brevity, the detailed steps are omitted here.
\end{proof}
\par
Corollary~\ref{thecorollaryofLipschitzcontinuousinequality} shows the $L_{\theta\theta}$-smoothness of $\mathcal{L}(\bm{\theta},\bm{\mu})$ with respect to $\bm{\theta}$.
It is a direct extension of Lemma~\ref{thelemmaofsmooth} and serves as the foundation for analyzing the first-order stationary convergence of Algorithm~\ref{distributedscalableprimaldualAlgorithm}.
\par
Recall Theorem~\ref{Convergencetheoreminparameterestimation} states that for all $i,j \in \mathcal{N}$, as $m \to \infty$, $\hat{\mu}^{i}_{j,m}$ and $\hat{\theta}^{i}_{j,m}$ converge to $\mu_{j,m}$ and $\theta_{j,m}$, respectively.
By leveraging this observation and Corollaries~\ref{thecorollaryerrorbetweentrueandused}-\ref{thecorollaryofLipschitzcontinuousinequality},
{\color{blue}the proof of Theorem~\ref{thetheoremconvergenceofpolicygradient} is as follows.}

\begin{proof}
Recalling the definitions of $\bar{h}^{\bm{\hat{\pi}}_{\bm{\hat{\theta}}_{m}},\bm{\hat{\mu}}_{m}}_{app,i}$ in (\ref{theapproximatedgradientinalgorithmdesignmubar}) and $\bar{g}^{\bm{\hat{\pi}}_{\bm{\hat{\theta}}_{m}},\bm{\hat{\mu}}_{m+1}}_{app,i}$ in (\ref{theapproximatedgradientinalgorithmdesignthetabar}), we denote
$\bm{\bar{h}}^{\bm{\hat{\pi}}_{\bm{\hat{\theta}}_{m}},\bm{\hat{\mu}}_{m}}_{app}=\big((\bar{h}^{\bm{\hat{\pi}}_{\bm{\hat{\theta}}_{m}},\bm{\hat{\mu}}_{m}}_{app,1})^{\top},\cdots,(\bar{h}^{\bm{\hat{\pi}}_{\bm{\hat{\theta}}_{m}},\bm{\hat{\mu}}_{m}}_{app,N})^{\top}\big)^{\top}$ and  $\bm{\bar{g}}^{\bm{\hat{\pi}}_{\bm{\hat{\theta}}_{m}},\bm{\hat{\mu}}_{m+1}}_{app}=\big((\bar{g}^{\bm{\hat{\pi}}_{\bm{\hat{\theta}}_{m}},\bm{\hat{\mu}}_{m+1}}_{app,1})^{\top},\cdots,(\bar{g}^{\bm{\hat{\pi}}_{\bm{\hat{\theta}}_{m}},\bm{\hat{\mu}}_{m+1}}_{app,N})^{\top}\big)^{\top}$.
\par
Consider that the feasible set $\Theta_{i}$ is convex for all $i\in\mathcal{N}$, by the update of the policy parameter in (\ref{theupdateoftruepolicyparameterstheta}) and the property of the projection operator, it holds that
\begin{align}
\langle\bm{\theta}_{m}+\eta_{\theta,m}\bm{\bar{g}}^{\bm{\hat{\pi}}_{\bm{\hat{\theta}}_{m}},\bm{\hat{\mu}}_{m+1}}_{app}-\bm{\theta}_{m+1},\bm{\theta}-\bm{\theta}_{m+1}\rangle\leq 0, \forall \bm{\theta}\in\bm{\Theta}.\notag
\end{align}
Then, we have
\begin{align}
\langle\bm{\bar{g}}^{\bm{\hat{\pi}}_{\bm{\hat{\theta}}_{m}},\bm{\hat{\mu}}_{m+1}}_{app},\bm{\theta}-\bm{\theta}_{m+1}\rangle\leq\frac{1}{\eta_{\theta,m}}\langle\bm{\theta}_{m+1}-\bm{\theta}_{m},\bm{\theta}-\bm{\theta}_{m+1}\rangle\label{thetaonlyY-2}
\end{align}
for all $\bm{\theta}\in\bm{\Theta}$.
According to the definition of $\mathcal{Y}(\bm{\theta},\bm{\mu})$ in (\ref{thedefinitionofmathcalY}), we have
\begin{align}
&[\mathcal{Y}(\bm{\theta}_{m},\bm{\mu}_{m+1})]^{2}\notag\\
=&\Big[\max\limits_{\bm{\theta}\in\bm{\Theta},\|\bm{\theta}-\bm{\theta}_{m}\|\leq1}\langle\nabla_{\bm{\theta}}\mathcal{L}(\bm{\theta}_{m},\bm{\mu}_{m+1}),\bm{\theta}-\bm{\theta}_{m}\rangle\Big]^{2}\notag\\
\leq&2\Big[\max\limits_{\bm{\theta}\in\bm{\Theta},\|\bm{\theta}-\bm{\theta}_{m}\|\leq1}\langle\bm{\bar{g}}^{\bm{\hat{\pi}}_{\bm{\hat{\theta}}_{m}},\bm{\hat{\mu}}_{m+1}}_{app},\bm{\theta}-\bm{\theta}_{m}\rangle\Big]^{2}\notag\\
&+2\big\|\nabla_{\bm{\theta}}\mathcal{L}(\bm{\theta}_{m},\bm{\mu}_{m+1})-\bm{\bar{g}}^{\bm{\hat{\pi}}_{\bm{\hat{\theta}}_{m}},\bm{\hat{\mu}}_{m+1}}_{app}\big\|^{2}\notag\\
\leq&2\Big[\max\limits_{\bm{\theta}\in\bm{\Theta},\|\bm{\theta}-\bm{\theta}_{m}\|\leq1}\langle\bm{\bar{g}}^{\bm{\hat{\pi}}_{\bm{\hat{\theta}}_{m}},\bm{\hat{\mu}}_{m+1}}_{app},\bm{\theta}-\bm{\theta}_{m}\rangle\Big]^{2}\notag\\
&+2\big(\epsilon_{\theta}(m,K_{\theta})+\epsilon_{0}(\kappa,\kappa_{p})\big)^{2},\label{thetaonlyY1}
\end{align}
where the last inequality can be obtained by Case (ii) in Corollary~\ref{thecorollaryerrorbetweentrueandused}.
For $\max\limits_{\bm{\theta}\in\bm{\Theta},\|\bm{\theta}-\bm{\theta}_{m}\|\leq1}\langle\bm{\bar{g}}^{\bm{\hat{\pi}}_{\bm{\hat{\theta}}_{m}},\bm{\hat{\mu}}_{m+1}}_{app},\bm{\theta}-\bm{\theta}_{m}\rangle$ on the right side of (\ref{thetaonlyY1}),
we have
\begin{align}
&\max\limits_{\bm{\theta}\in\bm{\Theta},\|\bm{\theta}-\bm{\theta}_{m}\|\leq1}\langle\bm{\bar{g}}^{\bm{\hat{\pi}}_{\bm{\hat{\theta}}_{m}},\bm{\hat{\mu}}_{m+1}}_{app},\bm{\theta}-\bm{\theta}_{m}\rangle\notag\\
=&\max\limits_{\bm{\theta}\in\bm{\Theta},\|\bm{\theta}-\bm{\theta}_{m}\|\leq1}\big\{\langle\bm{\bar{g}}^{\bm{\hat{\pi}}_{\bm{\hat{\theta}}_{m}},\bm{\hat{\mu}}_{m+1}}_{app},\bm{\theta}-\bm{\theta}_{m+1}\rangle\notag\\
&+\langle\bm{\bar{g}}^{\bm{\hat{\pi}}_{\bm{\hat{\theta}}_{m}},\bm{\hat{\mu}}_{m+1}}_{app},\bm{\theta}_{m+1}-\bm{\theta}_{m}\rangle\big\}\notag\\
\leq&\max\limits_{\bm{\theta}\in\bm{\Theta},\|\bm{\theta}-\bm{\theta}_{m}\|\leq1}\Big\{\frac{1}{\eta_{\theta,m}}\langle\bm{\theta}_{m+1}-\bm{\theta}_{m},\bm{\theta}-\bm{\theta}_{m}\rangle\notag\\
&+\frac{1}{\eta_{\theta,m}}\langle\bm{\theta}_{m+1}-\bm{\theta}_{m},\bm{\theta}_{m}-\bm{\theta}_{m+1}\rangle\notag\\
&+\big\|\bm{\bar{g}}^{\bm{\hat{\pi}}_{\bm{\hat{\theta}}_{m}},\bm{\hat{\mu}}_{m+1}}_{app}\big\|\big\|\bm{\theta}_{m+1}-\bm{\theta}_{m}\big\|\Big\}\label{theorem3inequalityY3-1}\\
\leq&\frac{1}{\eta_{\theta,m}}\|\bm{\theta}_{m+1}-\bm{\theta}_{m}\|-\frac{1}{\eta_{\theta,m}}\|\bm{\theta}_{m+1}-\bm{\theta}_{m}\|^{2}\notag\\
&+\big\|\bm{\bar{g}}^{\bm{\hat{\pi}}_{\bm{\hat{\theta}}_{m}},\bm{\hat{\mu}}_{m+1}}_{app}\big\|\big\|\bm{\theta}_{m+1}-\bm{\theta}_{m}\big\|\notag\\
\leq&\Big(\frac{1}{\eta_{\theta,m}}+\sqrt{N}L^{\theta}_{1}\Big)\big\|\bm{\theta}_{m+1}-\bm{\theta}_{m}\big\|,\label{theorem3inequalityY3-2}
\end{align}
where (\ref{theorem3inequalityY3-1}) can be obtained by (\ref{thetaonlyY-2}) and (\ref{theorem3inequalityY3-2}) uses (\ref{theapproximatedpolicygradient2}).
Substituting (\ref{theorem3inequalityY3-2}) into (\ref{thetaonlyY1}), we have
\begin{align}
[\mathcal{Y}(\bm{\theta}_{m},\bm{\mu}_{m+1})]^{2}
\leq&2\Big(\frac{1}{\eta_{\theta,m}}+\sqrt{N}L^{\theta}_{1}\Big)^{2}\big\|\bm{\theta}_{m+1}-\bm{\theta}_{m}\big\|^{2}\notag\\
&+2\big(\epsilon_{\theta}(m,K_{\theta})+\epsilon_{0}(\kappa,\kappa_{p})\big)^{2}.\label{thetaonlymedi}
\end{align}
\par
In order to properly upper-bound the term $\big\|\bm{\theta}_{m+1}-\bm{\theta}_{m}\big\|^{2}$ in (\ref{thetaonlymedi}), we focus on the smoothness of $\mathcal{L}(\bm{\theta},\bm{\mu})$ with respect to $\bm{\theta}$ at the fixed point $\bm{\mu}_{m+1}$ and derive that
\begin{align}
&\mathcal{L}(\bm{\theta}_{m+1},\bm{\mu}_{m+1})-\mathcal{L}(\bm{\theta}_{m},\bm{\mu}_{m+1})\notag\\
\geq&\langle\nabla_{\bm{\theta}}\mathcal{L}(\bm{\theta}_{m},\bm{\mu}_{m+1}),\bm{\theta}_{m+1}-\bm{\theta}_{m}\rangle-\frac{L_{\theta\theta}}{2}\big\|\bm{\theta}_{m+1}-\bm{\theta}_{m}\big\|^{2},\label{thetaonly10-1}
\end{align}
where (\ref{thetaonly10-1}) uses the $L_{\theta\theta}$-smoothness of $\mathcal{L}(\bm{\theta},\bm{\mu})$ in Corollary~\ref{thecorollaryofLipschitzcontinuousinequality}.
For $\langle\nabla_{\bm{\theta}}\mathcal{L}(\bm{\theta}_{m},\bm{\mu}_{m+1}),\bm{\theta}_{m+1}-\bm{\theta}_{m}\rangle$ on the right hand of (\ref{thetaonly10-1}), we have that
\begin{align}
&\langle\nabla_{\bm{\theta}}\mathcal{L}(\bm{\theta}_{m},\bm{\mu}_{m+1}),\bm{\theta}_{m+1}-\bm{\theta}_{m}\rangle\notag\\
=&\langle\bm{\bar{g}}^{\bm{\hat{\pi}}_{\bm{\hat{\theta}}_{m}},\bm{\hat{\mu}}_{m+1}}_{app},\bm{\theta}_{m+1}-\bm{\theta}_{m}\rangle\notag\\
&+\langle\nabla_{\bm{\theta}}\mathcal{L}(\bm{\theta}_{m},\bm{\mu}_{m+1})-\bm{\bar{g}}^{\bm{\hat{\pi}}_{\bm{\hat{\theta}}_{m}},\bm{\hat{\mu}}_{m+1}}_{app},\bm{\theta}_{m+1}-\bm{\theta}_{m}\rangle\notag\\
\geq&\frac{1}{\eta_{\theta,m}}\big\|\bm{\theta}_{m+1}-\bm{\theta}_{m}\big\|^{2}-\frac{\eta_{\theta,m}}{2}\big\|\nabla_{\bm{\theta}}\mathcal{L}(\bm{\theta}_{m},\bm{\mu}_{m+1})\notag\\
&-\bm{\bar{g}}^{\bm{\hat{\pi}}_{\bm{\hat{\theta}}_{m}},\bm{\hat{\mu}}_{m+1}}_{app}\big\|^{2}-\frac{1}{2\eta_{\theta,m}}\big\|\bm{\theta}_{m+1}-\bm{\theta}_{m}\big\|^{2}\label{thetaonly6-1}\\
\geq&\frac{1}{2\eta_{\theta,m}}\big\|\bm{\theta}_{m+1}-\bm{\theta}_{m}\big\|^{2}-\frac{\eta_{\theta,m}}{2}\big(\epsilon_{\theta}(m,K_{\theta})+\epsilon_{0}(\kappa,\kappa_{p})\big)^{2},\label{thetaonly6-2}
\end{align}
where (\ref{thetaonly6-1}) uses (\ref{thetaonlyY-2}) and the Cauchy's inequality that $\langle a,b\rangle\geq-\frac{k}{2}\|a\|^{2}-\frac{1}{2k}\|y\|^{2}$ for all $k>0$, and (\ref{thetaonly6-2}) can be achieved by Case (ii) in Corollary~\ref{thecorollaryerrorbetweentrueandused}.
Substituting (\ref{thetaonly6-2}) into (\ref{thetaonly10-1}) and using {\color{blue}$\eta_{\theta,m}=1/(2\sqrt{m}+L_{\theta\theta})$}, i.e., $\frac{1}{2\eta_{\theta,m}}-\frac{L_{\theta\theta}}{2}={\color{blue}\sqrt{m}}$, we have
\begin{align}
&\mathcal{L}(\bm{\theta}_{m+1},\bm{\mu}_{m+1})-\mathcal{L}(\bm{\theta}_{m},\bm{\mu}_{m+1})\notag\\
\geq&{\color{blue}\sqrt{m}}\big\|\bm{\theta}_{m+1}-\bm{\theta}_{m}\big\|^{2}-\frac{\eta_{\theta,m}}{2}\big(\epsilon_{\theta}(m,K_{\theta})+\epsilon_{0}(\kappa,\kappa_{p})\big)^{2}.\label{thetaonlyDm}
\end{align}
We introduce a auxiliary term
\begin{align}
\mathcal{C}_{m}:=\mathcal{L}(\bm{\theta}_{m+1},\bm{\mu}_{m+2})-\mathcal{L}(\bm{\theta}_{m+1},\bm{\mu}_{m+1}).\label{thebridgetermdefinition}
\end{align}
By the definition of $\mathcal{L}(\bm{\theta},\bm{\mu})$ in (\ref{LagrangianfunctionofCMARL}) and its affine dependence on $\bm{\mu}$, $\mathcal{C}_{m}$ admits the exact expression
\begin{align}
\mathcal{C}_{m}=\frac{1}{N}\sum_{i=1}^{N}(\mu_{i,m+2}-\mu_{i,m+1})\big(G_{i}(\bm{\theta}_{m+1})-c_{i}\big).\label{thebridgetermexact}
\end{align}
By the non-expansiveness of the projection operator and (\ref{theapproximatedpolicygradientmu1}), we have
\begin{align}
&\big\|\bm{\mu}_{m+2}-\bm{\mu}_{m+1}\big\|\notag\\
=&\big\|P_{\bm{\mathcal{U}}}\big(\bm{\mu}_{m+1}-\eta_{\mu,m+1}\bm{\bar{h}}^{\bm{\hat{\pi}}_{\bm{\hat{\theta}}_{m+1}},\bm{\hat{\mu}}_{m+1}}_{app}\big)-P_{\bm{\mathcal{U}}}(\bm{\mu}_{m+1})\big\|\notag\\
\leq&\eta_{\mu,m+1}\sqrt{N}L^{\mu}_{1}.\label{thebridgetermnonexpansive}
\end{align}
Combining (\ref{thebridgetermexact})-(\ref{thebridgetermnonexpansive}), we use  Assumption~\ref{theassumptionofreward} and further have
\begin{align}
|\mathcal{C}_{m}|\leq (R_{g}+c_{max})L^{\mu}_{1}\eta_{\mu,m+1}=\mathcal{O}\Big(\frac{1}{m}\Big).\label{thebridgetermbound}
\end{align}
By the definition of $\mathcal{C}_{m}$ in  (\ref{thebridgetermdefinition}), (\ref{thetaonlyDm}) can be rewritten as
\begin{align}
&{\color{blue}\sqrt{m}}\big\|\bm{\theta}_{m+1}-\bm{\theta}_{m}\big\|^{2}\notag\\
\leq& \mathcal{L}(\bm{\theta}_{m+1},\bm{\mu}_{m+2})-\mathcal{C}_{m}-\mathcal{L}(\bm{\theta}_{m},\bm{\mu}_{m+1})\notag\\
&+\frac{\eta_{\theta,m}}{2}\big(\epsilon_{\theta}(m,K_{\theta})+\epsilon_{0}(\kappa,\kappa_{p})\big)^{2}.\label{thetaonlyDmrewritten}
\end{align}
Summing (\ref{thetaonlyDmrewritten}) over $m=1,\cdots,M-1$, we have
\begin{align}
&\sum_{m=1}^{M-1}{\color{blue}\sqrt{m}}\big\|\bm{\theta}_{m+1}-\bm{\theta}_{m}\big\|^{2}\notag\\
\leq&\mathcal{L}(\bm{\theta}_{M},\bm{\mu}_{M+1})-\mathcal{L}(\bm{\theta}_{1},\bm{\mu}_{2})-\sum_{m=1}^{M-1}\mathcal{C}_{m}\notag\\
&+\sum_{m=1}^{M-1}\frac{\eta_{\theta,m}}{2}\big(\epsilon_{\theta}(m,K_{\theta})+\epsilon_{0}(\kappa,\kappa_{p})\big)^{2}\notag\\
\leq&\frac{2(R_{f}+\mu_{max}R_{g})}{1-\gamma}+\sum_{m=1}^{M-1}|\mathcal{C}_{m}|\notag\\
&+\sum_{m=1}^{M-1}\frac{\eta_{\theta,m}}{2}\big(\epsilon_{\theta}(m,K_{\theta})+\epsilon_{0}(\kappa,\kappa_{p})\big)^{2}.\label{thetaonlytelescoped}
\end{align}
Recall from  Theorem~\ref{Convergencetheoreminparameterestimation}, using $\|\hat{\theta}^{i}_{j,m}-\theta_{j,m}\|\sim{\color{blue}\mathcal{O}(1/\sqrt{m})}$ and $|\hat{\mu}^{i}_{j,m}-\mu_{j,m}|\sim{\color{blue}\mathcal{O}(1/m)}$, we have,
\begin{align}
&\sum_{m=1}^{M-1}\frac{\eta_{\theta,m}}{2}\big(\epsilon_{\theta}(m,K_{\theta})+\epsilon_{0}(\kappa,\kappa_{p})\big)^{2}\notag\\
\leq&\mathcal{O}(1)+\mathcal{O}\Bigg({\color{blue}\sqrt{M}}\frac{\log(2/\delta)}{2K_{\theta}}\Bigg)+\mathcal{O}\Big({\color{blue}\sqrt{M}}\epsilon_{0}(\kappa,\kappa_{p})^{2}\Big).\label{thetaonlyiiterm}
\end{align}
Substituting (\ref{thebridgetermbound}), (\ref{thetaonlyiiterm}) into (\ref{thetaonlytelescoped}), we obtain
\begin{align}
&\sum_{m=1}^{M-1}{\color{blue}\sqrt{m}}\big\|\bm{\theta}_{m+1}-\bm{\theta}_{m}\big\|^{2}\notag\\
\leq&\mathcal{O}\big(\log{M}\big)+\mathcal{O}\Bigg({\color{blue}\sqrt{M}}\frac{\log(2/\delta)}{2K_{\theta}}\Bigg)\notag\\
&+\mathcal{O}\Big({\color{blue}\sqrt{M}}\epsilon_{0}(\kappa,\kappa_{p})^{2}\Big).\label{thetaonlysumbound}
\end{align}
Define $\chi_{m}$ as a weight and described as
\begin{align}
\chi_{m}=\frac{{\color{blue}\sqrt{m}}}{2\Big({\color{blue}2\sqrt{m}}+L_{\theta\theta}+\sqrt{N}L^{\theta}_{1}\Big)^{2}}\sim\mathcal{O}\Big(\frac{1}{{\color{blue}\sqrt{m}}}\Big),\label{thetaonlyweight}
\end{align}
then multiplying both sides of (\ref{thetaonlymedi}) by $\chi_{m}$ and summing over $m=1,\cdots,M-1$, we have, by (\ref{thetaonlysumbound}),
\begin{align}
&\sum_{m=1}^{M-1}\chi_{m}\big[\mathcal{Y}(\bm{\theta}_{m},\bm{\mu}_{m+1})\big]^{2}\notag\\
\leq&\sum_{m=1}^{M-1}{\color{blue}\sqrt{m}}\big\|\bm{\theta}_{m+1}-\bm{\theta}_{m}\big\|^{2}\notag\\
&+2\sum_{m=1}^{M-1}\chi_{m}\big(\epsilon_{\theta}(m,K_{\theta})+\epsilon_{0}(\kappa,\kappa_{p})\big)^{2}\notag\\
=&\mathcal{O}\big(\log{M}\big)+\mathcal{O}\Bigg({\color{blue}\sqrt{M}}\frac{\log(2/\delta)}{2K_{\theta}}\Bigg)+\mathcal{O}\Big({\color{blue}\sqrt{M}}\epsilon_{0}(\kappa,\kappa_{p})^{2}\Big).\label{thetaonlyweightedsum}
\end{align}
Dividing both sides by $\sum_{m=1}^{M-1}\chi_{m}=\mathcal{O}({\color{blue}\sqrt{M}})$, we obtain
\begin{align}
&\frac{1}{\sum_{m=1}^{M-1}\chi_{m}}\sum_{m=1}^{M-1}\chi_{m}\big[\mathcal{Y}(\bm{\theta}_{m},\bm{\mu}_{m+1})\big]^{2}\notag\\
=&\mathcal{O}\Big(\frac{{\color{blue}\log{M}}}{{\color{blue}\sqrt{M}}}\Big)+\mathcal{O}\Bigg(\frac{\log(2/\delta)}{2K_{\theta}}\Bigg)+\mathcal{O}\Big(\epsilon_{0}(\kappa,\kappa_{p})^{2}\Big).\label{thetaonlyfinalbeforesub}
\end{align}
Since $\frac{1}{\sum_{m=1}^{M-1}\chi_{m}}\sum_{m=1}^{M-1}\chi_{m}\big[\mathcal{Y}(\bm{\theta}_{m},\bm{\mu}_{m+1})\big]^{2}$ is a weighted average of $M-1$ non-negative terms, there must exist $m^{*}\in\{1,\cdots,M-1\}$ such that
\begin{align}
\big[\mathcal{Y}(\bm{\theta}_{m^{*}},\bm{\mu}_{m^{*}+1})\big]^{2}=&\mathcal{O}\Big(\frac{{\color{blue}\log{M}}}{{\color{blue}\sqrt{M}}}\Big)+\mathcal{O}\Bigg(\frac{\log(2/\delta)}{2K_{\theta}}\Bigg)\notag\\
&+\mathcal{O}\Big(\epsilon_{0}(\kappa,\kappa_{p})^{2}\Big).\label{thetaonlyfinal}
\end{align}
By substituting the settings {\color{blue}$M=\tilde{\mathcal{O}}(1/\epsilon^{2})$} and $K_{\theta}=\log(2/\delta)/2\epsilon$ into (\ref{thetaonlyfinal}),
we can rigorously complete the proof.
\end{proof}

\section{Appendix}\label{Appendixpart}
{\color{blue}In this section, we provide detailed proofs of the lemmas presented in the main body and the convergence analysis.}

\subsection{Proof of Lemma~\ref{thelemmaofthepolicygradienttheoreminNMARLCP}}\label{ProofofLemmathelemmaofthepolicygradienttheoreminNMARLCP}
\begin{proof}
Recalling the definition of $\mathcal{L}(\bm{\theta},\bm{\mu})$ provided in (\ref{LagrangianfunctionofCMARL}), the expression of gradient $\nabla_{\mu_{i}}\mathcal{L}(\bm{\theta},\bm{\mu})$ in (\ref{thecoupledpolicygradienttheorem2}) can be directly derived.
\par
For $\nabla_{\theta_{i}}\mathcal{L}(\bm{\theta},\bm{\mu})$, we
combine the global Lagrangian $Q$-function in (\ref{thedefinitionofglobalQfunction}) with the policy gradient theorem in~\cite{Sutton2000} and have
\begin{align}
&\nabla_{\theta_{i}}\mathcal{L}(\bm{\theta},\bm{\mu})\notag\\
=&\frac{1}{1-\gamma}\mathbb{E}_{\bm{s}\sim d^{\bm{\pi_{\theta}}}_{\bm{\rho}},\bm{a}\sim\bm{\pi_{\theta}}}\Big[Q^{\bm{\pi_{\theta}}}(\bm{s},\bm{a};\bm{\mu})\nabla_{\theta_{i}}\log\bm{\pi_{\theta}}(\bm{a}|\bm{s})\Big]\notag\\
=&\frac{1}{1-\gamma}\mathbb{E}_{\bm{s}\sim d^{\bm{\pi_{\theta}}}_{\bm{\rho}},\bm{a}\sim\bm{\pi_{\theta}}}\Big[Q^{\bm{\pi_{\theta}}}(\bm{s},\bm{a};\bm{\mu})\notag\\
&\times\nabla_{\theta_{i}}\Big(\sum_{j=1}^{N}\log\pi_{j}(a_{j}|s_{\mathcal{N}^{E,\kappa_{p}}},\theta_{j},\theta_{\mathcal{N}^{E,\kappa_{p}}_{j,-j}})\Big)\Big].\label{midleofthepolicygradienttheorem}
\end{align}
Since only the policies $\{\pi_{j}(\cdot|s_{\mathcal{N}^{E,\kappa_{p}}_{j}},\theta_{j},\theta_{\mathcal{N}^{E,\kappa_{p}}_{j,-j}})\}_{j\in\mathcal{N}^{E,\kappa_{p}}_{i}}$ are related to $\theta_{i}$, (\ref{midleofthepolicygradienttheorem}) can
consequently derive (\ref{thecoupledpolicygradienttheorem}).
\end{proof}

\subsection{Proof of Lemma~\ref{thelemmaofexponentialdecayproperty}}\label{ProofofLemmathelemmaofexponentialdecayproperty}
\begin{proof}
In the CMARL problem, for any joint policy $\bm{\pi_{\theta}}$, let $\xi_{1,t}=\mathrm{Pr}^{\bm{\pi_{\theta}}}\big(s_{i,t},a_{i,t}|\bm{s}_{0}=(s_{\mathcal{N}^{E,\kappa}_{i}},s_{\mathcal{N}^{E,\kappa}_{-i}}),\bm{a}_{0}=(a_{\mathcal{N}^{E,\kappa}_{i}},a_{\mathcal{N}^{E,\kappa}_{-i}})\big)$ represent the probability of occurrence of $(s_{i,t},a_{i,t})$ at time $t$ under joint policy $\bm{\pi_{\theta}}$,
conditioned on the initial state $(s_{\mathcal{N}^{E,\kappa}_{i}},s_{\mathcal{N}^{E,\kappa}_{-i}})$ and initial action $(a_{\mathcal{N}^{E,\kappa}_{i}},a_{\mathcal{N}^{E,\kappa}_{-i}})$.
Similarly, define $\xi_{2,t}=\mathrm{Pr}^{\bm{\pi_{\theta}}}\big(s_{i,t},a_{i,t}|\bm{s}_{0}=(s_{\mathcal{N}^{E,\kappa}_{i}},s'_{\mathcal{N}^{E,\kappa}_{-i}}),\bm{a}_{0}=(a_{\mathcal{N}^{E,\kappa}_{i}},a'_{\mathcal{N}^{E,\kappa}_{-i}})\big)$.
Let $\mathrm{TV}(\xi_{1,t},\xi_{2,t})$ represent the total variation distance between $\xi_{1,t}$ and $\xi_{2,t}$.
From the definitions of the state transition probability functions $\{\mathcal{P}_{i}(s'_{i}|s_{\mathcal{N}^{E}_{i}},a_{i})\}_{i\in\mathcal{N}}$ and coupled policies $\{\pi_{i}(\cdot|s_{\mathcal{N}^{E,\kappa_{p}}_{i}},\theta_{i},\theta_{\mathcal{N}^{E,\kappa_{p}}_{i,-i}})\}_{i\in\mathcal{N}}$, we can directly have that $\xi_{1,t}=\xi_{2,t}$ when $t\leq h(\kappa,\kappa_{p})$.
Given this fact, we obtain
\begin{align}
&\Big|Q^{\bm{\pi_{\theta}}}_{i}(s_{\mathcal{N}^{E,\kappa}_{i}},s_{\mathcal{N}^{E,\kappa}_{-i}},a_{\mathcal{N}^{E,\kappa}_{i}},a_{\mathcal{N}^{E,\kappa}_{-i}};\bm{\mu})\notag\\
&-Q^{\bm{\pi_{\theta}}}_{i}(s_{\mathcal{N}^{E,\kappa}_{i}},s'_{\mathcal{N}^{E,\kappa}_{-i}},a_{\mathcal{N}^{E,\kappa}_{i}},a'_{\mathcal{N}^{E,\kappa}_{-i}};\bm{\mu})\Big|\notag\\
=&\Big|\frac{1}{N}\sum^{\infty}_{t=0}\gamma^{t}\mathbb{E}_{\xi_{1,t}}[f_{i}(s_{i,t},a_{i,t})+\mu_{i}g_{i}(s_{i,t},a_{i,t})]\notag\\
&-\frac{1}{N}\sum^{\infty}_{t=0}\gamma^{t}\mathbb{E}_{\xi_{2,t}}[f_{i}(s_{i,t},a_{i,t})+\mu_{i}g_{i}(s_{i,t},a_{i,t})]\Big|\label{theerrorofneighborsfunction1state}\\
\leq&\frac{1}{N}\sum^{\infty}_{t=0}\Big|\gamma^{t}\mathbb{E}_{\xi_{1,t}}[r_{i}(s_{i,t},a_{i,t})+\mu_{i}g_{i}(s_{i,t},a_{i,t})]\notag\\
&-\gamma^{t}\mathbb{E}_{\xi_{2,t}}[r_{i}(s_{i,t},a_{i,t})+\mu_{i}g_{i}(s_{i,t},a_{i,t})]\Big|\label{theerrorofneighborsfunction2state}\\
=&\frac{1}{N}\sum^{\infty}_{t=h(\kappa,\kappa_{p})+1}\Big|\gamma^{t}\mathbb{E}_{\xi_{1,t}}[r_{i}(s_{i,t},a_{i,t})+\mu_{i}g_{i}(s_{i,t},a_{i,t})]\notag\\
&-\gamma^{t}\mathbb{E}_{\xi_{2,t}}[r_{i}(s_{i,t},a_{i,t})+\mu_{i}g_{i}(s_{i,t},a_{i,t})]\Big|\label{theerrorofneighborsfunction3state}\\
\leq&\sum^{\infty}_{t=h(\kappa,\kappa_{p})+1}2\gamma^{t}(R_{f}+\mu_{max}R_{g})\mathrm{TV}(\xi_{1,t},\xi_{2,t})\label{theerrorofneighborsfunction4state}\\
\leq&\frac{2(R_{f}+\mu_{max}R_{g})}{(1-\gamma)\gamma^{\kappa-h(\kappa,\kappa_{p})}}\gamma^{\kappa+1},\label{theerrorofneighborsfunction5state}
\end{align}
where (\ref{theerrorofneighborsfunction1state}) comes from the definition of local Lagrangian $Q$-function in (\ref{thedefinitionoflocalQfunction}), (\ref{theerrorofneighborsfunction2state}) can be obtained by the fact that $|\sum_{t=0}^{\infty}(x_{t}-y_{t})|\leq\sum_{t=0}^{\infty}|x_{t}-y_{t}|,\forall x_{t},y_{t}\in\mathbb{R}$,
(\ref{theerrorofneighborsfunction3state}) is achieved by $\xi_{1,t}=\xi_{2,t}$ for all $t\leq h(\kappa,\kappa_{p})$,
(\ref{theerrorofneighborsfunction4state}) follows from Assumptions~\ref{theassumptionofreward} and~\ref{theassumptionofFOSP}, and
(\ref{theerrorofneighborsfunction5state}) comes from
the fact that $\mathrm{TV}(\xi_{1,t},\xi_{2,t})$ is upper bounded by 1.
\end{proof}

\subsection{Proof of Lemma~\ref{thelemmaoftruncatederror}}\label{ProofofLemmathelemmaoftruncatederror}
\begin{proof}
For any $l\in\mathcal{N}^{E,\kappa+2\kappa_{p}}_{-i}$, we have
\begin{align}
&\mathbb{E}_{\bm{s}\sim d^{\bm{\pi_{\theta}}}_{\bm{\rho}},\bm{a}\sim\bm{\pi_{\theta}}}\Big[Q^{\bm{\pi_{\theta}}}_{tru,l}(s_{\mathcal{N}^{E,\kappa}_{l}},a_{\mathcal{N}^{E,\kappa}_{l}};\bm{\mu})\notag\\
&\times\sum_{j\in\mathcal{N}^{E,\kappa_{p}}_{i}}\nabla_{\theta_{i}}\log\pi_{j}(a_{j}|s_{\mathcal{N}^{E,\kappa_{p}}_{j}},\theta_{j},\theta_{\mathcal{N}^{E,\kappa_{p}}_{j,-j}})\Big]\notag\\
=&\sum_{j\in\mathcal{N}^{E,\kappa_{p}}_{i}}\mathbb{E}_{\bm{s}\sim d^{\bm{\pi_{\theta}}}_{\bm{\rho}},\bm{a}\sim\bm{\pi_{\theta}}}\Big[Q^{\bm{\pi_{\theta}}}_{tru,l}(s_{\mathcal{N}^{E,\kappa}_{l}},a_{\mathcal{N}^{E,\kappa}_{l}};\bm{\mu})\notag\\
&\times\nabla_{\theta_{i}}\log\pi_{j}(a_{j}|s_{\mathcal{N}^{E,\kappa_{p}}_{j}},\theta_{j},\theta_{\mathcal{N}^{E,\kappa_{p}}_{j,-j}})\Big]\notag\\
=&\sum_{j\in\mathcal{N}^{E,\kappa_{p}}_{i}}\sum_{\bm{s},\bm{a}}d^{\bm{\pi_{\theta}}}_{\bm{\rho}}(\bm{s})\prod_{i=1}^{N}\pi_{i}(a_{i}|s_{\mathcal{N}^{E,\kappa_{p}}_{i}},\theta_{i},\theta_{\mathcal{N}^{E,\kappa_{p}}_{i,-i}})\notag\\
&\times\Big[Q^{\bm{\pi_{\theta}}}_{tru,l}(s_{\mathcal{N}^{E,\kappa}_{l}},a_{\mathcal{N}^{E,\kappa}_{l}};\bm{\mu})\frac{\nabla_{\theta_{i}}\pi_{j}(a_{j}|s_{\mathcal{N}^{E,\kappa_{p}}_{j}},\theta_{j},\theta_{\mathcal{N}^{E,\kappa_{p}}_{j,-j}})}{\pi_{j}(a_{j}|s_{\mathcal{N}^{E,\kappa_{p}}_{j}},\theta_{j},\theta_{\mathcal{N}^{E,\kappa_{p}}_{j,-j}})}\Big]\notag\\
=&\sum_{j\in\mathcal{N}^{E,\kappa_{p}}_{i}}\sum_{\bm{s},a_{-j}}d^{\bm{\pi_{\theta}}}_{\bm{\rho}}(\bm{s})\prod_{i\neq j}\pi_{i}(a_{i}|s_{\mathcal{N}^{E,\kappa_{p}}_{i}},\theta_{i},\theta_{\mathcal{N}^{E,\kappa_{p}}_{i,-i}})\notag\\
&\times\Big[Q^{\bm{\pi_{\theta}}}_{tru,l}(s_{\mathcal{N}^{E,\kappa}_{l}},a_{\mathcal{N}^{E,\kappa}_{l}};\bm{\mu})\!\sum_{a_{j}}\!\nabla_{\theta_{i}}\pi_{j}(a_{j}|s_{\mathcal{N}^{E,\kappa_{p}}_{j}},\notag\\
&\theta_{j},\theta_{\mathcal{N}^{E,\kappa_{p}}_{j,-j}})\Big]\notag\\
=&0,\label{thetruncatederror3-1}
\end{align}
where the last inequality can be obtained by the fact that
\begin{align}\notag
&\sum_{a_{j}}\nabla_{\theta_{i}}\pi_{j}(a_{j}|s_{\mathcal{N}^{E,\kappa_{p}}_{j}},\theta_{j},\theta_{\mathcal{N}^{E,\kappa_{p}}_{j,-j}})\notag\\
=&\nabla_{\theta_{i}}\sum_{a_{j}}\pi_{j}(a_{j}|s_{\mathcal{N}^{E,\kappa_{p}}_{j}},\theta_{j},\theta_{\mathcal{N}^{E,\kappa_{p}}_{j,-j}})=\nabla_{\theta_{i}}1=0.
\end{align}
\end{proof}

\subsection{Proof of Lemma~\ref{theunbiasedestimationlemma}}\label{ProofofLemmatheunbiasedestimationlemma}
\begin{proof}
For Case (i), by the definition of $\hat{h}^{\bm{\hat{\pi}}_{\bm{\hat{\theta}}_{m}},\bm{\hat{\mu}}_{m}}_{app,i}(k)$ in (\ref{theapproximatedgradientinalgorithmdesignmu}), we have
\begin{align}
&\mathbb{E}_{T_{1}(k)}\Big[\hat{h}^{\hat{\bm{\pi}}_{\bm{\hat{\theta}}_{m}},\bm{\hat{\mu}}_{m}}_{app,i}(k)\Big|\bm{s}_{0}\sim\bm{\rho},\bm{\hat{\pi}}_{\bm{\hat{\theta}}_{m}}\Big]\notag\\
=&\frac{1}{N}\Big(\mathbb{E}_{T_{1}(k)}\Big[\sum_{t=0}^{T_{1}(k)}\gamma^{t/2}g_{i,t}\Big|\bm{s}_{0}\sim\bm{\rho},\bm{\hat{\pi}}_{\bm{\hat{\theta}}_{m}}\Big]-c_{i}\Big)\notag\\
=&\frac{1}{N}\Big(\mathbb{E}_{T_{1}(k)}\Big[\sum_{t=0}^{\infty}\mathds{1}_{\{0\leq t\leq T_{1}(k)\}}\gamma^{t/2}g_{i,t}\Big|\bm{s}_{0}\sim\bm{\rho},\bm{\hat{\pi}}_{\bm{\hat{\theta}}_{m}}\Big]-c_{i}\Big).\label{theunbiasedestimatemu1}
\end{align}
By the boundedness of rewards in Assumption~\ref{theassumptionofreward}, for any $T'>0$, we have that
\begin{align}
&\mathbb{E}_{T_{1}(k)}\Big[\sum_{t=0}^{T'}\mathds{1}_{\{0\leq t\leq T_{1}(k)\}}\gamma^{t/2}g_{i,t}\Big|\bm{s}_{0}\sim\bm{\rho},\bm{\hat{\pi}}_{\bm{\hat{\theta}}_{m}}\Big]\notag\\
\leq&R_{g}\mathbb{E}_{T_{1}(k)}\Big[\sum_{t=0}^{T'}\mathds{1}_{\{0\leq t\leq T_{1}(k)\}}\gamma^{t/2}\Big|\bm{s}_{0}\sim\bm{\rho},\bm{\hat{\pi}}_{\bm{\hat{\theta}}_{m}}\Big].\label{theunbiasedestimatemu2}
\end{align}
Note that on the right-hand side of (\ref{theunbiasedestimatemu2}), the random variable in the expectation is monotonically increasing and the limit as $T'\rightarrow\infty$ exists.
By using the monotone convergence theorem in~\cite{Yeh2006}, we can have that
\begin{align}
&\frac{1}{N}\mathbb{E}_{T_{1}(k)}\Big[\sum_{t=0}^{\infty}\mathds{1}_{\{0\leq t\leq T_{1}(k)\}}\gamma^{t/2}g_{i,t}\Big|\bm{s}_{0}\sim\bm{\rho},\bm{\hat{\pi}}_{\bm{\hat{\theta}}_{m}}\Big]\notag\\
=&\frac{1}{N}\sum_{t=0}^{\infty}\mathbb{E}\Big[\mathbb{E}_{T_{1}(k)}[\mathds{1}_{\{0\leq t\leq T_{1}(k)\}}]\gamma^{t/2}g_{i,t}\Big|\bm{s}_{0}\sim\bm{\rho},\bm{\hat{\pi}}_{\bm{\hat{\theta}}_{m}}\Big]\notag\\
=&\frac{1}{N}\sum_{t=0}^{\infty}\mathbb{E}\Big[\gamma^{t}g_{i,t}\Big|\bm{s}_{0}\sim\bm{\rho},\bm{\hat{\pi}}_{\bm{\hat{\theta}}_{m}}\Big]\notag\\
=&\frac{1}{N}\mathbb{E}\Big[\sum_{t=0}^{\infty}\gamma^{t}g_{i,t}\Big|\bm{s}_{0}\sim\bm{\rho},\bm{\hat{\pi}}_{\bm{\hat{\theta}}_{m}}\Big],\label{theunbiasedestimatemu3}
\end{align}
where the second equality can be achieved by the fact that $T_{1}(k)\sim\mathrm{Geom}(1-\gamma^{1/2})$ and $\mathbb{E}_{T_{1}(k)}[\mathds{1}_{\{0\leq t\leq T_{1}(k)\}}]=\gamma^{t/2}$, and the last equality uses the dominated convergence theorem.
By substituting (\ref{theunbiasedestimatemu3}) into (\ref{theunbiasedestimatemu1}) and referring to (\ref{thepolicygradientapproxiamtionbynewfunctionmu2}), Case (i) can be rigorously proved.
\par
For Case (ii), by the definition of $\hat{g}^{\bm{\hat{\pi}}_{\bm{\hat{\theta}}_{m}},\bm{\hat{\mu}}_{m+1}}_{app,i}(k)$ in (\ref{theapproximatedgradientinalgorithmdesign}) and the independence of $T_{2}(k)$ and $T_{3}(k)$, we have
\begin{align}
&\mathbb{E}_{T_{2}(k),T_{3}(k)}\Big[\hat{g}^{\bm{\hat{\pi}}_{\bm{\hat{\theta}}_{m}},\bm{\hat{\mu}}_{m+1}}_{app,i}(k)\Big|\bm{s}_{0}\sim\bm{\rho},\bm{\hat{\pi}}_{\bm{\hat{\theta}}_{m}}\Big]\notag\\
=&\frac{1}{1-\gamma}\mathbb{E}_{T_{2}(k),T_{3}(k)}\Big[\hat{Q}^{\bm{\hat{\pi}}_{\bm{\hat{\theta}}_{m}},\bm{\hat{\mu}}_{m+1}}_{i,T_{2}(k)}\Big(\sum_{j\in\mathcal{N}^{E,\kappa_{p}}_{i}}\!\!\nabla_{\theta_{i}}\log\pi_{j}(a_{j,T_{2}(k)}|\notag\\
&s_{\mathcal{N}^{E,\kappa_{p}}_{j},T_{2}(k)},\hat{\theta}^{i}_{j,m},\hat{\theta}^{i}_{\mathcal{N}^{E,\kappa_{p}}_{j,-j},m})\Big)\Big|\bm{s}_{0}\sim\bm{\rho},\bm{\hat{\pi}}_{\bm{\hat{\theta}}_{m}}\Big]\notag\\
=&\frac{1}{1-\gamma}\mathbb{E}_{T_{2}(k)}\Bigg[\mathbb{E}_{T_{3}(k)}\Big[\hat{Q}^{\bm{\hat{\pi}}_{\bm{\hat{\theta}}_{m}},\bm{\hat{\mu}}_{m+1}}_{i,T_{2}(k)}\notag\\
&\times\Big(\!\!\sum_{j\in\mathcal{N}^{E,\kappa_{p}}_{i}}\!\!\nabla_{\theta_{i}}\log\pi_{j}(a_{j,T_{2}(k)}|s_{\mathcal{N}^{E,\kappa}_{j},T_{2}(k)},\hat{\theta}^{i}_{j,m},\hat{\theta}^{i}_{\mathcal{N}^{E,\kappa_{p}}_{j,-j},m})\Big)\Big|\notag\\
&\bm{s}_{T_{2}(k)},\bm{a}_{T_{2}(k)},\bm{\hat{\pi}}_{\bm{\hat{\theta}}_{m}}\Big]\Bigg|\bm{s}_{0}\sim\bm{\rho},\bm{\hat{\pi}}_{\bm{\hat{\theta}}_{m}}\Bigg].\label{theunbiasedestimate1}
\end{align}
According to Theorem 3.4 in~\cite{ZhangSIAM2020}, we have that
\begin{align}\label{theinequalityofunbiasedestimate}
&\widehat{Q^{\bm{\hat{\pi}}_{\bm{\hat{\theta}}_{m}}}_{i}}(\bm{s}_{T_{2}(k)},\bm{a}_{T_{2}(k)};\bm{\hat{\mu}}^{i}_{m+1})\notag\\
=&\mathbb{E}_{T_{3}(k)}\Big[\hat{Q}^{\bm{\hat{\pi}}_{\bm{\hat{\theta}}_{m}},\bm{\hat{\mu}}_{m+1}}_{i,T_{2}(k)}\Big|\bm{s}_{T_{2}(k)},\bm{a}_{T_{2}(k)},\bm{\hat{\pi}}_{\bm{\hat{\theta}}_{m}}\Big].
\end{align}
Substituting (\ref{theinequalityofunbiasedestimate}) into (\ref{theunbiasedestimate1}), we further have
\begin{align}
&\mathbb{E}_{T_{2}(k),T_{3}(k)}\Big[\hat{g}^{\bm{\hat{\pi}}_{\bm{\hat{\theta}}_{m}}}_{app,i}(k)\Big|\bm{s}_{0}\sim\bm{\rho},\bm{\hat{\pi}}_{\bm{\hat{\theta}}_{m}}\Big]\notag\\
=&\frac{1}{1-\gamma}\mathbb{E}_{T_{2}(k)}\Big[\widehat{Q^{\bm{\hat{\pi}}_{\bm{\hat{\theta}}_{m}}}_{i}}(\bm{s}_{T_{2}(k)},\bm{a}_{T_{2}(k)};\bm{\hat{\mu}}^{i}_{m+1})\notag\\
&\times\sum_{j\in\mathcal{N}^{E,\kappa_{p}}_{i}}\nabla_{\theta_{i}}\log\pi_{j}(a_{j,T_{2}(k)}|s_{\mathcal{N}^{E,\kappa_{p}}_{j},T_{2}(k)},\notag\\
&\hat{\theta}^{i}_{j,m},\hat{\theta}^{i}_{\mathcal{N}^{E,\kappa_{p}}_{j,-j},m})\Big|\bm{s}_{0}\sim\bm{\rho},\bm{\hat{\pi}}_{\bm{\hat{\theta}}_{m}}\Big]\notag\\
=&\frac{1}{1-\gamma}\mathbb{E}_{T_{2}(k)}\Big[\!\!\sum_{t'=0}^{\infty}\!\!\mathds{1}_{\{t'=T_{2}(k)\}}\widehat{Q^{\bm{\hat{\pi}}_{\bm{\hat{\theta}}_{m}}}_{i}}(\bm{s}_{T_{2}(k)},\bm{a}_{T_{2}(k)};\bm{\hat{\mu}}^{i}_{m+1})\notag\\
&\times\sum_{j\in\mathcal{N}^{E,\kappa_{p}}_{i}}\nabla_{\theta_{i}}\log\pi_{j}(a_{j,T_{2}(k)}|s_{\mathcal{N}^{E,\kappa_{p}}_{j},T_{2}(k)},\notag\\
&\hat{\theta}^{i}_{j,m},\hat{\theta}^{i}_{\mathcal{N}^{E,\kappa_{p}}_{j,-j},m})\Big|\bm{s}_{0}\sim\bm{\rho},\bm{\hat{\pi}}_{\bm{\hat{\theta}}_{m}}\Big]\notag\\
=&\frac{1}{1-\gamma}\sum_{t'=0}^{\infty}\mathbb{P}\big(t'=T_{2}(k)\big)\mathbb{E}\Big[\widehat{Q^{\bm{\hat{\pi}}_{\bm{\hat{\theta}}_{m}}}_{i}}(\bm{s}_{t'},\bm{a}_{t'};\bm{\hat{\mu}}^{i}_{m+1})\notag\\
&\times\sum_{j\in\mathcal{N}^{E,\kappa_{p}}_{i}}\nabla_{\theta_{i}}\log\pi_{j}(a_{j,t'}|s_{\mathcal{N}^{E,\kappa}_{j},t'},\hat{\theta}^{i}_{j,m},\hat{\theta}^{i}_{\mathcal{N}^{E,\kappa_{p}}_{j,-j},m})\Big|\notag\\
&\bm{s}_{0}\sim\bm{\rho},\bm{\hat{\pi}}_{\bm{\hat{\theta}}_{m}}\Big]\notag\\
=&\sum_{t'=0}^{\infty}\gamma^{t'}\mathbb{E}\Big[\widehat{Q^{\bm{\hat{\pi}}_{\bm{\hat{\theta}}_{m}}}_{i}}(\bm{s}_{t'},\bm{a}_{t'};\bm{\hat{\mu}}^{i}_{m+1})\notag\\
&\times\sum_{j\in\mathcal{N}^{E,\kappa_{p}}_{i}}\nabla_{\theta_{i}}\log\pi_{j}(a_{j,t'}|s_{\mathcal{N}^{E,\kappa_{p}}_{j},t'},\hat{\theta}^{i}_{j,m},\hat{\theta}^{i}_{\mathcal{N}^{E,\kappa_{p}}_{j,-j},m})\Big|\notag\\
&\bm{s}_{0}\sim\bm{\rho},\bm{\hat{\pi}}_{\bm{\hat{\theta}}_{m}}\Big]\label{theunbiasedestimate4}\\
=&\frac{1}{1-\gamma}\mathbb{E}_{\bm{s}\sim d^{\bm{\hat{\pi}}_{\bm{\hat{\theta}}_{m}}}_{\bm{\rho}},\bm{a}\sim\bm{\hat{\pi}}_{\bm{\hat{\theta}}_{m}}}\Big[\widehat{Q^{\bm{\hat{\pi}}_{\bm{\hat{\theta}}_{m}}}_{i}}(\bm{s},\bm{a};\bm{\hat{\mu}}^{i}_{m+1})\notag\\
&\times\sum_{j\in\mathcal{N}^{E,\kappa_{p}}_{i}}\nabla_{\theta_{i}}\log\pi_{j}(a_{j}|s_{\mathcal{N}^{E,\kappa_{p}}_{j}},\hat{\theta}^{j}_{j,m},\hat{\theta}^{j}_{\mathcal{N}^{E,\kappa_{p}}_{j,-j},t})\Big]\label{theunbiasedestimate5}\\
=&g^{\bm{\hat{\pi}}_{\bm{\hat{\theta}}_{m}},\bm{\hat{\mu}}_{m+1}}_{app,i},\notag
\end{align}
where the equality (\ref{theunbiasedestimate4}) comes from the fact that $T_{2}(k)\sim\mathrm{Geom}(1-\gamma)$ and $\mathbb{P}\big(t'=T_{2}(k)\big)=(1-\gamma)\gamma^{t'}$, (\ref{theunbiasedestimate5}) comes from the definition of the discounted state visitation distribution in (\ref{Thediscountedstatevisitationdistribution}), and the last equality achieves by the definition of $g^{\bm{\hat{\pi}}_{\bm{\hat{\theta}}_{m}},\bm{\hat{\mu}}_{m+1}}_{app,i}$ in (\ref{thepolicygradientapproxiamtionbynewfunction2}).
Hence, the proof of Case (ii) is completed.
\end{proof}

\subsection{Proof of Lemma~\ref{thelemmaofsmooth}}\label{ProofofLemmathelemmaofsmooth}
\begin{proof}
For Case (i), by the definitions of $\nabla_{\mu_{i}}\mathcal{L}(\bm{\theta},\bm{\mu})$ in (\ref{thecoupledpolicygradienttheorem2}) and $G_{i}(\bm{\theta})$ in (\ref{thelongtermreturninNMARLGG}), we have
\begin{align}
&\nabla_{\mu_{i}}\mathcal{L}(\bm{\theta}_{m},\bm{\mu}_{m})\notag\\
=&\frac{1}{N}\mathbb{E}_{\bm{s}\sim\bm{\rho}}\Big[\sum_{t=0}^{\infty}\gamma^{t}g_{i}(s_{i,t},a_{i,t})\Big|\bm{s}_{0}=\bm{s},\bm{a}_{t}\sim\bm{\pi}_{\bm{\theta}_{m}}(\cdot|\bm{s}_{t})\Big]\notag\\
=&\frac{1}{N}\sum_{t=0}^{\infty}\gamma^{t}\int_{\bm{s}_{0},\bm{a}_{0}}g_{i}(s_{i,t},a_{i,t})\rho^{\bm{\pi}_{\bm{\theta}_{m}}}_{t}(\bm{s}_{0:t},\bm{a}_{0:t})d_{\bm{s}_{0:t},\bm{a}_{0:t}},\label{thesmoothofmu1}
\end{align}
where $\rho^{\bm{\pi}_{\bm{\theta}_{m}}}_{t}(\bm{s}_{0:t},\bm{a}_{0:t})$ is the probability of the sequence $\{\bm{s},\bm{a}\}_{0:t}$ generated by the joint policy $\bm{\pi}_{\bm{\theta}_{m}}$, which is described as
\begin{align}\label{theprobabilityofsequence}
\rho^{\bm{\pi}_{\bm{\theta}_{m}}}_{t}(\bm{s}_{0:t},\bm{a}_{0:t})=&\prod_{h=0}^{t-1}\bm{\mathcal{P}}(\bm{s}_{h+1}|\bm{s}_{h},\bm{a}_{h})\prod_{h=0}^{t}\bm{\pi_{\theta}}(\bm{a}_{h}|\bm{s}_{h})\bm{\rho}(\bm{s}_{0}).
\end{align}
Similar to (\ref{theprobabilityofsequence}), define $\rho^{\bm{\hat{\pi}}_{\bm{\theta}_{m}}}_{t}(\bm{s}_{0:t},\bm{a}_{0:t})$ as the probability of the sequence $\{\bm{s},\bm{a}\}_{0:t}$ generated by the joint policy $\bm{\hat{\pi}}_{\bm{\hat{\theta}}_{m}}$.
By using (\ref{thesmoothofmu1}), we have that
\begin{align}
&\Big\|\nabla_{\mu_{i}}\mathcal{L}(\bm{\theta}_{m},\bm{\mu}_{m})-\nabla_{\mu_{i}}\tilde{\mathcal{L}}(\bm{\hat{\theta}}_{m},\bm{\hat{\mu}}_{m})\Big\|\notag\\
=&\frac{1}{N}\Big|\sum_{t=0}^{\infty}\gamma^{t}\int_{\bm{s}_{0},\bm{a}_{0}}g_{i}(s_{i,t},a_{i,t})\notag\\
&\times\big(\rho^{\bm{\pi}_{\bm{\theta}_{m}}}_{t}(\bm{s}_{0:t},\bm{a}_{0:t})-\rho^{\bm{\hat{\pi}}_{\bm{\theta}_{m}}}_{t}(\bm{s}_{0:t},\bm{a}_{0:t})\big)d_{\bm{s}_{0:t},\bm{a}_{0:t}}\Big|\notag\\
\leq&\frac{1}{N}\sum_{t=0}^{\infty}\gamma^{t}\int_{\bm{s}_{0},\bm{a}_{0}}\big|g_{i}(s_{i,t},a_{i,t})\big|\notag\\
&\times\big|\rho^{\bm{\pi_{\theta}}}_{t}(\bm{s}_{0:t},\bm{a}_{0:t})-\rho^{\bm{\hat{\pi}}_{\bm{\hat{\theta}}_{m}}}_{t}(\bm{s}_{0:t},\bm{a}_{0:t})\big|d_{\bm{s}_{0:t},\bm{a}_{0:t}},\label{thesmoothofmu2}
\end{align}
where the equality uses the fact that Lagrangian gradient $\nabla_{\mu_{i}}\mathcal{L}(\bm{\theta},\bm{\mu})$ is real value and the last inequality can be achieved by the absolute value inequality, i.e., $|x+y|\leq|x|+|y|$ for all $x,y\in\mathbb{R}$.
Based on the definition of $\rho^{\bm{\pi_{\theta}}}_{t}(\bm{s}_{0:t},\bm{a}_{0:t})$ in (\ref{theprobabilityofsequence}),
the term
$\big|\rho^{\bm{\pi}_{\bm{\theta}_{m}}}_{t}(\bm{s}_{0:t},\bm{a}_{0:t})-\rho^{\bm{\hat{\pi}}_{\bm{\hat{\theta}}_{m}}}_{t}(\bm{s}_{0:t},\bm{a}_{0:t})\big|$ in the right-hand of (\ref{thesmoothofmu2}) can be further written as
\begin{align}
&\big|\rho^{\bm{\pi}_{\bm{\theta}_{m}}}_{t}(\bm{s}_{0:t},\bm{a}_{0:t})-\rho^{\bm{\hat{\pi}}_{\bm{\hat{\theta}}_{m}}}_{t}(\bm{s}_{0:t},\bm{a}_{0:t})\big|\notag\\
=&\Big|\bm{\rho}(\bm{s}_{0})\prod_{h=0}^{t-1}\bm{\mathcal{P}}(\bm{s}_{h+1}|\bm{s}_{h},\bm{a}_{h})\notag\\
&\times\Big(\prod_{h=0}^{t}\bm{\pi}_{\bm{\theta}_{m}}(\bm{a}_{h}|\bm{s}_{h})-\prod_{h=0}^{t}\bm{\hat{\pi}}_{\bm{\hat{\theta}}_{m}}(\bm{a}_{h}|\bm{s}_{h})\Big)\Big|\notag\\
=&\bm{\rho}(\bm{s}_{0})\prod_{h=0}^{t-1}\bm{\mathcal{P}}(\bm{s}_{h+1}|\bm{s}_{h},\bm{a}_{h})\notag\\
&\times\Big|\prod_{h=0}^{t}\bm{\hat{\pi}}_{\mathbf{1}_{N}\otimes\bm{\theta}_{m}}(\bm{a}_{h}|\bm{s}_{h})-\prod_{h=0}^{t}\bm{\hat{\pi}}_{\bm{\hat{\theta}}_{m}}(\bm{a}_{h}|\bm{s}_{h})\Big|\label{thesmoothpartmu2-2add}\\
=&\bm{\rho}(\bm{s}_{0})\prod_{h=0}^{t-1}\bm{\mathcal{P}}(\bm{s}_{h+1}|\bm{s}_{h},\bm{a}_{h})\Big|(\mathbf{1}_{N}\otimes\bm{\theta}_{m}-\bm{\hat{\theta}}_{m})^{\top}\notag\\
&\Big(\sum_{h'=0}^{t}\nabla_{\bm{\hat{\theta}}}\bm{\hat{\pi}_{\check{\theta}}}(\bm{a}_{h'}|\bm{s}_{h'})\prod_{\substack{h=0\\h\neq h'}}^{t}\bm{\hat{\pi}_{\check{\theta}}}(\bm{a}_{h}|\bm{s}_{h})\Big)\Big|\label{thesmoothpartmu2-2}\\
=&\bm{\rho}(\bm{s}_{0})\prod_{h=0}^{t-1}\bm{\mathcal{P}}(\bm{s}_{h+1}|\bm{s}_{h},\bm{a}_{h})\Big|(\mathbf{1}_{N}\otimes\bm{\theta}_{m}-\bm{\hat{\theta}}_{m})^{\top}\notag\\
&\Big(\sum_{h'=0}^{t}\nabla_{\bm{\hat{\theta}}}\log\bm{\hat{\pi}_{\check{\theta}}}(\bm{a}_{h'}|\bm{s}_{h'})\prod_{h=0}^{t}\bm{\hat{\pi}_{\check{\theta}}}(\bm{a}_{h}|\bm{s}_{h})\Big)\Big|\notag\\
\leq&\|\mathbf{1}_{N}\otimes\bm{\theta}_{m}-\bm{\hat{\theta}}_{m}\|\sum_{h'=0}^{t}\|\nabla_{\bm{\hat{\theta}}}\log\bm{\hat{\pi}_{\check{\theta}}}(\bm{a}_{h'}|\bm{s}_{h'})\|\label{thesmoothpartmu2-4}\\
\leq&(t+1)BN\|\mathbf{1}_{N}\otimes\bm{\theta}_{m}-\bm{\hat{\theta}}_{m}\|,\label{thesmoothpartmu2-5}
\end{align}
where the equality (\ref{thesmoothpartmu2-2add}) uses the fact that $\bm{\hat{\pi}}_{\mathbf{1}_{N}\otimes\bm{\theta}_{m}}=\bm{\hat{\pi}}_{\bm{\hat{\theta}}_{m}}$, equality (\ref{thesmoothpartmu2-2}) comes from the mean value theorem and $\bm{\check{\theta}}=\xi(\mathbf{1}_{N}\otimes\bm{\theta}_{m})+(1-\xi)\bm{\hat{\theta}}_{m}$ with $\xi\in[0,1]$, the inequality (\ref{thesmoothpartmu2-4}) uses the fact that $|x^{\top}y|\leq\|x\|\|y\|$ for all $x,y\in\mathbb{R}^{\sum_{i=1}^{N}d_{\theta_{i}}}$, and
the last inequality can be obtained by Assumption~\ref{theassumptionofpolicy} and
\begin{align}
&\|\nabla_{\bm{\hat{\theta}}}\log\bm{\hat{\pi}_{\check{\theta}}}(\bm{a}_{h'}|\bm{s}_{h'})\|\notag\\
=&\begin{Vmatrix}
\nabla_{\hat{\theta}^{1}_{1}}\log\pi_{1}(a_{1,h'}|s_{\mathcal{N}^{E,\kappa_{p}}_{1},h'},\check{\theta}^{1}_{1},\check{\theta}^{1}_{\mathcal{N}^{E,\kappa_{p}}_{1,-1}})\\
\vdots\\
\nabla_{\hat{\theta}^{1}_{N}}\log\pi_{1}(a_{1,h'}|s_{\mathcal{N}^{E,\kappa_{p}}_{1},h'},\check{\theta}^{1}_{1},\check{\theta}^{1}_{\mathcal{N}^{E,\kappa_{p}}_{1,-1}})\\
\vdots\\
\nabla_{\hat{\theta}^{N}_{N}}\log\pi_{N}(a_{N,h'}|s_{\mathcal{N}^{E,\kappa_{p}}_{N},h'},\check{\theta}^{N}_{N},\check{\theta}^{N}_{\mathcal{N}^{E,\kappa_{p}}_{N,-N}})
\end{Vmatrix}\leq BN.\notag
\end{align}
\par
Substituting (\ref{thesmoothpartmu2-5}) into (\ref{thesmoothofmu2}), we have
\begin{align}
&\Big\|\nabla_{\mu_{i}}\mathcal{L}(\bm{\theta}_{m},\bm{\mu}_{m})-\nabla_{\mu_{i}}\tilde{\mathcal{L}}(\bm{\hat{\theta}}_{m},\bm{\hat{\mu}}_{m})\Big\|\notag\\
\leq&\frac{1}{N}\sum_{t=0}^{\infty}\gamma^{t}(t+1)BNR_{g}\|\mathbf{1}_{N}\otimes\bm{\theta}_{m}-\bm{\hat{\theta}}_{m}\|\notag\\
\leq&\frac{BR_{g}}{(1-\gamma)^{2}}\|\mathbf{1}_{N}\otimes\bm{\theta}_{m}-\bm{\hat{\theta}}_{m}\|,\notag
\end{align}
which completes the Case (i).
\par
For Case (ii), by the definitions of $\nabla_{\theta_{i}}\mathcal{L}(\bm{\theta},\bm{\mu})$ in (\ref{thecoupledpolicygradienttheorem}) and $\rho^{\bm{\pi}_{\bm{\theta}_{m}}}_{t}(\bm{s}_{0:t},\bm{a}_{0:t})$ in (\ref{theprobabilityofsequence}), we have
\begin{align}
&\nabla_{\theta_{i}}\mathcal{L}(\bm{\theta}_{m},\bm{\mu}_{m+1})\notag\\
=&\frac{1}{1-\gamma}\mathbb{E}_{\bm{s}\sim d^{\bm{\pi}_{\bm{\theta}_{m}}}_{\bm{\rho}},\bm{a}\sim\bm{\pi}_{\bm{\theta}_{m}}}\Big[Q^{\bm{\pi}_{\bm{\theta}_{m}}}(\bm{s},\bm{a};\bm{\mu}_{m+1})\notag\\
&\times\sum_{j\in\mathcal{N}^{E,\kappa_{p}}_{i}}\nabla_{\theta_{i}}\log\pi_{j}(a_{j}|s_{\mathcal{N}^{E,\kappa_{p}}_{j}},\theta_{j,m},\theta_{\mathcal{N}^{E,\kappa_{p}}_{j,-j},m})\Big]\notag\\
=&\sum_{j\in\mathcal{N}^{E,\kappa_{p}}_{i}}\sum_{t'=0}^{\infty}\sum_{t=0}^{\infty}\gamma^{t'+t}\int_{\bm{s}_{t'},\bm{a}_{t'}}\Big(\frac{1}{N}\notag\\
&\times\sum_{k\in\mathcal{N}}\big(r_{k}(s_{k,t'+t},a_{k,t'+t})+\mu_{k,m+1}g_{k}(s_{k,t'+t},a_{k,t'+t})\big)\Big)\notag\\
&\times\nabla_{\theta_{i}}\log\pi_{j}(a_{j,t'}|s_{\mathcal{N}^{E,\kappa_{p}}_{j},t'},\theta_{j,m},\theta_{\mathcal{N}^{E,\kappa_{p}}_{j,-j},m})\notag\\
&\times\rho^{\bm{\pi}_{\bm{\theta}_{m}}}_{t'+t}(\bm{s}_{0:t'+t},\bm{a}_{0:t'+t})d_{\bm{s}_{0:t'+t}}d_{\bm{a}_{0:t'+t}}.\label{thepolicygradientsumformal}
\end{align}
Based on the definition of $\nabla_{\theta_{i}}\tilde{\mathcal{L}}(\bm{\hat{\theta}}_{m},\bm{\hat{\mu}}_{m+1})$ in (\ref{thepolicygradientofexecutionpolicy}), we have
\begin{align}
&\|\nabla_{\theta_{i}}\mathcal{L}(\bm{\theta}_{m},\bm{\mu}_{m+1})-\nabla_{\theta_{i}}\tilde{\mathcal{L}}(\bm{\hat{\theta}}_{m},\bm{\hat{\mu}}_{m+1})\|\notag\\
=&\Bigg\|\sum_{j\in\mathcal{N}^{E,\kappa_{p}}_{i}}\sum_{t'=0}^{\infty}\sum_{t=0}^{\infty}\gamma^{t'+t}\Bigg(\int_{\bm{s}_{t'},\bm{a}_{t'}}\Big(\frac{1}{N}\notag\\
&\times\sum_{k\in\mathcal{N}}\big(\mu_{k,m+1}g_{k}(s_{k,t'+t},a_{k,t'+t})\notag\\
&-\hat{\mu}^{i}_{k,m+1}g_{k}(s_{k,t'+t},a_{k,t'+t})\big)\Big)\notag\\
&\times\nabla_{\theta_{i}}\log\pi_{j}(a_{j,t'}|s_{\mathcal{N}^{E,\kappa_{p}}_{j},t'},\theta_{j,m},\theta_{\mathcal{N}^{E,\kappa_{p}}_{j,-j},m})\notag\\
&\times\rho^{\bm{\pi}_{\bm{\theta}_{m}}}_{t'+t}(\bm{s}_{0:t'+t},\bm{a}_{0:t'+t})d_{\bm{s}_{0:t'+t}}d_{\bm{a}_{0:t'+t}}\Bigg)\notag\\
&+\gamma^{t'+t}\Bigg(\int_{\bm{s}_{t'},\bm{a}_{t'}}\Big(\frac{1}{N}\sum_{k\in\mathcal{N}}\big(r_{k}(s_{k,t'+t},a_{k,t'+t})\notag\\
&+\hat{\mu}^{i}_{k,m+1}g_{k}(s_{k,t'+t},a_{k,t'+t})\big)\Big)\notag\\
&\times\Big(\nabla_{\theta_{i}}\log\pi_{j}(a_{j,t'}|s_{\mathcal{N}^{E,\kappa_{p}}_{j},t'},\theta_{j,m},\theta_{\mathcal{N}^{E,\kappa_{p}}_{j,-j},m})\notag\\
&-\nabla_{\theta_{i}}\log\pi_{j}(a_{j,t'}|s_{\mathcal{N}^{E,\kappa_{p}}_{j},t'},\hat{\theta}^{i}_{j,m},\hat{\theta}^{i}_{\mathcal{N}^{E,\kappa_{p}}_{j,-j},m})\Big)\notag\\
&\times\rho^{\bm{\pi}_{\bm{\theta}_{m}}}_{t'+t}(\bm{s}_{0:t'+t},\bm{a}_{0:t'+t})d_{\bm{s}_{0:t'+t}}d_{\bm{a}_{0:t'+t}}\Bigg)\notag\\
&+\gamma^{t'+t}\Bigg(\int_{\bm{s}_{t'},\bm{a}_{t'}}\Big(\frac{1}{N}\sum_{k\in\mathcal{N}}\big(r_{k}(s_{k,t'+t},a_{k,t'+t})\notag\\
&+\hat{\mu}^{i}_{k,m+1}g_{k}(s_{k,t'+t},a_{k,t'+t})\big)\Big)\notag\\
&\times\nabla_{\theta_{i}}\log\pi_{j}(a_{j,t'}|s_{j,t'},\hat{\theta}^{i}_{j,m},\hat{\theta}^{i}_{\mathcal{N}^{E,\kappa_{p}}_{j,-j},m})\notag\\
&\times\big(\rho^{\bm{\pi}_{\bm{\theta}_{m}}}_{t'+t}(\bm{s}_{0:t'+t},\bm{a}_{0:t'+t})\notag\\
&-\rho^{\bm{\hat{\pi}}_{\bm{\hat{\theta}_{t}}}}_{t'+t}(\bm{s}_{0:t'+t},\bm{a}_{0:t'+t})\big)d_{\bm{s}_{0:t'+t}}d_{\bm{a}_{0:t'+t}}\Bigg)\Bigg\|\notag\\
\leq&\sum_{j\in\mathcal{N}^{E,\kappa_{p}}_{i}}\sum_{t'=0}^{\infty}\sum_{t=0}^{\infty}\gamma^{t'+t}\Big(\frac{BR_{g}}{N}\|\bm{\mu}_{m+1}-\bm{\hat{\mu}}^{i}_{m+1}\|_{1}\notag\\
&+L(R_{f}+\tilde{\mu}_{max}R_{g})\|\bm{\theta}_{m}-\bm{\hat{\theta}}^{i}_{m}\|\notag\\
&+(t'+t+1)B^{2}N(R_{f}+\tilde{\mu}_{max}R_{g})\|\mathbf{1}_{N}\otimes\bm{\theta}_{m}-\bm{\hat{\theta}}_{m}\|\Big)\label{thekeyinsmooth}\\
\leq&\frac{B\sqrt{N}R_{g}}{(1-\gamma)^{2}}\|\bm{\mu}_{m+1}-\bm{\hat{\mu}}^{i}_{m+1}\|+\Big(\frac{LN(R_{f}+\tilde{\mu}_{max}R_{g})}{(1-\gamma)^{2}}\notag\\
&+\frac{(1+\gamma)B^{2}N^{2}(R_{f}+\tilde{\mu}_{max}R_{g})}{(1-\gamma)^{3}}\Big)\|\mathbf{1}_{N}\otimes\bm{\theta}_{m}-\bm{\hat{\theta}}_{m}\|,\label{thekeyinsmooth2}
\end{align}
where the inequality (\ref{thekeyinsmooth}) can be obtained by Assumptions~\ref{theassumptionofreward} and~\ref{theassumptionofpolicy} and (\ref{thekeyinsmooth2}) uses the fact that $\|x\|_{1}\leq\sqrt{N}\|x\|$ for any vector $x\in\mathbb{R}^{N}$.
Hence, we complete the proof of this lemma.
\end{proof}

\subsection{Proof of Corollary~\ref{thecorollaryerrorbetweentrueandused}}\label{ProofofCorollarythecorollaryerrorbetweentrueandused}
\begin{proof}
For Case (i), we have that
\begin{align}
&\big\|\nabla_{\mu_{i}}\mathcal{L}(\bm{\theta}_{m},\bm{\mu}_{m})-\bar{h}^{\bm{\hat{\pi}}_{\bm{\hat{\theta}}_{m}},\bm{\hat{\mu}}_{m}}_{app,i}\big\|\notag\\
\leq&\big\|\nabla_{\mu_{i}}\mathcal{L}(\bm{\theta}_{m},\bm{\mu}_{m})-\nabla_{\mu_{i}}\tilde{\mathcal{L}}(\bm{\theta}_{m},\bm{\mu}_{m})\big\|\notag\\
&+\big\|\nabla_{\mu_{i}}\tilde{\mathcal{L}}(\bm{\theta}_{m},\bm{\mu}_{m})-\bar{h}^{\bm{\hat{\pi}}_{\bm{\hat{\theta}}_{m}},\bm{\hat{\mu}}_{m}}_{app,i}\big\|\notag\\
\leq&L^{\mu}_{2}\big\|\mathbf{1}_{N}\otimes\bm{\theta}_{m}-\bm{\hat{\theta}}_{m}\big\|+2L^{\mu}_{1}\sqrt{\frac{\log{(2/\delta)}}{2K_{\mu}}}\label{thecorollaryerrorbetweentrueandused2-1}\\
=&\epsilon_{\mu}(m,K_{\mu})/\sqrt{N},
\end{align}
where (\ref{thecorollaryerrorbetweentrueandused2-1}) uses Case (i) in Lemma~\ref{thelemmaofsmooth} and (\ref{theunbiasedestimationcorollary1}).
\par
For Case (ii), we have
\begin{align}
&\big\|\nabla_{\theta_{i}}\mathcal{L}(\bm{\theta}_{m},\bm{\mu}_{m+1})-\bar{g}^{\bm{\hat{\pi}}_{\bm{\hat{\theta}}_{m}},\bm{\hat{\mu}}_{m+1}}_{app,i}\big\|\notag\\
\leq&\underbrace{\big\|\nabla_{\theta_{i}}\mathcal{L}(\bm{\theta}_{m},\bm{\mu}_{m+1})-\nabla_{\theta_{i}}\tilde{\mathcal{L}}(\bm{\hat{\theta}}_{m},\bm{\hat{\mu}}_{m+1})\big\|}_{\mathrm{(i)}}\notag\\
&+\underbrace{\big\|\nabla_{\theta_{i}}\tilde{\mathcal{L}}(\bm{\hat{\theta}}_{m},\bm{\hat{\mu}}_{m+1})-g^{\bm{\hat{\pi}}_{\bm{\hat{\theta}}_{m}},\bm{\hat{\mu}}_{m+1}}_{app,i}\big\|}_{\mathrm{(ii)}}\notag\\
&+\underbrace{\big\|g^{\bm{\hat{\pi}}_{\bm{\hat{\theta}}_{m}},\bm{\hat{\mu}}_{m+1}}_{app,i}-\bar{g}^{\bm{\hat{\pi}}_{\bm{\hat{\theta}}_{m}},\bm{\hat{\mu}}_{m+1}}_{app,i}\big\|}_{\mathrm{(iii)}}\notag\\
\leq&L^{\theta\mu}_{2}\|\mathbf{1}_{N}\otimes\bm{\mu}_{m+1}-\bm{\hat{\mu}}_{m+1}\|+L^{\theta\theta}_{2}\|\mathbf{1}_{N}\otimes\bm{\theta}_{m}-\bm{\hat{\theta}}_{m}\|\notag\\
&+\frac{2(R_{f}+\tilde{\mu}_{max}R_{g})BN}{(1-\gamma)^{2}}\gamma^{h(\kappa,\kappa_{p})+1}\notag\\
&+2L^{\theta}_{1}\sqrt{\frac{\log{(2/\delta)}}{2K_{\theta}}}\label{theinequalityofcorollary1}\\
=&\big(\epsilon_{\theta}(m,K_{\theta})+\epsilon_{0}(\kappa,\kappa_{p})\big)/\sqrt{N},
\end{align}
where (\ref{theinequalityofcorollary1}) can be obtained the Case (ii) in Lemma~\ref{thelemmaofsmooth}, (\ref{theerroroftruncatedpolicygradientapproximate}), (\ref{theunbiasedestimationcorollary2}), and the fact that
$\big\|\bm{\mu}_{m}-\bm{\hat{\mu}}^{i}_{m}\big\|\leq\big\|\mathbf{1}_{N}\otimes\bm{\mu}_{m}-\bm{\hat{\mu}}_{m}\big\|$.
\end{proof}

\vspace*{-1.0cm}

\end{document}